\documentclass[11pt,a4paper]{article}

\usepackage[disable]{todonotes}

\usepackage{style}
\usepackage{xcolor}
\usepackage{shortcuts}

\title{Deterministic Distance Approximation in MPC\\ via Improved Hitting Sets\footnote{This research was funded in whole or in part by the Austrian Science Fund (FWF) \url{https://doi.org/10.55776/P36280} and \url{https://doi.org/10.55776/I6915}. For open access purposes, the author has applied a CC BY public copyright license to any author-accepted manuscript version arising from this submission. This research was supported by the Israel Science Foundation (grant No. 2829/25).}}
\author{ Kyungjin Cho, Michal Dory, Yannic Maus, Tijn de Vos}
\date{}

\begin{document}

\begin{titlepage}
    \maketitle
    \thispagestyle{empty}
   
    \begin{abstract}
        In this paper, we provide the first deterministic algorithms with sublogarithmic round complexity for spanners and approximate shortest paths in various MPC models. Moreover, we significantly improve upon the state of the art in the deterministic Congested Clique.
In particular, we obtain the following four results on undirected graphs:
\begin{enumerate}
    \item In both linear MPC and Congested Clique, we obtain an $O(k)$ stretch-spanner of a weighted graph of size $O(n^{1+1/k})$ in $O(1)$ rounds, for some parameter $k\ge 0$. For $k=O(\log{n})$, this leads to an $O(\log n)$ approximation of APSP in constant rounds in both models. 
    \item In sublinear MPC, we obtain an $O(k^{1+\varepsilon})$-stretch spanner of a weighted graph of size $O(n^{1+1/k})$ in $O(\log k)$ rounds, for any fixed constant $\eps>0$. 
     \item In Congested Clique, we obtain $O(1)$-approximate APSP for weighted graphs in $O(\log \log \log n)$ rounds.
    \item In near-linear MPC, we obtain $(1+\varepsilon)$-approximate single-source shortest paths and $O(1)$-approximate all-pairs shortest paths for unweighted graphs in $\textsf{poly}\log \log n$  rounds. Our algorithm only requires a single near-linear memory machine, where the rest can have sublinear memory. 
\end{enumerate}
 Our deterministic algorithms obtain similar guarantees to the state of the art randomized algorithms without incurring additional factors in the round complexity.
To obtain these results, we inspect the randomized algorithms and isolate a randomized sampling routine.
Then we derandomize these sampling routines by using a deterministic hitting set. 
Hereto, we develop a versatile deterministic hitting set algorithm, which we hope will have further derandomization applications.

    \end{abstract}

    \newpage
    \thispagestyle{empty}  
     \tableofcontents
    \thispagestyle{empty}    
     \newpage
      \listoftodos
  \end{titlepage}

\tijni{add authors}
\section{Introduction}



In the last decades, processing massive amounts of data has become omnipresent. Processing this data in parallel is one of the main approaches to achieve efficient algorithms. 
%
The massively parallel computation (MPC) model~\cite{karloff2010model} is a modern parallel model developed to model large-scale parallel processing settings such as MapReduce \cite{dean2008mapreduce}
, Hadoop \cite{white2012hadoop}, Spark \cite{zaharia2010spark}, and Dryad \cite{isard2007dryad}, 
that deal with massive data. In this model, the input is distributed between a set of $N$ machines each with  \emph{limited memory} $L$, that communicate with each other via all-to-all communication in synchronous rounds. 
In each round, each machine can send and receive a total of $L$ words of $O(\log(NL))$ bits and perform internal computation. 
The goal is to minimize the number of communication rounds.
A central line of research focuses on obtaining fast algorithms for {undirected graph problems in MPC}. The ultimate goal is to develop algorithms that significantly outperform those in traditional parallel computing settings. 
Note that all graphs considered in this paper are undirected.
\kyungjin{spedified that it is for undirected graph problem}

There are several variants of the MPC model that mostly differ in the  bound on the memory available per machine. In the \emph{super-linear} MPC model each machine has $L=n^{1+\gamma}$ memory for a constant $\gamma$, where $n$ is the number of vertices in the input graph, in the \emph{near-linear} MPC model each machine has $L=\tilde{O}(n)$ memory, 
in the \emph{linear} MPC model each machine has $L=O(n)$ memory\footnote{We write $\tilde O(f):=O(f\poly \log f)$ for any function $f$.},
and in the \emph{sublinear} MPC model each machine has $L=n^{\delta}$ memory for a constant $\delta < 1$. In each of these models, we ideally have $O(m)$ \emph{total} space, where $m$ is the input size. Sometimes, algorithms incur a small overhead, e.g., they have $\tilde O(m)$ or $O(m^{1+\rho})$ total space for some small constant $\rho$. 
For more details on the MPC model, see \Cref{sec:preliminaries}.

\paragraph{Algorithms for the MPC Model}\tijni{in the SPAA format, there's automatically a period after a paragraph heading, so we took them out of the intro-part. But all other parts have a period. Should make this consistent}
The MPC model has received a lot of attention in recent years. A rich line of work led to fast algorithms for various graph problems 
such as connectivity \cite{andoni2018parallel,behnezhad2019near,CoyC23,FischerGG22}, 
minimum spanning tree  \cite{lattanzi2011filtering,nowicki2021deterministic}, various coloring problems \cite{chang2019complexity,czumaj2021simple}, 
maximal matching and maximal independent set \cite{czumaj2018round,ghaffari2018improved, behnezhad2019massively,behnezhad2019exponentially},
minimum cut \cite{lattanzi2011filtering,ghaffari2020massively}, shortest paths and spanners \cite{biswas2021massively,podc2021spanner,fischer2022massively, dory2024massively}, and more. 
The main goal is to obtain very fast algorithms that ideally take sublogarithmic  or even constant number of rounds. While in the super-linear and linear memory regimes of MPC many problems indeed have constant or $\poly(\log{\log{n}})$ round algorithms, in the sublinear memory regime many important problems such as computing minimum spanning tree or shortest paths are conjectured to require $\Omega(\log{n})$ rounds.


\paragraph{Deterministic MPC algorithms}
While many of the algorithms discussed above are randomized, substantial effort has also been devoted to developing deterministic MPC algorithms. In several cases, these algorithms achieve runtimes matching the randomized state of the art—and sometimes even match known conditional lower bounds—for fundamental problems such as connectivity, minimum spanning tree, and vertex coloring~\cite{czumaj2021simple,nowicki2021deterministic,CoyC23,FischerGG22}.
However, gaps between randomized and deterministic runtimes persist for other problems, including maximal matching, maximal independent set, and ruling sets~\cite{CzumajDP21,PaiP22,FischerGG23,GilibertiP24,ji2025fast}, as well as for many distance computation problems that are the primary focus of this work.


\paragraph{Distance computation} 
Distance problems are fundamental topics in graph algorithms, including the study of \emph{Single-Source Shortest Paths (SSSP)},  \emph{All-Pairs Shortest Paths (APSP)}, and \emph{spanners}.
The approximate SSSP problem asks to compute, for a given source vertex, approximate distances to all other vertices, while the APSP problem asks to compute approximate distances between every pair of vertices. 
For a weighted graph, a \emph{$k$-spanner} is a sparse subgraph that approximately preserves pairwise distances up to a multiplicative $k$ factor. Any graph with $n$ vertices admits a $(2k-1)$-spanner of size at most $O(n^{1+1/k})$~\cite{althofer1993sparse}.
Spanners have been used for various distance problems, including 
transhipment-based distance approximation~\cite{li2020faster,becker2021near} and fast distance sketches~\cite{dinitz2020massively}.
Due to their theoretical significance and various real-world applications, these problems have been extensively studied from multiple perspectives.
In particular, there has been substantial recent progress in developing MPC algorithms.

\paragraph{Polylogarithmic-round MPC algorithms for distance computation}
In \emph{sublinear} MPC there are polylogarithmic algorithms for APSP and SSSP, including deterministic algorithms. 
Note that in the sublinear MPC model there is a conditional $\Omega(\log{n})$ lower bound for these problems based on the widely believed 1-vs-2 cycles conjecture, see, e.g.,~\cite{ghaffari2019conditional,nanongkai2022equivalence}. 

Hajiaghayi, Lattanzi, Seddighin, and Stein~\cite{hajiaghayi2019mapreduce} give a deterministic $O(\log n)$-round algorithm for APSP in MPC with $n^\delta$ memory per machine and $O(n^{3-\delta/2})$ total space.
There is also a randomized $\poly (\log{n})$-round sublinear MPC algorithm for distance sketches  by Dinitz and Nazari~\cite{dinitz2020massively}.  

These results are explicit MPC algorithms. 
In addition, many \emph{implicit} MPC results arise from simulations of classical PRAM algorithms. In particular, PRAM algorithms that use a polynomial number of processors can be simulated in the MPC model with only constant overhead in round complexity~\cite{goodrich2011sorting}.
In sublinear MPC, this leads to $\poly (\log{n})$-round algorithms for $(1+\varepsilon)$-approximate SSSP~\cite{li2020faster, andoni2020parallel,ElkinM21,rozhovn2022undirected} and $\poly\log n$-spanners~\cite{bezdrighinEG+2022}. Many of these algorithms are deterministic \cite{ElkinM21,rozhovn2022undirected,bezdrighinEG+2022}.

\begin{table*}
    \renewcommand{\arraystretch}{1.3}
    \centering\resizebox{\textwidth}{!}{
    \begin{tabular}{|c|c|c|c|c|c|c|}
    \hline
        Regime &Stretch & Size & Rounds & Total space & Randomness & \\
    \hline 
    \hline
    Congested Clique &$O(k)$ & $O(n^{1+1/k}k)$ & $O(\log k)$ & - & Deterministic &\cite{ParterY18} \\
    \hline

     Congested Clique &$(1+\varepsilon)(2k-1)$ & $O(n^{1+1/k})$ & $O(1)$ & - & Randomized &\cite{chechik2022constant} \\
    \hline

    Congested Clique & $O(k)$ & $O(n^{1+1/k})$ & $O(1)$  & - & Deterministic & This paper (\Cref{thm:linear_MPC_spanner_weighted})\\
    \hline

    Near-linear MPC &$O(k)$ & $O(n^{1+1/k})$ & $O(1)$ & $\tilde O(n+m)$ & Randomized &\cite{podc2021spanner} \\
    \hline

    Linear MPC & $O(k)$ & $O(n^{1+1/k})$ & $O(1)$  & $ O(n+m)$  & Deterministic & This paper (\Cref{thm:linear_MPC_spanner_weighted})\\
    \hline
    
    Sublinear MPC & $O(k^{1+\varepsilon})$ & $O(n^{1+1/k}\log k)$ & $O(\log k)$  & $\tilde O(m+n)$ & Randomized & \cite{biswas2021massively} \\
    \hline


    Sublinear MPC & $O(k^{1+\varepsilon})$ & $O(n^{1+1/k})$ & $O(\log k)$  & $O(m+n)$ & Deterministic & This paper (\Cref{thm:sublinear_MPC_spannerSimplified}) \\
    \hline
    \end{tabular}

    }
        
    \caption{     
    Summary of our spanner algorithms and the previous results.}
    \label{tab:summary_spanner}
\end{table*}

\paragraph{Sublogarithmic-round MPC algorithms for distance computation} 
As discussed above, in sublinear MPC there is a conditional $\Omega(\log{n})$ lower bound for approximate shortest paths, however this lower bound does not apply to spanners and there are indeed faster algorithms for constructing graph spanners. In addition, in the near-linear memory MPC there are sublogarithmic algorithms for approximate shortest paths as we discuss next.

\begin{table*}
    \renewcommand{\arraystretch}{1.3}
    \centering\resizebox{\textwidth}{!}{
    \begin{tabular}{|c|c|c|c|c|c|}
    \hline
        Regime & Approximate &  Rounds & Randomness & Weighted &\\
    \hline 
    \hline
    Congested Clique &$O(1)$-approximate & $O(\log\log\log n)$  & Randomized & Weighted & \cite{BuiCCDL24} \\
    \hline

     Congested Clique &$(3+\varepsilon)$-approximate & $O(\log^2 n/\varepsilon^2)$ &  Deterministic & Weighted &\cite{DBLP:conf/podc/Censor-HillelDK19} \\
    \hline

    Congested Clique &$O(1)$-approximate & $O(\log\log\log n)$  & Deterministic & Weighted & This paper (\Cref{thm:apsp_in_cc}) \\
    \hline

    Near-linear MPC &$O(\log n)$-approximate & $O(1)$  & Randomized & Weighted & \cite{podc2021spanner,fischer2022massively} \\
    \hline

    Linear MPC &$O(\log n)$-approximate & $O(1)$  & Deterministic & Weighted & This paper (\Cref{cor:APSPlogn}) \\
    \hline

    Near-linear MPC &$O(1)$-approximate & $\textsf{poly}(\log\log n)$  & Randomized & Unweighted &\cite{dory2024massively} \\
    \hline

    Near-linear MPC &$O(1)$-approximate & $\textsf{poly}(\log\log n)$  & Deterministic & Unweighted &This paper (\Cref{thm:APSP_sublinearMPCSimplified}) \\
    \hline
    \end{tabular}
    }
        
    \caption{     
    Summary of our APSP algorithms and the previous sublogarithmic results.}
    \label{tab:summary_distances}
\end{table*}
For spanners, 
\cite{biswas2021massively} provide a randomized algorithm that computes an $O(k^{1+\eps})$-spanner of size $O(n^{1+1/k})$ in $O(\log k)$ rounds. 
In the near-linear MPC model, there are faster randomized $O(1)$-round algorithms to compute an $O(k)$-spanner with $O(n^{1+1/k})$ edges~\cite{podc2021spanner, fischer2022massively}. 
Moreover, there are $O(1)$-round $O(\log{n})$-approximation algorithms for weighted APSP in the near-linear MPC model, based on the construction of spanners with stretch $O(\log n)$~\cite{biswas2021massively,podc2021spanner,fischer2022massively}. 

In \emph{unweighted} graphs there are randomized $\poly (\log{\log{n}})$-round\footnote{In this introduction, we consider $\eps$ to be a constant.}\tijn{new (reviewer C)} algorithms for
$(1+\eps)$-approximate SSSP and $O(1)$-approximate APSP using one near-linear memory machine, where the rest of machines can have sublinear memory
\cite{dory2024massively}.



All the above mentioned algorithms are randomized, which leads to the following open question:

\begin{question}\label{question:1}
    Can we obtain deterministic algorithms with sublogarithmic round complexity for distance computation in MPC?
\end{question}

\subsection{Contributions on Distance Problems}\label{sec:our resutls}
In this paper, we answer \Cref{question:1} in the affirmative: we give the first sublogarithmic-round deterministic MPC algorithms for spanners and approximate shortest paths. Moreover, our algorithms have optimal or close-to-optimal total space.
Since linear MPC algorithms transfer to the so-called Congested Clique, we also obtain novel results there -- sometimes improving significantly on the state of the art.  

In the Congested Clique model~\cite{lotker2005minimum}, we have an $n$-node communication network and an input graph $G=(V,E)$ on the same nodes, where each node initially knows its incident edges. Computation proceeds in synchronous rounds: in each round, every node can send an $O(\log n)$-bit message to every other node, and perform unlimited local computation between rounds. The objective is to solve a problem on $G$ using as few rounds as possible.

See \Cref{sec:related_work} for existing work in the Congested Clique model. We remark that while linear MPC algorithms can be simulated in Congested Clique \cite{BehnezhadDH18}, most existing Congested Clique algorithms do not directly imply MPC algorithms as the memory requirements do not fit the MPC model.

Additionally, we also obtain the first deterministic sublogarithmic round algorithms for computing spanners in the sublinear MPC model.

\subsubsection{Spanners}\label{sec:our_results_linearMPC}

\paragraph{Linear MPC and Congested Clique (\Cref{sec:unweighted_spanner_linearMPC,sec:weighted_spanner_linearMPC})} 
Our first result is an $O(k)$-stretch spanner of size $O(n^{1+1/k})$ in linear MPC and Congested Clique.

\begin{restatable}{theorem}{CorLinearMPCSpannerWeighted}\label{thm:linear_MPC_spanner_weighted}
    There exists a deterministic algorithm that, given a positive integer $k$ and a weighted graph $G=(V,E,w)$ with $w\colon E \to [\poly n]$ on $n$ vertices, computes an $O(k)$-spanner with $O(n^{1+1/k})$ edges in constant rounds in the linear MPC model using $O(m+n)$ total space. 
    Furthermore, it can be simulated in the Congested Clique.
\end{restatable}

This is the first deterministic sublogarithmic-round spanner construction in the MPC model. 
It matches the randomized state-of-the-art runtime by Dory, Fischer, Khoury, and Leitersdorf~\cite{podc2021spanner} while slightly improving both the size of the spanner and the total memory size by a $\log n$ factor. 


In Congested Clique, this improves upon the deterministic state of the art by Parter and Yogev~\cite{ParterY18}, who compute an $O(k)$-stretch spanner of size $O(k\cdot n^{1+1/k})$ in $O(\log k)$ rounds. { 
Chechik and Zhang~\cite{chechik2022constant} provide faster randomized spanners with stretch $(1+\eps)(2k-1)$ and size $O(n^{1+1/k})$ in constant rounds. We match their results up to constant factors. } Our spanner construction works in linear MPC, where \cite{chechik2022constant} does not translate to MPC because of its total memory requirements. 

By setting $k=\log n$, \Cref{thm:linear_MPC_spanner_weighted} gives an $O(\log n)$-spanner of size $O(n)$, that fits in one machine not only in the Congested Clique but also in the linear MPC model with $O(n)$ local space. 
Furthermore, if we have a machine storing such a spanner, we can locally compute an $O(\log n)$-approximate APSP.
Therefore, we can obtain a constant round algorithm for $O(\log n)$-approximate APSP in either model. See \Cref{cor:APSPlogn}  in \Cref{sec:apsp_corollary} for details.

\paragraph{Sublinear MPC (\Cref{sec:weighted_spanner_sublinear})}
Also in the sublinear MPC model, we give the first sublogarithmic-round deterministic spanner algorithm.

\begin{restatable}{theorem}{ThmSublinearMPCSpannerSimplified}[Simplified version of \Cref{thm:sublinear_MPC_spanner}]\label{thm:sublinear_MPC_spannerSimplified}
    For constants $\delta,\eps<1$, given a weighted graph $G$ on $n$ vertices and positive parameter $k\geq 1$, we can deterministically construct an $O(k^{1+\varepsilon})$-spanner of $G$ with $O(n^{1+1/k})$ edges in $O(\log k)$ rounds in the sublinear MPC model with $O(n^{\delta})$ local space and $O(m+n)$ total space.
\end{restatable}
This improves exponentially on the current state-of-the art deterministic algorithm~\cite{bezdrighinEG+2022} and it nearly matches the randomized state of the art by Biswas, Dory, Ghaffari, Mitrovi{\'c}, and Nazari~\cite{biswas2021massively} {up to a constant factor in the stretch, size, and the round complexity}. \kyungjini{Since the hitting set rounds have been increased for large $N$ (number of sets) compared to $U$ (set of elements), the round complexity of \Cref{thm:sublinear_MPC_spanner} also increased by $t$. Fortunately, the simplified complexity here(\Cref{thm:sublinear_MPC_spannerSimplified}) is still hold. @we need to modify the explanation here.}
Moreover, we improve the total space from $\tilde O(m+n)$  to $O(m+n)$.


The round complexity of \Cref{thm:sublinear_MPC_spannerSimplified} is nearly optimal, assuming the aforementioned 1-vs-2 cycle conjecture. This follows from an $\Omega(k)$ lower bound in the distributed LOCAL model \cite{derbel2008locality} for the closely related problem of constructing spanners with optimal parameters, namely $(2k-1)$-spanners with $O(n^{1+1/k})$ edges, which implies a conditional lower bound of $\Omega(\log k)$ in the sublinear MPC model \cite{ghaffari2019conditional}. 



\subsubsection{Shortest Paths}

As already mentioned above, we can obtain $O(\log{n})$-approximate APSP algorithms through the computation of spanners. Next, we present our results that compute $O(1)$-approximation for APSP \emph{directly}. 

\paragraph{Congested Clique (\Cref{sec:applications_constant_APSP_CC})}
In Congested Clique, we provide an algorithm for constant approximate APSP. 

\begin{restatable}{theorem}{ThmAPSPinCC}\label{thm:apsp_in_cc}
    There exists a deterministic algorithm that, given a weighted graph $G=(V,E,w)$ with $w\colon E \to [\poly n]$ on $n$ vertices, computes $O(1)$-approximate APSP in $O(\log\log \log n)$ rounds in Congested Clique.  
\end{restatable}
We obtain this result by derandomizing the algorithm of Bui, Chandra, Chang, Dory, and Leitersdorf~\cite{BuiCCDL24} in the same number of rounds and only incurring another constant factor in the stretch. 

This is exponentially faster than the previous fastest deterministic constant approximation for APSP: in unweighted graphs, Dory and Parter~\cite{DBLP:conf/podc/DoryP20} compute $(2+\eps)$-APSP in $\poly(\log \log n)$ rounds. In weighted graphs, we improve double-exponentially over the state of the art: Censor-Hillel, Dory, Korhonen, and Leitersdorf~\cite{DBLP:conf/podc/Censor-HillelDK19} compute $(3+\eps)$-APSP in $O(\log^2 n/\eps^2)$ rounds.

\paragraph{Near-Linear MPC (\Cref{sec:applications_heterogeneous})}
We also study approximate shortest paths in the near-linear MPC model. In fact, our results here are even stronger and use only a few near-linear space machines, while additional machines can have sublinear local space.
This setting is also known as the \emph{heterogeneous} MPC model, introduced by ~\cite{FischerHO25}.
{
In particular, our approximate SSSP and APSP algorithms require only one machine with near-linear space, while the approximate Multi-Source Shortest Paths (MSSP) algorithm requires a near-linear space machine for each source.}


We give the first deterministic algorithms for approximate SSSP, MSSP, and APSP {for unweighted graphs}. It matches the randomized state of the art by Dory and Matar~\cite{dory2024massively}.
In particular, we derive an MSSP algorithm: for a given set of sources $S\subseteq V$, it outputs approximate distances between every vertices $s\in S$ and $v\in V$ by using $|S|$ machines with near-linear local space and additional sublinear memory machines.

\begin{restatable}{theorem}{MSSPinSublinearMPCSimplified}[Simplified version of \Cref{thm:MSSP_sublinearMPC}]\label{thm:MSSP_sublinearMPC_simplified}
Let $G=(V,E)$ be an unweighted graph on $n$ vertices and $m$ edges, and let $\delta>0$ and $\varepsilon,\rho<1/2$ be constants.
For a fixed source set $S\subseteq V$ of size $O(n^{\rho})$,
there is a deterministic algorithm that computes all $(1+\varepsilon)$ approximate shortest paths for every pairs $(s,v)\in S\times V$
in $ \poly \log\log n$ rounds
in the MPC model using $|S|$ machines with $\tilde O(n)$ space and additional machines with
$O(n^{\delta})$ local space and $\tilde O((m+n^{1+\rho})n^{\rho})$ total space.
\end{restatable}

Note that \Cref{thm:MSSP_sublinearMPC_simplified} implies an approximate SSSP algorithm using a single near-linear space machine.
We also give a deterministic algorithm to construct a data structure supporting constant-round queries for approximate APSP by using a single machine with near-linear space.

\begin{restatable}{theorem}{APSPinSublinearMPCSimplified}[Simplified version of \Cref{thm:APSP_sublinearMPC}]\label{thm:APSP_sublinearMPCSimplified}
Let $G=(V,E)$ be an unweighted graph on $n$ vertices and $m$ edges, and let $\delta>0$ and $\varepsilon,\rho<1/2$, and $2\leq k\leq 1/\rho$ be constants.
There is a deterministic algorithm that computes a distance oracle of size $\tilde O(kn^{1+1/k})$ 
in $ \poly \log\log n$ rounds
in the MPC model using a single machine with $\tilde O(n)$ space and additional machines with
$O(n^{\delta})$ local space and $\tilde O((m+n^{1+\rho})n^{1/k})$ total space.
Our oracle supports $O(1)$-round queries for $(1+\varepsilon)(2k-1)$-approximate distances between any two vertices in $G$. 
\end{restatable}

We can further improve the total space in the above theorems to $O(m+n^{1+c})$ for a small constant $c>0$ at the cost of increasing the stretch by a constant factor, as discussed in \cite{dory2024massively}. In particular, if $m = \Omega(n^{1+c})$, the total space reduces to $O(m)$. To get this, we can first use our spanner algorithms and then run the approximate shortest paths algorithms on the spanner. Such an approach will increase the approximation factor to $\textsf{poly}(1/c)$ while improving the total memory to $O(m+n^{1+c})$. 
See \Cref{sec:applications_heterogeneous} for details. 
\subsection{A Glimpse at our Techniques}\label{sec:intro:technique}
For our distance computation results, we build on and adapt the algorithmic techniques underlying the randomized state-of-the-art approaches~\cite{podc2021spanner,biswas2021massively,BuiCCDL24,dory2024massively}. A key step in our approach is to isolate the role of randomness in these algorithms and identify a single structural component responsible for it,  which in all cases can be modeled as an instance of the \emph{hitting set} problem. For some problems, making this structure explicit (and efficiently solvable with deterministic algorithms) requires modifications to the respective original algorithm.

{
}

In the hitting set problem, we are given a finite universe of elements $U$, and $N$ subsets $S_1, \dots, S_{\HS} \subseteq U$. A \emph{hitting set} $D\subseteq U$ is a set such that $S_i\cap D \neq \emptyset$ for all $i\in [\HS]$. Trivially, $U$ is a hitting set itself. The aim is to compute a hitting set of small size.
When we are additionally guaranteed that each set $S_i$ contains at least $d$ elements, then we know that a simple greedy algorithm~\cite{LOVASZ1975383} gives a hitting set of size $O\left(\tfrac{|U|\log N}{d}\right)$. However, such an algorithm is highly sequential and not suitable for distributed models. A simple random sampling approach also gives a hitting set of size $O\left(\tfrac{|U|\log N}{d}\right)$: add each element $u\in U$ independently to $D$ with probability $p=\tfrac{\log N}{d}$.
We provide efficient deterministic algorithms to solve the hitting set problem, for a wide range of parameter regimes. In particular, we compute a hitting set of the aforementioned size in $O(\log \log d)$ rounds. 
We also give a \emph{constant round} algorithm that computes a slightly larger hitting set of size $O(\tfrac{|U|}{d^{0.99}})$.  
{This suboptimal size is still a sparse enough deterministic primitive for many applications. In particular, we use it to obtain many of the results in \Cref{sec:our resutls}. }


\subsection{Contributions on Computing Hitting Sets}\label{sec:overview_subsubsection_HS}
When considering the hitting set problem in the MPC model, it is not immediately clear what the `linear' and `sublinear' memory regimes are. 
In this paper, we consider it linear MPC when the local space $L$ is at least $d$, such that each set $S_i$ fits in a machine. To the best of our knowledge, previous work always set $L=|U|$ (which implies $L\ge d$). We choose the more flexible $L\ge d$ such that the algorithm is more versatile, as is necessary for some of our applications. 
In this regime, we obtain the following theorem. 


\begin{restatable}{theorem}{ThmLinearHittingSet}\label{thm:MPC hitting set intro: linear}
       Let $L$ be the parameter denoting the local space. There is a deterministic MPC algorithm that, given a universe $U$, a number $d\leq \min\{|U|,L\}$ and a collection of $N$ subsets $S_1,\ldots, S_N\subseteq U$ with $|S_i|\geq d$ for each $i\in[N]$, that computes a hitting set:
    \begin{enumerate}
        \item of size $O\left(\frac{|U|}{d^{0.99}}\right)$ in $O(1)$ rounds; or
        \item of size $O\left(\frac{|U|\log d}{d}\right)$ in $O(\log\log d)$ rounds.
    \end{enumerate}
Both algorithms assume $N\le|U|=\poly L$ and use linear total space: $O\left(L+|U|+\sum_{i=1}^N|S_i|\right)$.
\end{restatable}

For most applications, $N\le |U|=\poly L$ is a reasonable assumption. 
In \Cref{sec:hitting sets} we provide additional algorithms without this assumption. 

\medskip

We also study the case that $L<d$, which we call sublinear as single sets do not even fit into a machine. As is customary, we ask that $L=\Theta(d^\delta)$ for some constant $\delta$. 

\begin{restatable}{theorem}{ThmSublinearHittingSet}\label{thm:MPC hitting set intro: sublinear}
    Let $\delta\leq 1$ be a positive constant. There is a deterministic MPC algorithm that, given a universe $U$, a number $d\leq |U|$ and a collection of $N$ subsets $S_1,\ldots, S_N\subseteq U$ with $|S_i|\geq d$ for each $i\in[N]$, that computes a hitting set:
    \begin{enumerate}
        \item of size $O\left(\frac{|U|\log d}{d^{0.99}}\right)$ in $O(1)$ rounds; or
        \item of size $O\left(\frac{|U|\log d}{d}\right)$ in $O(\log\log d)$ rounds with $O(d^{\delta})$ local space.     
    \end{enumerate}
Both algorithms assume $|U|,N=\poly d$, use $O(d^{\delta})$ local space, and use linear total space: $O(|U|+\sum_{i=1}^N|S_i|)$.
\end{restatable}

An important point of this  ``linear'' and ``sublinear''  hitting set distinction, is that does not necessarily correspond to the applications in linear and sublinear MPC graph problems. To be precise, we use the ``linear hitting set'' in some ``sublinear MPC''  graph algorithms.
In such applications, we have $\Theta(n^\gamma)$ local space. If $d< n^\gamma$ in our hitting set instance, we can use the ``linear hitting set''. If $d\ge n^\gamma$, we use the ``sublinear hitting set''. Note that since \Cref{thm:MPC hitting set intro: sublinear} required $|U|,N =\poly d$, it cannot be applied to (nontrivial) instances with small, e.g., constant, $d$. However, exactly in such a regime, we can use the ``linear hitting set''.

We derandomize the computation of general hitting sets for the full range of variables, including $d\ge n^\delta$, and show that this has many applications: it is sufficient to derandomize spanners and approximate shortest paths in linear and sublinear MPC. Since derandomizing random sampling is a fundamental problem, we hope that our general deterministic hitting set will find applications in the derandomization of other problems.

\paragraph{Comparison to prior work}
\Cref{thm:MPC hitting set intro: linear,thm:MPC hitting set intro: sublinear} improve and generalize the state of the art for hitting sets in MPC: \cite{CoyC23,FischerGG22} compute hitting sets of size $O(|U|/d^{0.2})$ for the special case that $|U|=N$ and $d\le n^\delta$ in constant rounds in sublinear MPC. This result is not sufficient for us for two reasons: 1) The upper bound of $O(|U|/d^{0.2})$ is not small enough for our applications. In our distance computations, it is crucial that the hitting set is almost the same size as the randomized ideal $O(|U|\log n/d)$. 2) For \Cref{thm:sublinear_MPC_spanner,thm:MSSP_sublinearMPC_simplified,thm:APSP_sublinearMPCSimplified}, we also need hitting sets in the truly sublinear regime where $d\ge n^\delta$.

In a broad sense, random sampling for problem specific `hitting events' has been derandomized for several linear and sublinear MPC problems~\cite{czumaj2021simple,PaiP22,FischerGG23,GilibertiP24,ji2025fast}. 
In both hitting set algorithms and in most of these other works, it is assumed that $d\le n^\delta$.

We give an overview of our hitting set results, how we obtain them, and how they can be used to derandomize distance problems in \Cref{sec:overview}.



\newpage

\section{Technical Overview}\label{sec:overview}
In  \Cref{sec:overview:HS} we present our deterministic hitting set algorithms in a nutshell. Then, in \Cref{sec:overview:applications} we give an overview of the randomized components in the state-of-the-art algorithms for the distance problems presented in \Cref{sec:our resutls} and how we can use our deterministic hitting set algorithms to  systematically replace these randomized procedures.

\subsection{Deterministic Hitting Set Algorithms in a Nutshell}\label{sec:overview:HS}
To showcase our approach without many of the technical notation and details, we first consider a special case of the hitting set problem: the \emph{dominating set} problem.  For many of the theorems in this paper, the full generality of the hitting set problem is not necessary, and the dominating set problem suffices. For this overview, we further restrict to $d$-regular graphs, i.e., $|N(v)|=d$ for every $v\in V$. Here $N(v):= \{v\} \cup \{ u\in V: \{u,v\}\in E\}$ denotes the \emph{neighborhood} of $v$. 

\subsubsection{Dominating Sets}
 We define a \emph{$d$-dominating set} as follows.
\begin{definition}
    Let $G=(V,E)$ be a graph. A \emph{$d$-dominating set} is a subset of vertices $D\subseteq V$ such that for every vertex $v\in V\setminus D$ with $\deg(v)\ge d$, we have $N(v)\cap D\neq \emptyset$.
\end{definition}



First, we sketch how to obtain a $d$-dominating set of size $O\left(\tfrac{n\log n}{d}\right)$ using a randomized algorithm. This algorithm has two phases. First, we sample each vertex $v\in V$ with probability $p=\log n/d$ and denote the set of sampled vertices by $D_s$. Some vertices $v\in V$ may still be uncovered, i.e., $N(v)\cap D_s=\emptyset$ holds. Let $D_s
'$ be the set of uncovered vertices. By construction, $D:=D_s\cup D_s'$ is a $d$-dominating set of $G$.

Next, we consider the size of $D$. By linearity of expectation, we obtain $E[|D_s|]=O(np)=O(n\log n/d)$. For the size of $D_s'$, consider a vertex $v$. The probability that $v$ is not covered by $D_s$ is at most 
\begin{align}
    (1-p)^{|N(v)|}=\left(1-\tfrac{\log n}{d}\right)^d \le e^{-\log n} = n^{-1}.\label{eq:probs}
\end{align}
Now, again by linearity of expectation, we obtain that $|D_s'|\le1$ in expectation. Concluding, the expected size of the dominating set is 
\begin{align*}
    \E[|D|]=\E[|D_s|]+\E[|D_s'|] = O(n\log n/d)+1=O(n\log n/d).
\end{align*} 



\paragraph{Derandomization procedure}
To derandomize the process described above, we use the \emph{method of conditional expectation}~\cite{Luby93,motwani1989probabilistic}. This has become a standard technique in designing deterministic distributed algorithms in MPC, in particular see, e.g.,~\cite{ghaffari2018derandomizing,GhaffariK18,
ParterY18,
Censor-HillelPS20,
czumaj2021improved,
czumaj2021simple,
bezdrighinEG+2022,
DBLP:conf/podc/DoryP20,
PaiP22,
CoyC23,
FischerGG23,
GhaffariGHIR23,
GilibertiP24,
ji2025fast}. In the MPC model this derandomization method is particularly efficient if the randomness for the algorithm is very restricted in the sense that it can be extracted from an (extremely) short random seed. Next, we will discuss the method yielding a fully deterministic algorithm in more detail. 

For now, assume we are given a randomized dominating set algorithm whose only use of randomness is restricted to a random seed of $b$ bits -- think of a logarithmic seed length. If the expected dominating set size of the algorithm is at most some value $x$, then the goal of the method is to deterministically compute a random seed $s^*$ that attains this expectation; once the seed $s^*$ is known we can run the given algorithm using $s^*$ as the (deterministic) source of randomness  and can deterministically compute  a dominating set with size at most $x$.

To see how a seed $s$ leads to a dominating set, we let each random seed $s$ denote a set of sampled vertices $D_s$. We let $D_s'\subseteq V$ denote the set of vertices that are not covered, and hence added in phase two. 
This means we obtain the $d$-dominating set $D_s\cup D_s'$. As said above, the goal is now to find $s^*$ such that $|D_{s^*}\cup D_{s^*}'| \le x$. 
This is done iteratively, by considering $\Theta(\log n)$ bits of $s^*$ at a time, and compute what the best choice is for these bits: given a prefix $r$, we need to compute the next $\Theta(\log n)$ bits. The following observation summarizes the crucial step in our derandomization method (\Cref{lem:MPC_cond_exp}).

\begin{observation}\label{obs:parallel_no_comm}
    Given a prefix $r$, to compute the next $\Theta(\log n)$ bits, we need to determine if $v\in D_s$ or $v\in D_s'$ for all seeds $s$ with prefix $r$ \emph{in parallel}. In particular, this means it is sufficient if we can determine if $v\in D_s$ or $v\in D_s'$  for any seed $s$ \emph{without communication}. 
\end{observation}


We show that in the linear MPC model, we can store the information such that $v\in D_s$ or and $v\in D_s'$ can both be checked internally. We repeat this for $O(\lceil \tfrac{b}{\log n}\rceil)$ iterations until all bits are fixed. 


\paragraph{Short random seeds.}
 In the randomized sampling-based algorithm at the start of this section, we sampled each vertex \emph{independently} with probability $p$. This needs at least one random bit per vertex, so at least $n$ random bits in total. It turns out that we can cope with less randomness. We show that we in total need either
 \begin{enumerate}
     \item $O(\log n \cdot \log \log d)$ bits for a dominating set of size $O(n\log n/d)$; or
     \item  $O(\log n)$ random bits for a dominating set of size $O(n/d^{0.99})$.\label{item:fastHS}
 \end{enumerate}

As explained above, this leads to $O(1)$ and $O(\log \log d)$-round algorithms respectively. Many applications (sometimes after modifications) do not need dominating sets of the optimal size, and the slightly increased size of (2) suffices. Thus, for this warm-up we focus on (2). 

Here, we sample with  probability $p=1/d^{0.99}$ rather than the optimal $p=\log n/d$. 
This means that the expected size of $D_s$ becomes  $O(n/d^{0.99})$. By this increased size, we obtain some slack in the randomness: we no longer need complete independence between vertices to bound the size of the uncovered vertices $D_s'$. Instead, $O(1)$-\emph{wise independence} suffices\footnote{Given $n$ random variables $X_1, \dots X_n$, we say that they are \emph{$k$-wise independent} if any subset of $k$ random variables behaves as if they are independent. For (1), we use a different way to limit the seed size via pseudorandom generators~\cite{GopalanY20}.}. 
Vadhan~\cite{Vadhan12} gives a procedure to extract $n$ $O(1)$-wise independent random variables sufficient for our sampling process from a seed with $O(\log n)$ random bits.
Tail bounds for $k$-wise independent random variable then show that a vertex is uncovered with probability at most $1/d$ (compare to \Cref{eq:probs}). Hence the expected number of vertices $D_s'$ that join in the second phase is $O(n/d)$. In total this gives a $d$-dominating set of size $O(n/d^{0.99})+O(n/d)=O(n/d^{0.99})$.

Combining (2) with the explained derandomization method, we obtain the following simplified version of \Cref{thm:MPC hitting set intro: linear}. 
\begin{theorem}\label{thm:warmup_linear}
     There exists a deterministic MPC algorithm that, given a $d$-regular graph $G=(V,E)$ on $n$ vertices and $m$ edges, computes a $d$-dominating set of size $O(\tfrac{n}{d^{0.99}})$ in constant rounds. The algorithm uses linear $\Theta(n)$ local space and $O(m)$ total space.  
\end{theorem}

To obtain an analogous result in the sublinear regime, another challenge appears. We describe this challenge and how to overcome it next.

\paragraph{Sublinear MPC}
In the sublinear MPC model, we only have $O(n^\delta)$ local space, for some constant $\delta$.
In this regime, we can still internally find out if $v\in D_s$ for a single random seed $s$ as required in \Cref{obs:parallel_no_comm}. However, to see if $v\in D_s'$, i.e., to see if $N(v)\cap D_s=\emptyset$ for $\poly n$ seeds in parallel, we need access to all of $N(v)$. Since $N(v)$ does not necessarily fit on a single machine, we cannot check \emph{without communication} if $N(v)\cap D_s=\emptyset$. This means that without adaptations, the explained derandomization cannot be applied. Most previous sublinear MPC algorithms~\cite{FischerGG22,CoyC23,FischerGG23} get around this by assuming that $d\le n^\delta$ holds for their dominating set (hitting set) instance. 
However, not all applications satisfy $d\le n^\delta$. 
We overcome this by introducing an additional \emph{sparsification step}. A similar approach was applied by Giliberti and Parsaeian~\cite{GilibertiP24} for the special case of ruling sets.

\paragraph{Sparsification in Sublinear MPC}
The goal of our sparsification step is to obtain some subset $V' \subseteq V$, such that $N(v)\cap V' = n^\delta$ for every node $v\in V$. 
Let us first consider a simple randomized routine that achieves this (in expectation): we sample every node with probability $n^\delta/d$. This gives the following guarantees (in expectation):
\begin{itemize}
    \item $d'(v):=|N(v)\cap V' | = n^\delta$, and
    \item $|V'|=n\cdot \tfrac{n^\delta}{d}$.
\end{itemize}

Given such a set $V'$, we can fall back to the algorithm for the linear case since $d'\le n^\delta$. This would give a dominating set of size $O\left(\tfrac{|V'|}{n^{0.99\delta}}\right)= O\left(n\cdot \tfrac{n^\delta}{d\cdot n^{0.99\delta}}\right)= O\left(\tfrac{n}{d\cdot (n^\delta)^{-0.01}}\right) = O\left(\tfrac{n}{d^{0.99}}\right)$, where the last equality holds since $d\ge n^\delta$. 

So we reduced our problem to finding $V'$ deterministically; we cannot directly sample with probability $p=\tfrac{n^\delta}{d}$ and use the same method for derandomization suffers from the exact same problem as before. 
Instead we sample iteratively with probabilities $p_1,p_2, \dots, p_k$ such that $p=p_1\cdot p_2   \cdots   p_k$, and each $p_i\ge 1/n^{\delta/2}$ -- for this overview, we assume $p_i= 1/n^{\delta/2}$ for each $i$. In the randomized algorithm, this leads to exactly the same results. However, we now have a big advantage: if we sample with probability $p_i$, we can satisfy \Cref{obs:parallel_no_comm} as follows. Let $v\in V$ be some node with $d\gg n^{\delta}$ neighbors stored on $d/n^\delta$ consecutive machines. With one sampling step, the expected number of sampled nodes on each machine is 
\begin{align*}
    n^{\delta}\cdot p_i= n^{\delta} \cdot 1/n^{\delta/2} =n^{\delta/2}.
\end{align*}
And here is the crux: the number of sampled nodes on each machine can be checked locally, without any communication, satisfying \Cref{obs:parallel_no_comm}. Moreover, since this expectation is still quite large, we can use tail bounds to show that even with limited dependence, we obtain this expectation. 
In the above, we changed the perspective from \emph{nodes} to \emph{machines}: when sampling with a too low probability, there is only a guarantee on the sampled nodes in the \emph{neighborhood} as a whole, not on each machine. Instead, we measure the quality of a seed on each \emph{machine}, for which no communication is necessary. To ensure feasibility of our subsampling, i.e., to ensure that the set of sampled nodes is not too small, we apply our derandomization method with an adjusted potential function. For full details we refer to \Cref{sec:sublinear_HS}.

So combining our sparsification step with the linear MPC algorithm on the sparsified instance, we obtain the following simplified version of \Cref{thm:MPC hitting set intro: sublinear}. 
\begin{theorem}\label{thm:warmup_sublinear}
    Let $\delta>0$ be a positive constant. 
     There exists a deterministic MPC algorithm that, given a $d$-regular graph $G=(V,E)$ on $n$ vertices and $m$ edges, for some $d= n^{\Omega(1)}$, computes a $d$-dominating set of size $O(\tfrac{n}{d^{0.99}})$ in constant rounds. The algorithm uses linear $\Theta(n^\delta)$ local space and $O(m)$ total space.  
\end{theorem}


\subsection{Deterministic Algorithms for Distance Approximation}\label{sec:overview:applications}

In this paper, we design algorithms for distance problems using \emph{hitting set} algorithms as a central tool. Distance problems aim to (approximately) compute the distances between pairs of vertices in a graph. A fundamental strategy for coping with the complexity of shortest paths is to compress the global distance structure into a small set of representative vertices.
A basic example is the \emph{clustering} technique used to construct a spanner. 
It aims to decompose the vertex set into several \emph{clusters} so that the vertices belonging to the same clusters are close to each other.
Then the distance between clusters approximates the distance between two vertices in the original graph.
For a simple case, we consider an unweighted $d$-regular graph.
The idea is to select a small set of vertices such that each vertex $v$ has a selected vertex among $v$ or its neighbors. Then we can set a \emph{cluster} as the set of vertices adjacent to the same selected vertex. If a vertex $v$ is adjacent to multiple selected vertices, then it chooses an arbitrary one and joins the cluster.
A clustering with $x$ clusters yields an $O(1)$-spanner with $O(n + x^2)$ edges consisting of \textsf{(i)} one edge between every two adjacent clusters and \textsf{(ii)} one edge between each non-selected vertex and its adjacent selected vertex. 
We can further sparsify the graph by (possibly recursively) computing a spanner of the graph
obtained from the clustering while increasing the stretch.

In summary, partitioning the vertex set into a small number of clusters yields a sparse spanner. The key step in computing such a clustering is to select a small set of vertices such that every non-selected vertex is adjacent to at least one selected vertex. This selection task can be formulated as a hitting set problem on the graph; in particular, since we have a $d$-regular graph, it corresponds exactly to the $d$-dominating set problem introduced in \Cref{sec:overview:HS}.

\medskip 

Many randomized algorithms for approximate distance problems rely on sampling a small set of representative vertices, with the goal of ensuring suitable coverage of the graph, typically by guaranteeing that a sampled vertex lies close to most vertices. In this paper, we replace these random sampling steps with deterministic hitting set constructions to obtain our algorithms for distance problems. In the following, we briefly outline this approach for each of the distance problems considered in this paper.

\paragraph*{Spanners in linear MPC and Congested Clique}
As summarized by \Cref{thm:linear_MPC_spanner_weighted}, we study spanners of weighted graphs in linear MPC and Congested Clique.
In Congested Clique, there is a reduction that extends a spanner algorithm for unweighted graphs to weighted graphs~\cite{chechik2022constant}. We extend this reduction to work in linear MPC, see \Cref{sec:weighted_spanner_linearMPC}.  
Therefore, we focus on derandomizing the algorithm of~\cite{podc2021spanner} for unweighted graphs, which is a constant-round randomized algorithm in Congested Clique and near-linear MPC.

The algorithm partitions the graph into $O(\log n)$ hierarchical levels according to vertex degrees. 
At each level, it utilizes the presented clustering technique by selecting a small vertex set such that every vertex with a certain degree is adjacent to at least one selected vertex, and contracting the vertices to the adjacent selected vertex in some deterministic canonical way. 
We obtain a separate contracted graph for each of the  $O(\log n)$ levels.

It then computes a spanner of each one of the contracted graphs using a randomized sparsification technique. 
By executing the processes for all hierarchy levels in parallel, they achieve a constant-round algorithm. However, to enable this parallelism, they require $\Theta(n\log n)$ local space even for the MPC model.
In summary, the algorithm uses randomization in two places, first for constructing $O(\log{n})$ clustering graphs. Second, for a sparsification that allows to construct sparse spanners on each one of the clustering graphs.

In~\Cref{sec:unweighted_spanner_linearMPC}, we show how to derandomize the algorithm. We also improve the local space to $O(n)$. Recall that we need to derandomize two steps. First, building the clustering graphs. Second, the randomized sparsification technique.
For the latter step, we extend a deterministic sparsification designed for Congested Clique by Leitersdorf~\cite{Leitersdorf22}.
Our main contribution lies in the first step to build the clustering.

As discussed above, the clustering problem can be phrased as a hitting set problem and we use the hitting set algorithm of~\Cref{thm:MPC hitting set intro: linear} Part 1 to achieve a deterministic algorithm. This algorithm runs in constant rounds, but this comes at a cost that 
the resulting hitting set is significantly larger than the expected size of the their randomized algorithm. We prove that this is still good enough for the construction. 
Moreover, we can construct the $O(\log{n})$ hitting sets simultaneously in $O(1)$ rounds. 
Finally, we use a tighter analysis and obtain a linear-space MPC algorithm, rather than a near-linear-space one.


\paragraph*{Spanners in sublinear MPC}
To obtain the deterministic spanner algorithm in the sublinear MPC model from \Cref{thm:sublinear_MPC_spannerSimplified} we deviate from the above procedure but instead build upon the randomized algorithm  from \cite{biswas2021massively}.
Their randomized spanner algorithm for the sublinear MPC model uses a  recursive hierarchical structure~\cite{biswas2021massively}.
{At the core, it combines certain steps of growing clusters and steps where it contracts  clusters 
improving the round complexity from $O(k)$ to $O(\log{k})$ while increasing the stretch compared to the seminal spanner algorithm by Baswana and Sen~\cite{baswana2007simple}.
}

The randomness of the original algorithm occurs in {the steps that grow clusters. 
We reduce it to a  hitting set problem:} 
\begin{itemize}
    \item Given a set $\mathcal C$ of clusters of a weighted graph $G$ and a parameter $0<p<1$, for each node $v$ in $G$ the set $S_v\subseteq \mathcal C$ is defined as the $1/p$ closest clusters of $\mathcal C$ to $v$ and adjacent to $v$ in $G$.
    \item The goal is to compute $\mathcal D\subseteq \mathcal C$ so that
    $\mathcal D\cap S_v\neq \emptyset $ for every node $v\in G$.
\end{itemize}

In \Cref{sec:weighted_spanner_sublinear}, we derandomize this algorithm using our constant-round hitting set algorithms in the sublinear MPC model (\Cref{thm:MPC hitting set intro: linear,thm:MPC hitting set intro: sublinear} Part 1). 
The obtained hitting set size is not optimal and may exceed the expected size achieved by the randomized algorithm.
To address the overhead, we rebalance the parameters of the algorithm, including the probabilities $p$ used in the subroutines and the hierarchical depths.

\paragraph*{Constant approximate APSP in the Congested Clique}
Note that our spanner algorithm deterministically constructs an $O(\log n)$-spanner with $O(n)$ edges for an edge-weighted graph with $n$ vertices. 
The spanner algorithm immediately gives an $O(\log n)$-approximate APSP algorithm running in constant rounds in the Congested Clique model.

In \Cref{sec:apsp_corollary}, we further obtain a constant-approximate APSP that runs in $O(\log\log\log n)$ rounds in the Congested Clique model by derandomizing the algorithm of~\cite{BuiCCDL24} as summarized by \Cref{thm:apsp_in_cc}.
The randomness in the original algorithm arises from two types of subroutines: constructing an $O(\log n)$-approximate spanner of small size and computing hitting sets for the universe $U:= V$ and $N:=O(|V|)$  explicitly formulated sets.
Here, we adapt our deterministic spanner algorithm of \Cref{thm:linear_MPC_spanner_weighted} and our deterministic hitting set algorithms of \Cref{thm:MPC hitting set intro: linear}. 


\paragraph*{Approximate shortest paths in the near-linear MPC model}
\Cref{thm:MSSP_sublinearMPC_simplified,thm:APSP_sublinearMPCSimplified} summarize our deterministic algorithms in the near-linear MPC model.
The single-source shortest paths (SSSP), multi-source shortest paths (MSSP), and all-pairs shortest paths (APSP) problems were studied in the randomized version of this model by Dory and Matar~\cite{dory2024massively}.
 Their algorithm used a subroutine inspired by the Thorup–Zwick approach~\cite{thorup2005approximate}.
{
Briefly, it begins with hierarchical sampling $V=A_0\supseteq A_1\supseteq\ldots\supseteq A_\ell$ for some suitable parameter $\ell$.
Then, the algorithm works in $\ell$ sequential iterations for each of $A_i$'s.}
At the $i$th iteration, for a vertex $v$, 
it adds some weighted edges between $v$ and other close vertices $u\in A_{i-1}$.
If there is a close sampled vertex $u\in A_i$ reachable from $v$ using a small number of edges, then it adds one edge between $v$ and $u$.
Otherwise, it adds many edges for all $u\in A_{i-1}$ close to $v$.
Thus, the goal is for any vertex $v$, to bound the number of non-sampled vertices near $v$ if there is no close sampled vertex from $A_i$ near $v$.
We formulate this goal as a hitting set problem:
\begin{itemize}
    \item Given an unweighted graph $G$, multiple sources $A_{i-1}$, a weighted edge set $H$, a hop bound $h$, and a threshold $d$,
    for each vertex $v$, $S_v\subseteq A_{i-1}$ is defined as the $O(d\log n)$ closest sources among $A_{i-1}$ to $v$, where distances are measured by the shortest paths using at most $h$ edges in $E(G)\cup H$. 
    \item The goal is to compute a small $A_i \subseteq A_{i-1}$ so that $A_i\cap S_v\neq \emptyset$ for every vertex $v$ in $G$.
\end{itemize} 
In the original randomized algorithm~\cite{dory2024massively}, during the hierarchical sampling, $A_i$ was sampled from $A_{i-1}$ with probability $1/d$ so that the expected size of $A_i$ is at most $O(|A_{i-1}|/d)$ and the above condition holds with high probability, without explicitly formulating a hitting set problem {and actually even without explicitly computing the sets $S_v$.}

In \Cref{sec:applications_heterogeneous}, 
we avoid the hierarchical sampling by deterministically obtaining $A_i$ from $A_{i-1}$ at the beginning of each $i$th sequential iteration.
Particularly,
we first explicitly formulate the above hitting set problem.
Note that even if the desired $S_v$'s are not a subset of the neighbors of $v$, they are reachable using at most $h$ edges.
This property allows us to compute the sets efficiently in $O(h)$ rounds.
Then we deterministically obtain $A_i$ from $A_{i-1}$ of size at most $O(|A_{i-1}|/d)$ in $O(\log\log n)$ rounds by using \Cref{thm:MPC hitting set intro: linear,thm:MPC hitting set intro: sublinear} Part 2.
This is the same as the expected size desired by the previous random process.
Although deterministically computing $A_i$ requires $O(h+\log\log n)$ additional rounds per sequential iteration, this is not a bottleneck for the overall algorithm.
Thus, it does not affect the final round complexity.

{Down the line \cite{dory2024massively} uses the subroutine for various choices of the parameters\footnote{
In their full algorithm \cite{dory2024massively} sometimes use the subroutine by simulating $\ell=O(\log\log n)$ iterations in parallel while our method inherently requires that the $\ell$  iterations are performed sequentially. However, also in their algorithms the runtime bottleneck of their distance computations stems from using the subroutine sequentially for $\ell$ iterations. Hence, we obtain the same overall runtime.}, yielding a randomized  $\textsf{poly}(\log\log n)$-round  approximate distance algorithms; due to the deterministic subroutine in our case the results become deterministic as summarized in \Cref{thm:MSSP_sublinearMPC_simplified} and~\Cref{thm:APSP_sublinearMPCSimplified}. 
}

\subsection{Further Related Work}
\label{sec:related_work}

\paragraph{Congested Clique}
A linear MPC algorithm can be simulated in Congested Clique by allowing a constant time overhead~\cite{BehnezhadDH18}.
However, the reverse does not hold usually. This is because in the MPC model the goal is to use a total space similar to size of the input, $\tilde{O}(m)$, where in Congested Clique there are no specific memory requirements, hence a direct simulation may use too much memory.

Next, we review the distance approximation algorithms in the Congested Clique model. 
There are deterministic algorithms to compute a $(2k-1)$-spanner of size $\tilde O(n^{1+1/k})$ in $\poly {(\log\log n)}$ rounds or to compute an $O(k)$-spanner of size $O(k\cdot n^{1+1/k})$ in $O(\log k)$ rounds~\cite{ParterY18}.
Additionally, there are deterministic algorithms 
for the approximate APSP problem~\cite{censor2019algebraic,le2016further,censor2019sparse, DBLP:conf/podc/Censor-HillelDK19,DBLP:conf/podc/DoryP20}, including 
a $\poly (\log n)$-round algorithm for weighted graphs~\cite{DBLP:conf/podc/Censor-HillelDK19} and 
a $\poly (\log\log n)$-round algorithm for unweighted graphs~\cite{DBLP:conf/podc/DoryP20}. 
These APSP algorithms are based on matrix multiplication and hence require large total memory, which makes them unsuitable for linear MPC.
Furthermore, the current fastest algorithms are randomized, such as the $O(1)$-round $O(k)$-spanners~\cite{podc2021spanner, chechik2022constant}, $O(1)$-round $O(\log{n})$-approximation APSP~\cite{podc2021spanner, chechik2022constant}, and the $O(\log{\log{\log{n}}})$-round $O(1)$ approximation APSP algorithms~\cite{BuiCCDL24}.

\paragraph{Deterministic hitting sets in other distributed models} 
{
Censor-Hillel, Parter, and Schwartzman~\cite{Censor-HillelPS20} use hitting set like techniques to derandomize the Baswan-Sen Spanner (and MIS) in the Congested Clique. Later, 
Parter and Yogev~\cite{ParterY18} gave an explicit deterministic hitting set algorithm in the Congested Clique. The authors apply it to compute spanners and follow up work by Censor-Hilel, Dory, Korhonen, and Leitersdord~\cite{DBLP:conf/podc/Censor-HillelDK19} use it to compute APSP.  The hitting set algorithm takes $O((\log \log n)^3)$ rounds, but with newer pseudo-random generators of~\cite{GopalanY20} this becomes $O(\log \log n)$ rounds. 
They also have a constant construction that leads to a hitting set of size $O(n^{17/16}/d^{0.5})$. 
\Cref{thm:MPC hitting set intro: linear} can be seen as a generalization of their work,
where 1) we need to take additional care to make sure all internal computations do not exceed the local memory, 2) our constant round algorithm achieves a sparser hitting set, and 3) our algorithms work for a wider range of parameters. }

{
Later, Dory and Parter~\cite{DBLP:conf/podc/DoryP20} introduces the notion of a `soft' hitting set. For their shortest path algorithm, they need a definition akin to hitting sets with slightly different properties. They dub it a soft hitting set, and provide an algorithm to compute it in the same round complexity as~\cite{ParterY18}. }

{
Deterministic hitting sets and their application to spanners have also been studied in the related \CONGEST model~\cite{GhaffariK18,bezdrighinEG+2022,GhaffariGHIR23}.
This model is similar to the Congested Clique except that the nodes can only communicate with their neighbors~\cite{peleg00}.
This line of work lead to a deterministic hitting set in $\poly \log n$ rounds by Ghaffari, Grunau, Haeupler, Ilchi, and Rozho{\v{n}}~\cite{GhaffariGHIR23}, which they apply to derandomize the Baswana-Sen algorithm~\cite{baswana2007simple} to give a deterministic spanner algorithm.
They perform a multi-step process to make it computationally efficient not only in the \CONGEST model, but also in the PRAM model.
However, such techniques do not directly translate to the MPC model; it would lead to $\poly \log n$-round algorithms, which is exponentially slower than what we aim for. }

\tijni{next two paragraphs are new}
Berger, Rompel, and Shor~\cite{BergerRS94} provide a PRAM algorithm that works also for our variant of hitting set (termed `balanced set cover' in their paper). However, their algorithm is a recursive process that contains $\log m$ iterations. This leads to (at least one) factor $\log m$ in the round complexity, also in a possible MPC implementation. 

Agarwal and Ramachandran~\cite{AgarwalR20} compute a `blocking set' in \CONGEST, an object for shortest path computation that `hits' many paths. The algorithm is based on an implementation of a result of the aforementioned paper~\cite{BergerRS94} and implementing it in MPC also would come with at least logarithmic overhead.



\paragraph{Deterministic MPC algorithms}
For connectivity, there are deterministic algorithms in sublinear MPC that take $O(\log D+\log \log_{m/n} n)$ rounds~\cite{CoyC23,FischerGG22}. This matches the randomized state of the art, which is conjectured to be optimal~\cite{andoni2018parallel,behnezhad2019near}. Coy and Czumaj also show how to obtain a minimum spanning tree (MST) deterministically in $O(\log n)$ rounds, matching a conditional lower bound of $\Omega(\log n)$ which follows from the widely-believed 1vs2-Cycle conjecture. 
In linear MPC, 
Nowicki~\cite{nowicki2021deterministic} gives a deterministic $O(1)$-round algorithm for MST, hence also solving connectivity in this regime. 

A substantial body of work has established deterministic MPC algorithms for fundamental symmetry-breaking problems, showing that randomness can often be eliminated without sacrificing efficiency. 
While problems such as maximal matching, graph coloring, and maximal independent set may appear different from distance problems, their algorithms rely on similar techniques. In particular, sparsification and the method of conditional expectation have been used extensively to obtain deterministic counterparts. 
Moreover, MIS and its relaxed variants, such as ruling sets, are sometimes used in distance algorithms, for example, to select cluster centers in clustering-based approaches.
Czumaj, Davies, and Parter~\cite{czumaj2021simple} give a deterministic $(\Delta+1)$-vertex coloring algorithm in linear MPC using $O(1)$ rounds. 
In sublinear MPC, the authors gave an $O(\log \log \log n)$-rounds algorithm in follow up work~\cite{czumaj2021improved}. This matches both the randomized state of the art and a conditional lower bound for the problem~\cite{ghaffari2019conditional}.

Next we consider maximal independent set and maximal matching. 
Fischer, Giliberti, and Grunau~\cite{FischerGG23} provide algorithms of $O(\log \alpha+\log \log n)$ rounds for both problems in sublinear MPC, where $\alpha$ denotes the arboricity of the graph. 
The randomized state of the art is $O(\sqrt{\log \alpha}\log \log \alpha +\log \log n)$~\cite{ghaffari2019sparsifying,ghaffari2020improved}.
This gap is one of the motivations to look at the ruling set problem. 

In the linear MPC model, Giliberti and Parsaeian~\cite{GilibertiP24} give a constant round algorithm for computing $2$-ruling sets, matching the randomized algorithm by Cambus, Kuhn, Pai,
and Uitto~\cite{cambus2023time}.
Giliberti and Parsaeian also provide a deterministic $\tilde O(\sqrt{\log n})$ algorithm in sublinear MPC, which at the cost of slightly superlinear $O(m+n^{1+\eps})$ global memory is improved to $O(\log^{1/3}\Delta\log \log \Delta+\log \log n)$ rounds by Ji, Kothapalli, Pemmaraju, and Singh~\cite{ji2025fast}. The fastest randomized algorithm for this regime takes $\tilde O(\log^{1/6}\Delta \cdot  \log \log n)$ rounds~\cite{kothapalli2020sample}.

\paragraph{Derandomization in MPC} 
The algorithm of Nowicki~\cite{nowicki2021deterministic} mentioned above provides an inherently deterministic algorithm. All other algorithms contain a derandomization  step. 
As discussed above, some of the algorithms~\cite{CoyC23,FischerGG22} solve exactly the hitting set problem.
Others~\cite{czumaj2021simple,PaiP22,FischerGG23,GilibertiP24,ji2025fast} derandomize a more problem specific sampling step. Although the methods here are similar to ours (limited independence and the method of conditional expectation), the results are specific to the algorithm at hand and do not provide general derandomization of sampling. 

\tijni{new}
In our algorithms, we use $k$-wise independence and pseudorandom generators for reducing the amount of randomness for sampling nodes. This is in the same spirit as Newman's method~\cite{Newman91}, a reduction from public to private randomness in communication complexity. Potentially, we can also use Newman's method here as an alternative pseudorandom generator. However, Newman's method would require exponential computation and certain space requirements. It is unclear whether such space requirements can be met -- in particular in sublinear MPC. 

\tijni{could also add: blackbox randomization reduction in MPC by \cite{CharikarMT21}, but this is not helpful for making entire algs determinstic and uses too many random bits for our sampling step}

\tijni{also new}
\paragraph{Relation to Set Cover}
For both the classic tasks of finding an approximation of the optimal set cover or optimal hitting set it is a major open problem of the field to get sublogarithmic-round complexity in MPC, even using randomized algorithms. The hitting set problem in this paper is different from these classic variants and hence we even find $O(1)$-round algorithms; the solution is still sufficient for our applications. Different variants of these problems have also been considered in multiple works in \CONGEST~\cite{AgarwalR20} and PRAM~\cite{BergerRS94} where the distinction between our variant and the classic optimization problem does not matter as much because it usually only affects logarithmic factors, but in our case it is crucial to solve a variant of the problem that allows $O(1)$-round algorithms.

\subsection*{Outline}\kyungjini{In the SPAA camera ready, there is no appendix, so this outline is useless. So I omitted in spaa version by removing from our section tex files, instead, I added them here in the main.tex}
In \Cref{sec:preliminaries}, we introduce our models formally, provide MPC primitives, and summarize derandomization techniques used in our work. 
In \Cref{sec:hitting sets}, we give the full algorithms and proofs for our hitting set results. 
In \Cref{sec:applications_spanners}, we present the detailed algorithms and their analysis for spanners. In \Cref{sec:applications_shortest_paths}, we provide the results of the deterministic algorithms for the approximate shortest paths problems, including the approximate APSP problem.
\newpage
\section{Graphs, MPC Primitives, and Derandomization}
\label{sec:preliminaries}

\subsection{Graphs}\label{sec:prelim_graphs}
\paragraph{Notation.}
In this paper, we consider undirected graphs $G=(V,E)$ with $n$ vertices and $m$ edges. An edge is represented by an unordered pair $\{u,v\}$ of vertices. We say that a vertex $u$ is \emph{adjacent} to $v$ if $\{u,v\}\in E(G)$, and that the edge $\{u,v\}$ is \emph{incident} to both $u$ and $v$. When $G$ is not clear from context, we denote the vertex set and edge set of $G$ by $V(G)$ and $E(G)$, respectively, and let $\deg_G(u)$ be the number of edges in $E(G)$ incident to $u\in V(G)$.

When we write $G=(V,E,w)$, we refer to an edge-weighted graph with vertex set $V$ and edge set $E$. The weight function $w:V^2\to\mathbb{R}$ is symmetric, where $w(u,v)=\infty$ if $\{u,v\}\notin E$ (equivalently, $w:E\to\mathbb{R}\setminus\{\infty\}$). Thus each edge $\{u,v\}\in E$ has weight $w(u,v)=w(v,u)$. If all edges have unit weight, i.e., $w(u,v)=1$ for every $\{u,v\}\in E$, we simply write $G=(V,E)$ and call $G$ an \emph{unweighted} graph.

A path $\pi$ between vertices $s,t\in V$ in a graph $G=(V,E,w)$ is a sequence of edges in $E$ such that consecutive edges share an endpoint while the first and the last edges are incident to $s$ and $t$, respectively. The \emph{length} of $\pi$ is the sum of the weights of edges in the path. In an unweighted graph, this equals the number of edges in the path. A path is a \emph{shortest path} if no shorter path exists between $s$ and $t$.
Furthermore, we define the \emph{distance} between two vertices in $G$ as the length of a shortest path between them.

\paragraph{Spanners.}
 For an edge weighted graph $G=(V,E)$, a subgraph $H=(V,E_H)$ is called an $\alpha$-spanner of $G$ if for any two vertices $u,v\in V$, 
the distance between them in $H$ is at most $\alpha$ times of the distance in $G$.
Equivalently, it can be defined with a condition on the \emph{stretch} of each edge. 
That means $H$ is an $\alpha$-spanner if the distance between $u$ and $v$ in $H$ is at most $\alpha\cdot w(u,v)$ for each edge $\{u,v\}\in E$, where $w(u,v)$ is the weight of the edge $\{u,v\}$ in $G$.
It is well-known that every graph with $n$ vertices has a $(2k-1)$-spanner with $O(n^{1+1/k})$ edges, and such a spanner can be constructed by a greedy centralized algorithm with {a linear space complexity}~\cite{althofer1993sparse}.

\paragraph{All-pairs shortest paths problem.}
For an edge-weighted graph $G=(V,E)$, 
the \emph{all-pairs shortest paths (APSP)} problem asks to build a data structure supporting a query algorithm to return the distance between \emph{any} two given vertices $v,u\in V$ in $G$.
Additionally, an $\alpha$-approximate APSP problem aims to compute an $\alpha$-approximate distance between two query vertices for the approximate factor $\alpha>0$.
For instance, an $\alpha$-spanner supports a query of $\alpha$-approximate distance between any two vertices.
Therefore, $\alpha$-spanner immediately gives a solution for an $\alpha$-approximate APSP problem.

There are variations of the APSP problem by restricting the queries.
For instance, by fixing one query vertex as $s\in V$, the (approximate) \emph{single-source shortest paths (SSSP)} problem claims to compute the (approximate) distance between $s$ and any other vertex $v\in V$.
Additionally, for a vertex set $S\subseteq V$, the (approximate) \emph{multi-source shortest paths (MSSP)} problem supports the (approximate) distance query between any two vertices $s\in S$ and $v\in V$.

\subsection{The MPC Model}\label{sec:prelim_MPC}

In the \emph{massively parallel computation (MPC) model}, the algorithm consists of a sequence of synchronous \emph{rounds} using multiple $N$ \emph{machines} with \emph{local space} $L$ and the \emph{total space} $M=N\cdot L$. 
Usually, the total space is forced to be linear (or near-linear) in the input size.
During each round, the machines are isolated from each other, and work individually.
The machines can communicate with each other only between rounds.
In each of the communication rounds, the messages received or sent per machine cannot exceed the local space.
Throughout, we assume word size $O(\log M)$, such that the identifier of a machine can fit in one word.  
The communication is much more expensive compared to the individual computation.
Therefore, the model aims to minimize the communication, that is the number of rounds, while preserving the small local space per machine.

\paragraph*{Various regimes in the MPC graph algorithms.}
For the graph algorithms of an input graph $G$ with $n$ vertices, the MPC graph algorithm considers the local space related to the size $n$.
For instance, the \emph{linear}, the \emph{near-linear}, and the \emph{sublinear} regimes consider the local space $L=O(n),O(n\cdot \textsf{poly}(\log n))$, and $O(n^{\delta})$, respectively, where $0<\delta<1$ is a fixed constant.
The sublinear regime is the most desirable while the most challenging, since it does not allow a single machine to store information about all the vertices of the graph.
For instance, the hardness has been conjectured to require at least $\Omega(\log n)$ rounds in any sublinear MPC model to distinguish between a large cycle of length $n$ and two cycles of lengths $\lfloor n/2\rfloor$ and $\lceil n/2\rceil$, called the \emph{1-vs-2 cycles conjecture}. 

In this aspect, the \emph{heterogeneous MPC model} has been introduced~\cite{FischerHO25} that varies the local spaces of the machines.
In this model, there is a single machine or a small number of machines with larger local space, and many machines with small local space.
Even for a single machine with a near-linear local space and additional machines with sublinear local space, the heterogeneous MPC model circumvents the conditional hardness results from the sublinear regime, including the 1-vs-2 cycles problem and the approximate SSSP problem.
This is why various (approximate) distance problems, including SSSP and APSP problems, have been studied in the heterogeneous MPC model~\cite{FischerHO25,dory2024massively} instead of the classic sublinear MPC model.

 

\subsection{MPC Primitives}

\paragraph*{Basic tools implemented in the MPC model.}
In the MPC model, several useful tools have been developed, including sorting and summation of elements.
In this paper, we also utilize the algorithms as illustrated by the following lemma.

\begin{lemma}[\cite{goodrich1996sorting,goodrich2011sorting,dinitz2019massively}]\label{lem:sorting_prefixsum_MPC}
    Sorting, filtering, and prefix sum on a sequence of $M$ messages
    can be performed deterministically in $O(\log_\LS M)$ rounds in the MPC model using $\LS$
    local space and $O(M)$ total space.
        
    Additionally, in $O(\log_\LS N)$ rounds, a machine $\mathcal M$ can broadcast a constant number of messages to $N$ consecutive machines $\mathcal M_1,\ldots,\mathcal M_N$, given that $\mathcal M$ has the ids of the first and last machines ($\mathcal M_1$ and $\mathcal M_N$).    
\end{lemma}

\paragraph{Running Multiple MPC Algorithms Simultaneously.}
In some applications, we have multiple MPC algorithms, that run on part of the input, and need only part of the total space. Such algorithms can then we run simultaneously, where the total space becomes the sum of the space used by each algorithm. 


\begin{observation}\label{obs:parallelism}
    Let $A_1,\ldots,A_x$ be $x$ MPC algorithms, where each algorithm 
    $A_i$ runs in $T_i(n)$ rounds with $L_i(n)$ local space and $M_i(n)$ total space when $n$ input messages are given for $A_i$.
    
    When $n_i$ input messages are given along with the functions $L_i$ and $M_i$, that are labeled with index $i$, for each $i\in [x]$, we can run all algorithms simultaneously in $O(\max_i\{T_i(n_i)\}+\log_L M)$ rounds in the MPC model with $L$ local space and $M$ total space if:
    \begin{enumerate}
        \item $M\geq 6\sum_i M_i(n_i)+6x$ and
        \item $L\geq \max_i\{L_i(n_i)\}$. 
    \end{enumerate}
\end{observation}
\begin{proof}
    During this algorithm, we suppose that we have $\lfloor M/L\rfloor$ machines with $L$ local space, and the input messages and the functions $M_i$'s and $L_i$'s are given arbitrarily distributed across the machines.  
    We have at most $\sum_i M_i(n_i)$ input messages and $2x$ functions, while we do not know the exact value of $n_i$.
    In the following, 
    we first calculate the needed space $M_i(n_i)$ to run $A_i$ for each $i\in [x]$.
    Then we allocate the space for each algorithm, and finally, we run the algorithms.

   \subparagraph*{Preprocessing phase.} In this phase, the goal is to count the number of input messages $n_i$, and compute the memory $M_{i}(n_i)$ needed to simulate the algorithm $A_i$ for each $i\in[x]$. 
    To count the input messages, we collect the distributed input messages and functions in the first half machines in $O(\log_L M)$ rounds using the sorting algorithm of~\Cref{lem:sorting_prefixsum_MPC}. 
    After this process, we can assume that the input messages belonging to the same algorithm $A_i$ are stored in consecutive machines. 

    We derive the implementation so that each machine having an input message of $A_i$ knows the id of the first machine $\mathcal M_i$ having a message of $A_i$.
    Furthermore, the first machine $\mathcal M_i$ knows the ids of all machines having an input message of $A_i$.
    First, each machine creates a list of the indices of the algorithms for which it holds input messages. In the next two rounds, every machine sends this list to both its previous and next machines.
    After that, each machine $M$ knows if it is the first machine having a message of algorithm $A_i$ for each $i\in [x]$.
    For each algorithm $A_i$, we let $\mathcal M_i$ be the first machine that has a message of $A_i$.
    Note that the last machine that has a message of $A_i$ is $\mathcal M_{i+1}$ or its previous machine, without loss of generality, we suppose that it is $\mathcal M_{i+1}$.
    Note that $\mathcal M_i$ and $\mathcal M_{i+1}$ are the same machine if and only if the messages of $A_i$ are on a single machine. This is a trivial case, so we can assume that $\mathcal M_i\neq \mathcal M_{i+1}$
    Next, we collect the ids of every $\mathcal M_i$ in $x/\LS$ consecutive machines in $O(\log_\LS x)$ rounds by the sorting algorithm, and broadcast the id of $\mathcal M_i$ to the machines having a message of $A_i$, i.e, the machines between $\mathcal M_i$ and $\mathcal M_{i+1}$.
    This takes at most $O(\log_\LS M)$ rounds by~\Cref{lem:sorting_prefixsum_MPC}.
    Precisely, we first send to $\mathcal M_i$ the ids of $\mathcal M_{i+1}$, only if $\mathcal M_i\neq \mathcal M_{i+1}$.
    Then each $\mathcal M_i$ broadcasts its ids to the machines between $\mathcal M_i$ and $\mathcal M_{i+1}$, simultaneously.
    It takes $O(\log_\LS (M/\LS))$ rounds by~\Cref{lem:sorting_prefixsum_MPC}.
    By using the prefix sum algorithm illustrated by \Cref{lem:sorting_prefixsum_MPC}, we can calculate the number $n_i$ of input messages stored in the machines and send it to the machine $\mathcal M_i$ for every $i\in[x]$ in $O(\log_L M)$ rounds. 
    After that, each machine $\mathcal M_i$ calculates the total space $M_i(n_i)$ needed for the algorithm $A_i$.
    In the following, we allocate $M_i(n_i)$ space for each algorithm $A_i$. 

    \subparagraph*{Allocating phase.}In this phase, our goal is to allocate sufficient memory to all algorithms $A_i$. 
    We do this by generating \emph{dummy messages}. Then, we run a similar sorting and communication routine as in the previous phase, but now with not only the input messages, but also these dummy messages.
    To describe this in more detail, recall that, after the \textbf{Preprocessing phase}, the data is stored in the first half of the machines, so the last half of the machines $\mathcal X_1,\ldots, \mathcal X_\ell$ are empty.
    During this phase, we generate dummy messages in $\mathcal X_1,\ldots, \mathcal X_\ell$.
    Then, using the same process to the \textbf{Preprocessing phase},
    we can redistribute them along with the input messages so that for each algorithm $A_i$, its $2M_i(n_i)$ dummy messages and $n_i$ input messages are stored in consecutive machines in $O(\log_\LS M)$ rounds.
    Note that even if $A_i$ is worked in some MPC model with $M_i(n_i)$ total space, but we give additional marginal space for convenience. 
    Furthermore, we can make sure that each machine that has a message of $A_i$ knows the id of the first machine $\mathcal M_i'$ that has a message of $A_i$.
    Additionally, the first machine $\mathcal M_i'$ knows the ids of all machines having a message of $A_i$.

    To generate the dummy space, the first machine $\mathcal X_1$ broadcasts its id to each of the first half of the machines, that includes $\mathcal M_i$'s, using \Cref{lem:sorting_prefixsum_MPC}.
    Recall that $\mathcal M_i$ is the first machine having some message or functions of $A_i$ for $i\in [x]$. This machine knows the value $M_i(n_i)$.
    Then we use the prefix sum algorithm of \Cref{lem:sorting_prefixsum_MPC} again to compute the prefix sum $\sum_{j=1}^{i-1} \cdot M_j(n_j)$. We store this value at the machine $\mathcal M_i$ for each $i\in [x]$.
    In the next round, all machines $\mathcal M_i$ simultaneously do the following:
    \begin{itemize}
        \item Compute $\mathbb M(i)=2\cdot \sum_{j=1}^{i-1} M_j(n_j)+1$ and $\mathbb M(i+1)=2\cdot \sum_{j=1}^{i} M_j(n_j)+1$, and
        \item Send the values $\mathbb M(i)$ and $\mathbb M(i+1)$ to the machine $\mathcal X_{\lfloor\mathbb M(i)/L\rfloor+1}$.
    \end{itemize}
    Note that  $\mathbb M(i+1)=2\cdot \sum_{j=1}^{i} M_j(n_j)+1=\mathbb M(i)+2M_i(n_i)$.
    Then each $\mathcal X_{\lfloor\mathbb M(i)/L\rfloor+1}$ machines for $i\in[x]$ broadcasts, simultaneously, the received values $\mathbb M(i)$ and $\mathbb M(i+1)$ to the consecutive following machines $\mathcal X_{\lfloor\mathbb M(i)/L\rfloor+2},\ldots, \mathcal X_{\lfloor\mathbb M(i+1)/L\rfloor+1}$ by \Cref{lem:sorting_prefixsum_MPC}.
    Finally, each machine $\mathcal X_{j'}$ generates $\max\{L,\mathbb M(i+1)-L\cdot (j'-1)\}$ dummy messages {labeled by $i$} for the algorithm $A_i$ if ${j'}={\lfloor\mathbb M(i)/L\rfloor+j}$ for some $j$.
    Note that it is possible that a single machine $\mathcal X_{j'}$ generates dummy elements for several algorithms without exceeding the local space by construction.
    At the termination of the process, the last half of the machines $\mathcal X_1,\ldots, \mathcal X_\ell$ store $2M_i(n_i)$ dummy messages for all algorithms $\{A_i\}_i$.

    Using the same process to the \textbf{Preprocessing phase},
    we can redistribute the dummy messages and the input messages in $O(\log_\LS M)$ rounds so that for each algorithm $A_i$, its $2M_i(n_i)$ dummy messages and $n_i$ input messages are stored in the consecutive machines.
    Furthermore, we can make each machine having a message of $A_i$ know the id of the first machine $\mathcal M_i'$ having a message of $A_i$.
    Additionally, the first machine $\mathcal M_i'$ knows the ids of all machines having an input message of $A_i$ along with the id of $\mathcal M_{i+1}'$. 
    Therefore, we are ready to run the algorithms $A_i$'s.

    \subparagraph*{Running algorithms.}
    Let $\mathcal M_i'$ be the first machine storing a message belonging to the algorithm $A_i$.
    Note that $\mathcal M_i'=\mathcal M_{i+1}'$ if and only if the machine $\mathcal M_i'$ stores all the messages of $A_i$.
    In such a case, we can run the algorithm $A_i$ by using the single machine $\mathcal M_i'$ by using the input messages and the memory of its dummy messages allocated to $A_i$.
    Therefore, we suppose that such algorithms $A_i$'s are already simulated in $O(1)$ rounds.
    In the following, we simulate the remaining algorithms $A_i$'s of which $\mathcal M_i'\neq \mathcal M_{i+1}'$.

    For a fixed index $i$, recall that the consecutive machines $\mathcal N_1,\ldots,\mathcal N_{N_i}$ between $\mathcal M_i'$ and $\mathcal M_{i+1}'$ (including $\mathcal M_i'$) has the messages only for the algorithm $A_i$, where $N_i=\lceil(2M_i(n_i)+n_i)/L\rceil$
    Additionally, since  $\mathcal M_i'$ knows the id of $\mathcal M_{i+1}'$, each $\mathcal M_i'$ can broadcast the ids of $\mathcal M_i'$ and $\mathcal M_{i+1}'$ to $\mathcal N_1,\ldots,\mathcal N_{N_i}$.
    That means that $\mathcal N_1,\ldots,\mathcal N_{N_i}$ know the ids of themselves that have sufficient total space to run $A_i$.
    In the following, our goal is to simulate $A_i$ by using the consecutive $N_i=\lceil(2M_i(n_i)+n_i)/L\rceil$ machines between $\mathcal M_i'$ and $\mathcal M_{i+1}'$ in $O(T_i(n_i))$ rounds. This ensures that we can simultaneously run the algorithms $A_i$'s in $O(\max_i(T_i))$ rounds.

    To simulate the algorithm $A_i$, we start by removing all of the dummy messages in  $\mathcal N_1,\ldots,\mathcal N_{N_i}$.
    Note that the algorithm $A_i$ is implemented for an MPC model with $L_i(n)$ local space and $M_i(n)$ total space using $\lfloor M_i(n)/L_i(n)\rfloor$ machines.
    Here, each of  $\mathcal N_1,\ldots,\mathcal N_{N_i}$ has $L\geq L_i(n_i)$ local space.
    Then we split the $L$ local space of each machine $\mathcal N_j$ for$j\in[N_i]$ into $\lceil L/L_i(n_i)\rceil$ parts $\mathcal N_j^1,\ldots, \mathcal N_j^{\lfloor L/L_i(n_i)\rfloor}$ each of which has $L_i(n_i)$ space by ignoring the remaining $L-L_i(n_i)\cdot\lfloor L/L_i(n_i)\rfloor$ space.
    Then we can give the different ids to each parts $\mathcal N_j^1,\ldots, \mathcal N_j^{\lfloor L/L_i(n_i)\rfloor}$.
    Note that there are at least $\lfloor M_i(n)/L_i(n)\rfloor$ number of $\mathcal N_j^x$ since we define $\mathcal N_j^x$ for all $x\in [N_i] $ and $j\in[\lfloor L/L_i(n_i)\rfloor]$.
    Therefore, we can run the algorithm $A_i$ within at most $T_i(n_i)$ rounds using the machines $\mathcal N_1,\ldots,\mathcal N_{N_i}$ by treating $\mathcal N_j^x$'s with $x\in[\lfloor L/L_i(n_i)\rfloor]$ and $j\in [N_i]$ like independent (at least) $\lfloor M_i(n)/L_i(n)\rfloor$ machines with $L_i(n_i)$ local space.
    This is because during the implementation of $A_i$, each $\mathcal N_j^x$ sends and receives at most $L_i(n_i)$ messages at the same time.
    That means each machine $\mathcal N_x$ sends and receives at most $L_i(n_i)\cdot\lfloor L/L_i(n_i)\rfloor\leq L$ messages at any point during the implementation.
    This completes the proof of the observation.
\end{proof}
This observation allows for flexibility in how we account for the
machine allocation of individual subroutines. 
Throughout the paper, we ensure
that every algorithm respects the prescribed upper bound on the local space at
all times, along with the total space. The number of machines used by a subroutine is allowed to be omitted as long as the resulting total space remains within the claimed bounds. 
In particular, we do not require every subroutine to be implemented
using exactly $O(m/n)$ machines; some subroutines may use more machines with
smaller local space without affecting the overall guarantees, by the above
observation. 

\subsection{Simulating MPC Algorithms in the Congested Clique}

Using Lenzen's routing~\cite{Lenzen13}, it is quite immediate to see that we can simulate linear MPC algorithms in the Congested Clique. For a formal proof, we refer to Behnezhad, Derakhshan, and Hajiaghayi~\cite{BehnezhadDH18}.

\begin{lemma}[Theorem 3.2 in \cite{BehnezhadDH18}]\label{lem:MPC_to_CC}
    Let $\A$ be an algorithm that for an input graph $G$, with $n$ vertices, runs in $T$ rounds of MPC with $O(n)$ local memory and $O(n^2)$ total memory. One can simulate $\A$ in $O(T )$ rounds of Congested Clique.
\end{lemma}

\subsection{Hash Functions for Derandomization} \label{sec:hash_functions}
\paragraph{$k$-wise independence.} 

\begin{definition}
    For $N, M, k \in N$ such that $d\le N$ , a family of functions $\H = \{h : [N ] \to [M ]\}$ is $k$-wise independent if for all distinct $x_1, x_2, \dots, x_k \in [N ]$, the random variables $H(x_1), \dots,  H(x_k)$ are independent and uniformly distributed in $[M]$ when $H$ is chosen randomly from $\H$.
\end{definition}

See, e.g., \cite{Vadhan12} for an explicit construction of $\H$. 

\begin{lemma}[\cite{Vadhan12} Corollary 3.34]\label{lem:encoding_independent}
    For positive integers $N, M, k$, there is a family of $k$-wise independent hash functions $\H = \{h : [N]\to [M]\}$ such that choosing a random function from $\H$ takes $k(\log N+\log M)$ random bits.
\end{lemma}

For $k$-wise independent function, we have the following tail bound. 
\begin{lemma}[\cite{BellareR94} Lemma 2.3]\label{lem:tail_bound_independent}
    Let $k\ge 4$ be an even integer. Let $X_1, \dots , X_n$ be $k$-wise independent random variables taking values in $[0, 1]$. Let $X = \sum_{i=1}^n X_i$ denote their sum and let $\mu := E[X]$. If $\mu \ge k$, then, for any $\eps>0$, we have
    \begin{equation*}
        \P[ |X-\mu|\geq \eps\mu] \le 8\left( \frac{2k}{\eps^2\mu}\right)^{k/2}.
    \end{equation*}
\end{lemma}

\paragraph{Pseudorandom generator for hitting events.}
We also have the following pseudorandom generator of Gopalan and Yeudayoff~\cite{GopalanY20}, which is particularly designed for hitting events. 
\begin{lemma}[\cite{GopalanY20} Theorem 1.9]\label{lem:hash_functions_for_optimal}
For positive integers $N$ and $M$, and for error parameter $\epsilon >0$, there is a family of hash functions $\H$ from $[N]$ to $[M]$ with seed length $O((\log \log N + \log(M/\eps)) \log \log(M/\eps))$, such that for every subsets $\textbf{I}_1,\dots , \textbf{I}_N \subseteq [M]$, we have:
\begin{equation*}
    \left| \P_{h\sim \H}[\forall i\in [N], h(i)\in \textbf{I}_i] - \P_{g\sim \mathcal U}[\forall i\in [N], g(i)\in \textbf{I}_i]\right| \le \eps,
\end{equation*}
where $\mathcal U$ is the family of all functions from $[N ]$ to $[M ]$. We can sample from $\H$ in polynomial time.
\end{lemma}

\subsection{Derandomization in MPC}
\paragraph{Conditional expectation method.}
In this paper, we utilize a hash function family $\mathcal H$ described in \Cref{lem:tail_bound_independent} or \Cref{lem:hash_functions_for_optimal}, and our goal is to find a good hash function $h$ in $\mathcal H$ efficiently. For this, we apply the \emph{conditional expectations method}, see, e.g., \cite[Chapter 6.3]{mitzenmacher2017probability}.
The method iteratively restricts the hash function family so that it finally returns a desired hash function deterministically.
The following lemma illustrates the method.
\begin{restatable}{lemma}{LemCondExp}\label{lem:MPC_cond_exp}
    We are given a family $\mathcal H$ of hash functions and $\FC$ objective functions $f_1,\ldots,f_\FC\colon \mathcal H\to \mathbb R$ such that:
    \begin{enumerate}
        \item The hash functions of $\mathcal H$ are encoded by $s$-bit random seeds,
        \item $\E_{h\sim \mathcal H}[\sum_{i=1}^{\FC}f_i(h)]\leq x$, and
        \item $\E_{h\sim\mathcal H}[f_i(h)\mid \textnormal{the prefix of } h\in \mathcal H\textnormal{ is } \textsf{r } ]$ can be computed in a single machine with $\LS$ local space for any $i\in [\FC]$ and any prefix string $\textsf{r}\in \{0,1\}^*$ of length at most $s$.
    \end{enumerate}
    There exists a deterministic MPC algorithm that computes a function $h^*\in\mathcal H$ with $\sum_{i=1}^{\FC} f_i(h^*)\leq x$ in $O(s\log \FC/\log^2 \LS)$ rounds with $O(\LS)$ local space and $O(\FC\LS)$ total space.
\end{restatable}
Here, we present the main ideas first, while the details are in the proof.
    This algorithm decomposes the random seeds into $t=O(s/\log \LS)$ chunks $R_1,\ldots, R_t$ of $\lfloor \log \LS\rfloor$ bits random seeds, and iteratively determines each chunk in $O(\log \FC/\log \LS)$ rounds.
    At the $j$-th iteration, we determine the $j$-th chunk $R_j$ by assuming that the prefix $R_1\cdots R_{j-1}$ is fixed in the previous iterations.
    That means that the family of hash functions $\mathcal H$ is restricted to $\mathcal H_{j}$ during the previous $(j-1)$ iterations, and the $j$-th iteration receives $\mathcal H_j$ and restricts it further to $\mathcal H_{j+1}$ by determining the chunk $R_j$. Note that there are at most $\LS$ number of different assignments $\textsf r_1,\ldots, \textsf r_\LS$ for $R_j$, since $R_j$ consists of $\lfloor \log L \rfloor$ bits.

    To do this, we need $\FC$ machines $M^{1},\ldots,M^\FC$ with $2\LS$ local space.
    The machine $M^i$ computes the expectations of $f_i(h)$ over the hash functions $h$ in $\mathcal H_j$ whose $R_j$ chunk is $\textsf r$ for all $\textsf r\in \{\textsf r_1,\ldots,\textsf r_\LS\}$.
    The machine $M^i$ can compute it locally without any communication by using $\LS$ space to compute each expectation value and $\LS$ space to store the computed value.
    Therefore, it is enough to have $2\LS$ local space.
    Furthermore, we can compute the expectation value of $\sum_{i=1}^\FC f_i$ for the prefixes $\textsf r_1,\ldots,\textsf r_\LS$ by summation of the calculated expectation values, which takes $O(\log \FC/\log\LS)$ rounds by \Cref{lem:sorting_prefixsum_MPC}.
    Therefore, determining one chunk takes $O(\log \FC/\log\LS)$ rounds, and thus, the lemma holds.

\newpage
\section{Hitting Sets in MPC}\label{sec:hitting sets}

In this section, we provide our deterministic hitting set algorithms. In particular, we prove \Cref{thm:MPC hitting set intro: linear,thm:MPC hitting set intro: sublinear}. 
In fact, we show more general versions, where the demands on $N$ and $|U|$ are largely dropped. Since these results are rather technical, and we mostly do not need them for our applications, we defer the statements to their respective technical sections. Here, we restate the simpler version \Cref{thm:MPC hitting set intro: linear,thm:MPC hitting set intro: sublinear} for convenience. 

We split our results into two sections: `linear MPC' and `sublinear MPC'. We note that these descriptions are with respect to the hitting set problem \emph{itself} and do not always coincide with the application. To be precise, let the hitting set instance consist of $N$ sets $S_1, \dots, S_N$, where $|S_i|\ge d$. Let the local memory be denoted by $L$. If $L =\Omega(d)$, then we call the instance `linear': a set $S_i$ fits on one machine. Note that we do not set $L=\Theta(d)$, but allow for (much) bigger $L$ as well. As we will demonstrate, such additional space can have additional running time advantages in the derandomization process. 

If $L<d$, we call the instance `sublinear'. In this case, we will always assume $L=O(d^\delta)$ for some constant $\delta \in (0,1)$ -- analogous to the sublinear MPC model for graph problems: here we have $O(n^\delta)$ for some constant $\delta \in (0,1)$, where $n=|V|$ denotes the number of vertices. Now suppose we are faced with a hitting set problem in the context of graphs. If this is linear MPC, i.e., $\Theta(n)$ local space, we will also use the `linear MPC' hitting set. If this in sublinear MPC, i.e., $\Theta(n^\delta)$ local space, it depends on the hitting set instance which hitting set we use: if $d< n^\delta$, we can use the `linear MPC' hitting set. If $d\ge n^\delta$, we use the `sublinear MPC' hitting set. 
When presenting the theorems below, we detail which version is used for each of the applications in \Cref{sec:applications_spanners,sec:applications_shortest_paths}.

Note that in each of these cases, we have that $|U|,N = O(n)$, and that the local space $L$ is at least $\Omega(d^\delta)$ when we apply a hitting set. This means that we always have $|U|,N =\poly L$, as required by  \Cref{thm:MPC hitting set intro: linear,thm:MPC hitting set intro: sublinear}. If we do not have this assumption, we obtain results with additional factors $\log_L |U|$ and $\log_L N$, see \Cref{sec:linear_HS,sec:sublinear_HS}.

\paragraph{Linear MPC Hitting Sets.}

\ThmLinearHittingSet*
\begin{proof}
    We prove part 1 in \Cref{thm:linear MPC hitting set} (\Cref{sec:linear MPC}), part 2 in \Cref{thm:optimal_linear MPC hitting set} (\Cref{sec:opt_lin_hitting_set}). 
 The results as stated here follow from these two lemmas by setting $|U|,N =\poly L$. 
\end{proof}

Additionally, we show in \Cref{sec:linear MPC large N,sec:opt_lin_hitting_set_lareg_N} that we can drop the $N\le |U|$ demand, i.e., we can also deal with the case that the number of sets is very large. 
\tijni{I added a separate lemma for the result with $\log N$}
For part 1, if $N=\poly L$, we obtain a hitting set of size $O\left(\frac{|U|\log_d N}{d^{0.99}}\right)$ in $O(\log_d N)$ rounds. See \Cref{thm:linear MPC hitting set large N} for the formal statement. 
For part 2, if $N=\poly L$, we obtain a hitting set of  size $O\left(\frac{|U|\log N}{d}\right)$ in $O(\log\log (d+N))$ rounds.  See \Cref{thm:optimal_linear MPC hitting set large N} for the formal statement.

\Cref{thm:MPC hitting set intro: linear} Part 1 is used in our spanner in linear MPC (\Cref{sec:unweighted_spanner_linearMPC}), our spanner in sublinear MPC \Cref{sec:weighted_spanner_sublinear}, and APSP in the Congested Clique (\Cref{sec:applications_constant_APSP_CC}). 
\Cref{thm:MPC hitting set intro: linear} Part 2 is used in APSP in the Congested Clique \Cref{sec:applications_constant_APSP_CC}. 
\Cref{thm:optimal_linear MPC hitting set large N} is used in APSP in the near-linear regime (\Cref{sec:applications_heterogeneous}).

\paragraph{Sublinear MPC Hitting Sets.}
Next, we consider the `sublinear' regime. In this case, we have relatively large $d$. In particular, $d$ is larger than the local space, which is here denoted as $O(d^\delta)$ for some constant $\delta\in (0,1)$.


\ThmSublinearHittingSet*
\begin{proof}
    We prove part 1 in \Cref{thm:sublinear_MPC_hitting set} (\Cref{sec:non-optimal hitting set sublinear MPC}) and part 2 in \Cref{thm:optimal_sublinear_MPC_hitting set} (\Cref{sec:optimal hitting set sublinear MPC}).  The results as stated here follow from these two lemmas by setting $|U|,N =\poly d$. 
\end{proof}
\tijni{I changed part 1 to $O(\frac{|U|\log d}{d^{0.99}})$, respectively to $O(\frac{|U|\log N}{d^{0.99}})$ for general $N$ in \Cref{thm:sublinear_MPC_hitting set}. }

\Cref{thm:MPC hitting set intro: sublinear} Part 1 is used in our spanner in sublinear MPC (\Cref{sec:weighted_spanner_sublinear}).
\Cref{thm:MPC hitting set intro: sublinear} Part 2 is used in APSP in the heterogeneous regime (\Cref{sec:applications_heterogeneous}).

\subsection{Preprocessing}
The given input for a hitting set problem is as follows: a universe $U$, a number $d\le |U|$ and collection of $\HS$ subsets $S_1, \dots, S_{\HS} \subseteq U$ with $|S_i|\ge d$ for each $i\in [\HS]$. This input might be larger than necessary in two ways: 1) Some sets $S_i$ might be much larger than $d$. Although this can only be helpful, our algorithm (and its randomized predecessor) do not exploit this. Hence, we can discard any such elements and prune each $S_i$ to have size exactly $|S_i|=d$.
2) The set $U$ might contain elements that are in none of the (pruned sets) $S_i$. Such elements never have to be part of the hitting set, and hence can be discarded. A pruned set $U$ should have size $O(Nd)$ at most.

The following preprocessing lemma formalizes this pruning procedure. 

\begin{lemma}
\label{lm:preprocessing}
     Let $L$ be a parameter. 
     There exists a deterministic MPC algorithm that, given a universe $U$, a number $d\le |U|$ and collection of $\HS$ subsets $S_1, \dots, S_{\HS} \subseteq U$ with $|S_i|\ge d$ for each $i\in [\HS]$ and total input size $M:=|U|+ \sum_{i=1}^\HS |S_i|$, outputs a  universe $\tilde U$ and a collection of $\HS$ subsets $\tilde S_1, \dots, \tilde S_{\HS} \subseteq U$, such that $|\tilde U|\le Nd$, and for each $i\in [\HS]$ $\tilde S_i\subseteq S_i$ and $|\tilde S_i|= d$. The algorithm takes $O(\log_L M)$ rounds, and uses $O(L)$ local space and $O(M)$ total space. 
    We can assume that the new instance is structured as follow: 
    \begin{itemize}
        \item $U$ is stored on $O(Nd/L)$ consecutive machines. 
        \item If $L\ge d$, each set $S_i$ is stored within one machine (with possibly multiple sets in the same machine). 
        \item If $L<d$, the sets $S_i$ are stored on multiple, consecutive machines (where each machine contains a subset of $S_i$ for exactly one index $i$).  
    \end{itemize}
\end{lemma}
\begin{proof}
    First, we prune the sets $S_i$: 
    Whenever a set $S_i$ is larger than $d$, we want to omit further elements. Hereto, we sort all elements using \Cref{lem:sorting_prefixsum_MPC} to each set $S_i$ on consecutive machines. We can count the elements in each set $S_i$ simultaneously again using \Cref{lem:sorting_prefixsum_MPC} and make this value known to all relevant machines. Then we delete the largest (in the ordering) $|S_i|-d$ elements for set $S_i$ locally. We can perform this for all sets $S_i$ simulatneously due to \Cref{obs:parallelism}.
    Afterwards, we sort again to structure the output as required.
    This takes $O(M)$ space and $O(\log_L M)$ rounds. 

    Next, we prune the set $U$:
    In another $O(\log_L M)$ rounds and with $O(M)$ space, we check which elements of $U$ appear in \emph{any} set $S_i$. This is done by sorting the tuples $\{(u,1) : u\in \bigcup_i S_i\}$ and $\{(u,0) : u\in U\}$, using \Cref{lem:sorting_prefixsum_MPC}. When an element $(u,0)$ does not neighbor any elements $(u,1)$, it is deleted. Then, we sort once more to store the remaining elements $(u,0)$ on $O(Nd/L)$ consecutive machines.  
\end{proof}

\tijni{new}
Furthermore, we can assume that $d$ is bounded away from $|U|$. A technicality needed to apply \Cref{lem:tail_bound_independent}. 
\begin{remark}\label{rem:large_d}
    Without loss of generality, we can assume that $d \le |U|/200$ when computing a hitting set of size $O\left(\tfrac{|U|}{d^{0.99}}\right)$. 
\end{remark}
\begin{proof}
    If $|U|/211< d \le |U|$. We can set $d'= |U|/211$. Clearly $|S_i| \ge d \ge d'$. This instance provides a hitting set parametrized by $d'$, which produces a hitting set of size  $O\left(\tfrac{|U|}{d'^{0.99}}\right)\le 211^{0.99}\cdot O\left(\tfrac{|U|}{d^{0.99}}\right)=O\left(\tfrac{|U|}{d^{0.99}}\right)$. 
\end{proof}

 \subsection{Linear MPC}\label{sec:linear_HS}
For the linear MPC hitting set problem, we provide four algorithms, outputting a hitting set of size $O(|U|/d^{0.99}+ N/d)$, $O(|U|\log_d \tfrac{N}{|U|}/d^{0.99})$, $O(|U|\log d/d+ N/d)$, and $O(|U|\log N/d)$ respectfully. 
We note that for the first two results, in \Cref{thm:MPC hitting set intro: linear} we assumed that $N\le |U|$, which means the latter term is absorbed in the first. For full generality, we drop this requirement in this section. In our applications, we either have $N\le |U|$ or we use the third result. 


\subsubsection{Fast Hitting Sets}\label{sec:linear MPC}
The goal of this section is to prove the following lemma. This lemma computes a hitting set that is of size $O(|U|/d^{0.99}+N/d)$. This is roughly a factor $d^{0.01}$ larger than our other hitting sets. However, this allows us to use hash functions with a very small seed. This translates in a low round complexity: for $|U|,N = O(\poly d)$ it only takes a constant number of rounds. 

 \begin{restatable}{lemma}{ThmLinearMPCHittingSet}\label{thm:linear MPC hitting set}
    There exists a deterministic MPC algorithm that, given a universe $U$, a number $d\le |U|$, and a collection of $\HS$ subsets $S_1, \dots, S_{\HS} \subseteq U$ with $|S_i|\ge d$ for each $i\in [\HS]$, computes a hitting set $D$ of size $ O(|U|/d^{0.99}+N/d)$. 
    Let $L\ge d$ denote the local space of each machine, and let $M:=L+|U|+ \sum_{i=1}^\HS |S_i|$. Then the algorithm takes $O(\log_L^2 N+\log_L M)$ rounds and $O(M)$ total space. 
\end{restatable}

The main idea is to derandomize random sampling. The simplest such version is given in the following algorithm.
\begin{algorithm}[h!]
\caption{Constructing an (incomplete) hitting set $D$ with random sampling}\label{alg:HS_simple_sampling}
\begin{algorithmic}[1]
\State $D \gets \emptyset$
\ForAll{$u \in U$}
    \State Add $u$ to $D$ with probability $p = 1/d^{0.99}$
\EndFor
\Return $D$
\end{algorithmic}
\end{algorithm}

However, this does not necessarily give a hitting set: Consider some set $S_i$ of size $d$. The probability that none of the elements of $S_i$ are sampled is 
\begin{align*}
    (1-p)^{|S_i|}= (1-d^{-0.99})^d \le e^{-d^{0.01}}.
\end{align*}
This means that for small $d$ ($d\le \Theta(\log N)$), we do not obtain a result with high probability. There are two possible solutions: (i) increase the sampling probability, or (ii) add additional elements to $D$. 
Solution (i) will be explored in \Cref{sec:opt_lin_hitting_set_lareg_N}. Here, we consider Solution (ii). The random algorithm becomes a two phase process, see \Cref{alg:HS_sampling}: first we add elements to $D$ by random sampling, then we add an element to $D$ for each set $S_i$ that is not hit by the first part. 

\begin{algorithm}
\caption{Constructing a hitting set $D$ with \textbf{independent} random sampling}\label{alg:HS_sampling}
\begin{algorithmic}[1]
\State $D \gets \emptyset$
\ForAll{$u \in U$}
    \State Add $u$ to $D$ with probability $p = 1/d^{0.99}$
\EndFor
\For{$i = 1$ to $N$ in parallel}
    \If{$S_i \cap D = \emptyset$}
        \State Add an arbitrary element of $S_i$ to $D$
    \EndIf
\EndFor
\Return $D$
\end{algorithmic}
\end{algorithm}

Now the set $D$ is guaranteed to be a hitting set. However, its size becomes less obvious. In \Cref{alg:HS_simple_sampling}, the expected size of $D=O(|U|/d^{0.99})$ by linearity of expectation. In \Cref{alg:HS_sampling}, we have a contribution of $O(|U|/d^{0.99})$ from the first phase \emph{and} a contribution from the second phase. The contribution of the second phase is $O(N/e^{d^{0.01}})$. 
Of course, our goal is to provide a deterministic algorithm. Our approach has two ingredients: 1) we reduce the randomness by using $O(1)$-wise independent hash functions, and 2) we deterministically choose the (short) `random' seed of this hash function using the method of conditional expectation. 
The fact that we only have $O(1)$-wise independence rather than full independence will change our probability analysis. We show that we can still bound the contribution of the second phase by $O(N/d)$.

\paragraph*{Ingredient 1: Reducing randomness via $O(1)$-wise independent hash functions.}
We let $\mathcal H$ be a family of $200$-wise independent hash functions $h: U\to [\lceil d^{0.99}\rceil]$.
Each hash function $h$ corresponds to two sets of elements $D_h,D'_h\subseteq U$. $D_h$ corresponds to the elements sampled directly in \Cref{alg:HS_sampling}. $D_h'$ covers all remaining sets $S_i$, it corresponds to  the elements added in the second phase of \Cref{alg:HS_sampling}. In the case of the hitting set being a dominating set in the graph, you can think of this as a node selecting itself, because non of its neighbors are in the dominating set. In this general setting, we need to select an arbitrary element of $S_i$, for each $i$ that has $S_i\cap D_h=\emptyset$. Formally, we define $D_h$ and $D_h'$ as:
\begin{align*}
    D_h &:= \{u \in U \mid h(u)=1\}, \text{ and} \\
    D'_h &:= \{ S_i[1] : S_i\cap D_h=\emptyset, i\in [\HS]\},
\end{align*}
where $S_i[1]$ denotes the first (lexicographically ordered) element of $S_i$.

Note that now $D:=D_h\cup D'_h$ is a hitting set of $S_1, \dots, S_{\HS}$. We denote $h\sim \H$ for taking a hash function $h\in \mathcal H$ according to the corresponding contribution. For each of the hash functions in this paper, $h\sim \H$ means we take $h\in \H$ uniformly at random. \tijn{new (reviewer C)}
By the following claim, the expectation of such a sampled hash function $h\sim \H$ corresponds to a hitting set of the desired size.

\begin{claim}\label{claim:size of S and U}
    Let $\mathcal H$ be a family of $200$-wise independent hash functions $h: U\to [\lceil d^{0.99}\rceil]$.
    For $h\sim \H$, the expected size of $D_h$ and $D'_h$ are $O(|U|/d^{0.99})$ and $O(\HS/d)$, respectively.
\end{claim}
\begin{proof}
    We first show that the expected size of $D_h$ equals $|U|/d^{0.99}$.
    For an element $u\in U$, we let $X_u$ be the random variable such that $X_u=1$ if $h(u)=1$, otherwise $X_u= 0$. Then the expected size of $D_h$ over $\H$ is same as the expectation of $\sum_{u\in U}X_u$. 
    By the linearity of expectations, $\E_{h\sim \H}[\sum_{u\in U} X_u]$ equals $|U|/d^{0.99}$, which implies the first claim.
    In the following, we bound the expected size of $D'_h$.

    We fix an index $i\in [\HS]$, and we bound the probability that it contributed to $D'_h$, i.e., $S_i \cap D_h=\emptyset$.
    Since the $X_u$'s are $200$-wise independent random variables, \Cref{lem:tail_bound_independent} implies the following bound:
    \begin{align}
        \P_{h\sim\H}[S_i \cap D_h=\emptyset]=\P_{h\sim\H}\left[ \sum_{u\in S_i} X_u=0 \right]&\leq \P_{h\sim\H}\left[ \left|\sum_{u\in S_i} X_u-\mu\right|\geq \mu \right] \nonumber \\
        &\leq 8 \left(\frac{2\cdot 200}{\mu}\right)^{200/2} 
        \leq O\left( \frac{1}{\mu^{100}}\right)
        \le O\left (\frac{1}{d} \right ),\label{eq: prob_undominated}
    \end{align}
    where $\mu\coloneqq \E\left[\sum_{u\in S_i}X_u\right]\geq d^{0.01}$ by the linearity of expectations. \tijn{next sentence is new}To apply \Cref{lem:tail_bound_independent}, we also need to satisfy $\mu \le k$. Since $\mu= |U|/d^{0.99}$ and we can assume that $d\le |U|/211$ by \Cref{rem:large_d}, we see that $\mu \ge |U| / (|U|/211)^{0.99}= 211^{0.99} >200$.

    For each $i\in [\HS]$, we let $Y_i$ be the indicator random variable for the event that $S_i \cap D_h=\emptyset$.
    Then $\E_{h\sim \H}[Y_i]\leq O(1/d)$ by Inequality~\ref{eq: prob_undominated}.
    Therefore, the expected size of $D'_h$ is at most $O(N/d)$ by the linearity of expectations, which completes the proof. 
\end{proof}

\paragraph{Ingredient 2: Derandomization via the method of conditional expectation.}
To deterministically obtain a good $O(1)$-wise independent hash function, we use the method of conditional expectation, \Cref{lem:MPC_cond_exp}.

We know that the $200$-wise independent hash functions from $h\colon [|U|] \to [\lceil d^{0.99}\rceil]$ need $200(\log |U|+\log \lceil d^{0.99}\rceil)=O(\log|U|)$ random bits by \Cref{lem:encoding_independent}. Next, we define objective functions that correspond to the expectations of the previous paragraph. 

We are describing the sizes of $D_h$ and $D_h'$ using objective functions. For each of these, we define its own set of objective functions, $f_i$ and $g_i$, respectively. \tijn{Made the definitions of $f_i$, $g_i$ more precise (reviewer C). will not do it in the other proofs that looks similar because the statement is ugly} In total, we have $k= |U|/\LS+O(Nd/\LS)$ functions. 

First, consider the elements of $U$ stored on $|U|/\LS$ machines. Let $f_1, \dots f_{|U|/\LS}$ denote the number of sampled elements on each such machine, i.e., if $U_i$ denotes the elements stored on machine $i$, then $f_i(h) =|\{ u: h(u)=1\}|$. 
Hence, we can write the size of the sampled elements as $D_h = \sum_{i=1}^{|U|/\LS}f_i(h)$. 

Second, the $\HS$ sets $S_1,\dots S_\HS$ are stored on $O(Nd/\LS)$ machines $M_1, \dots, M_{O(Nd/\LS)}$ so that no two machines store elements from the same set $S_i$. On each such machine $M_i$ let $g_{i}$ denote how many sets $S_j$ stored in $M_i$ have $S_j\cap D_h=\emptyset$, i.e., if $S_{\tfrac{L}{d}(i-1)}, \dots S_{\tfrac{L}{d}i-1}$ are stored on machine $i$, then $g_i(h) = \left|\left\{ j \in \left[\tfrac{L}{d}(i-1),\tfrac{L}{d}i\right) : S_j \cap D_h =\emptyset\right\}\right|
$. 
Hence, we can write the size of the elements from the second phase as $D'_h = \sum_{i=1}^{O(Nd/\LS)}g_i(h)$. 

\begin{claim}\label{claim:internal_computation}
    The values $\E_{h\sim\mathcal H}[f_i(h)\mid \textnormal{the prefix of } h\in \mathcal H\textnormal{ is } \textsf{r } ]$  and  $\E_{h\sim\mathcal H}[g_i(h)\mid \textnormal{the prefix of } h\in \mathcal H\textnormal{ is } \textsf{r } ]$ can be computed in a single machine with $\LS$ local space, for any $i$ where $f_i$ or $g_i$ is defined and any $\mathsf r\in\{0,1\}^*$. 
\end{claim}
\begin{proof}
    The function $f_i(h)$ denotes how many elements of $U$ on machine $M_i$ are sampled. For a fixed $h\colon U \to [\lceil 1/p\rceil ]$, the machine can check internally which elements are sampled (encoded by $h(u)=1)$), by counting this for every possible $h\in \H$ with prefix $r$, and dividing out by the total number of such $h$, we obtain $\E_{h\sim\mathcal H}[f_i(h)\mid \textnormal{the prefix of } h\in \mathcal H\textnormal{ is } \textsf{r } ]$. Note that since we only consider one function $h$ at the time, and we can consider $O(L)$ elements of $U$, we require only $L$ local space. 

   The function $g_i(h)$ denotes how many sets $S_j$ stored on machine $M_i$ satisfy $S_j \cap D_h=\emptyset$. Again, for fixed $h$ we can compute this value. For each set $S_j$, we see if $h(u)=1$ for at least one $u\in S_j$. If not, it contributed to $g_i(h)$. As before, $M_i$ sums the values $g_i(h)$ for every possible $h\in \H$ with prefix $r$ and divides out by the total number of such $h$ to obtain $\E_{h\sim\mathcal H}[g_i(h)\mid \textnormal{the prefix of } h\in \mathcal H\textnormal{ is } \textsf{r } ]$. Note that since we only consider one function $h$ at the time, and we can consider $O(L)$ elements of $U$, we require only $L$ local space.
\end{proof}

\paragraph{Putting everything together.}
Finally, we use the notation and results obtained above to prove \Cref{thm:linear MPC hitting set}.

\begin{proof}[Proof of \Cref{thm:linear MPC hitting set}]
    First, we use \Cref{lm:preprocessing} to reduce the size of the data we work with to $O(Nd)$ in $O(\log_L M)$ rounds. This also ensures that $|U|\le Nd$ and $U$ is stored on $O(Nd/L)$ consecutive machines and each set $S_i$ is stored within one machine (with $\lfloor L/d\rfloor $ such sets in the same machine).    

    We set $s= O(\log |U|)$ and $k= |U|/\LS+O(Nd/\LS)$. Then we see that:
    \begin{enumerate}
        \item The hash functions of $\H$ are encoded by $s$-bit random seeds by \Cref{lem:encoding_independent};
        \item $\E_{h\sim\H}\left[ \sum_{i=1}^{|U|/L} f_i(h)+\sum_{i=1}^{O(Nd/L)} g_i(h) \right] \le O(|U|/d^{0.99}+N/d)$ by \Cref{claim:size of S and U};
        \item The expectation of each objective function $f_i$ and $g_i$ can be computed in a single machine with $\LS$ local space for any prefix $\mathsf r\in\{0,1\}^*$ by \Cref{claim:internal_computation}.
    \end{enumerate}

    This means we satisfy the demands of \Cref{lem:MPC_cond_exp}. It then follows that we can compute a hitting set $D$ of size $O(|U|/d^{0.99}+N/d)$ in $O(s\log k/\log^2 L)=O(\log |U| \log (|U|/\LS+Nd/\LS)/ \log^2 L)=O(\log_L^2 N)$ rounds, using that $|U|=O(Nd)$ and $d\le L$. 
\end{proof}

\subsubsection{Fast Hitting Sets for large \texorpdfstring{$N$}{N}}\label{sec:linear MPC large N}
\tijni{this subsection and result is new}

In the previous section, we noted that for small $d$ ($d\le \Theta(\log N)$), sampling with probability $1/d^{0.99}$ does not give with high probability results. We overcame this by adding additional elements. Alternatively, we can increase the sampling probability to $p=400\cdot(1+ \log_d \tfrac{N}{|U|})/d^{0.99}$, which gives the following results.

\begin{restatable}{lemma}{ThmLinearMPCHittingSetLargeN}\label{thm:linear MPC hitting set large N}
    There exists a deterministic MPC algorithm that, given a universe $U$, a number $d\le |U|$, and a collection of $\HS$ subsets $S_1, \dots, S_{\HS} \subseteq U$ with $|S_i|\ge d$ for each $i\in [\HS]$, computes a hitting set $D$ of size $ O(|U|\log_d \tfrac{N}{|U|}/d^{0.99})$. 
    Let $L\ge d$ denote the local space of each machine, and let $M:=L+|U|+ \sum_{i=1}^\HS |S_i|$. Then the algorithm takes $O(\log_d \tfrac{N}{|U|}\cdot \log_L^2 N+\log_L M)$ rounds and $O(M)$ total space. 
\end{restatable}

The proof is analogous to \Cref{thm:linear MPC hitting set}, with only the probability increased. We follow the same notation, and proving the adaptations.

\paragraph*{Ingredient 1: Reducing randomness via $O\left(\log_d \tfrac{N}{|U|}\right)$-wise independent hash functions.}
First, we need a new version of \Cref{claim:size of S and U}, giving the expected size of the dominating set using the increased sampling probability. 

\begin{claim}\label{claim:size of S and U large N}
    Let $\mathcal H$ be a family of $200\left(1+\log_d \tfrac{N}{|U|}\right)$-wise independent hash functions $h: U\to [\lceil d^{0.99}/
    (400\cdot \log_d \tfrac{N}{|U|})\rceil]$.
    For $h\sim \H$, the expected size of $D_h$ and $D'_h$ are $O(|U|\log_d \tfrac{N}{|U|}/d^{0.99})$ and $O(|U|/d)$, respectively.
\end{claim}
\begin{proof}
    We first show that the expected size of $D_h$ equals $400\cdot |U|\cdot (1+\log_d \tfrac{N}{|U|})/d^{0.99}$.
    For an element $u\in U$, we let $X_u$ be the random variable such that $X_u=1$ if $h(u)=1$, otherwise $X_u= 0$. Then the expected size of $D_h$ over $\H$ is same as the expectation of $\sum_{u\in U}X_u$. 
    By the linearity of expectations, $\E_{h\sim \H}[\sum_{u\in U} X_u]$ equals $400\cdot |U|\cdot \left(1+\log_d \tfrac{N}{|U|}\right)/d^{0.99}$, which implies the first claim.
    In the following, we bound the expected size of $D'_h$.

    We fix an index $i\in [\HS]$, and we bound the probability that it contributed to $D'_h$, i.e., $S_i \cap D_h=\emptyset$.
    Since the $X_u$'s are $200(1+\log_d \tfrac{N}{|U|})$-wise independent random variables, \Cref{lem:tail_bound_independent} implies the following bound:
    \begin{align}
        \P_{h\sim\H}[S_i \cap D_h=\emptyset]=\P_{h\sim\H}\left[ \sum_{u\in S_i} X_u=0 \right]&\leq \P_{h\sim\H}\left[ \left|\sum_{u\in S_i} X_u-\mu\right|\geq \mu \right] \nonumber \\
        &\leq 8 \left(\frac{2\cdot 200\cdot (1+\log_d \tfrac{N}{|U|})}{\mu}\right)^{200/2\cdot \left(1+\log_d \tfrac{N}{|U|}\right)} \\
        &\leq 8 \left(\frac{2\cdot 200\cdot (1+\log_d\tfrac{N}{|U|})}{400\cdot (1+\log_d\tfrac{N}{|U|})\cdot  d^{0.01}}\right)^{200/2\cdot \left(1+\log_d \tfrac{N}{|U|}\right)} \\
        &\le O\left (\frac{|U|}{Nd} \right ),\label{eq: prob_undominated large N}
    \end{align}
    where $\mu\coloneqq \E\left[\sum_{u\in S_i}X_u\right]\geq 400\cdot \left(1+\log_d \tfrac{N}{|U|}\right)\cdot d^{0.01}$ by the linearity of expectations. To apply \Cref{lem:tail_bound_independent}, we also need to satisfy $\mu \le k$, which holds since $\mu= 400\cdot |U|\cdot \left(1+\log_d \tfrac{N}{|U|}\right)/d^{0.99}\ge 200 \cdot \left(1+\log_d \tfrac{N}{|U|}\right)$.

    For each $i\in [\HS]$, we let $Y_i$ be the indicator random variable for the event that $S_i \cap D_h=\emptyset$.
    Then $\E_{h\sim \H}[Y_i]\leq O(1/N)$ by Inequality~\ref{eq: prob_undominated large N}.
    Therefore, the expected size of $D'_h$ is at most $O(|U|/d)$ by the linearity of expectations, which completes the proof. 
\end{proof}

\paragraph{Ingredient 2: Derandomization via the method of conditional expectation.}
We use the method of conditional expectation as in \cref{sec:linear MPC}. In particular, \Cref{claim:internal_computation} holds since we are in the same setting. 

We know that the $200\left(1+\log_d \tfrac{N}{|U|}\right)$-wise independent hash functions from $h\colon [|U|] \to [\lceil d^{0.99}/(400\cdot\left(1+\log_d \tfrac{N}{|U|}\right))\rceil]$ need $200\cdot \left(1+\log_d \tfrac{N}{|U|}\right)(\log |U|+\log \lceil d^{0.99}\rceil)=O(\log_d \tfrac{N}{|U|}\cdot \log|U|)$ random bits by \Cref{lem:encoding_independent}.

\paragraph{Putting everything together.}
Now, we use the notation and results obtained above to prove \Cref{thm:linear MPC hitting set large N}.

\begin{proof}[Proof of \Cref{thm:linear MPC hitting set large N}]
    First, we use \Cref{lm:preprocessing} to reduce the size of the data we work with to $O(Nd)$ in $O(\log_L M)$ rounds. This also ensures that $|U|\le Nd$ and $U$ is stored on $O(Nd/L)$ consecutive machines and each set $S_i$ is stored within one machine (with $\lfloor L/d\rfloor $ such sets in the same machine).    

    We set $s= O(\log_d \tfrac{N}{|U|} \cdot \log |U|)$ and $k= |U|/\LS+O(Nd/\LS)$. Then we see that:
    \begin{enumerate}
        \item The hash functions of $\H$ are encoded by $s$-bit random seeds by \Cref{lem:encoding_independent};
        \item $\E_{h\sim\H}\left[ \sum_{i=1}^{|U|/L} f_i(h)+\sum_{i=1}^{O(Nd/L)} g_i(h) \right] \le O(|U|\log_d \tfrac{N}{|U|}/d^{0.99})$ by \Cref{claim:size of S and U large N};
        \item The expectation of each objective function $f_i$ and $g_i$ can be computed in a single machine with $\LS$ local space for any prefix $\mathsf r\in\{0,1\}^*$ by \Cref{claim:internal_computation}.
    \end{enumerate}

    This means we satisfy the demands of \Cref{lem:MPC_cond_exp}. It then follows that we can compute a hitting set $D$ of size $O(|U|\log_d \tfrac{N}{|U|}/d^{0.99})$ in $O(s\log k/\log^2 L)=O(\log_d \tfrac{N}{|U|} \cdot \log |U| \log (|U|/\LS+Nd/\LS)/ \log^2 L)=O(\log_d \tfrac{N}{|U|}\cdot \log_L^2 N)$ rounds, using that $|U|=O(Nd)$ and $d\le L$. 
\end{proof}

\subsubsection{Hitting Sets of Size \texorpdfstring{$O\left(\tfrac{|U|\log d}{d}+N/d\right)$}{}}
\label{sec:opt_lin_hitting_set}
In the following two sections, we obtain a smaller-sized hitting set by relaxing the round complexity. First, we obtain a version that is efficient for ``small'' $N$, i.e., $N=O(|U|)$. If $N$ is larger, we obtain a similar result in \Cref{sec:opt_lin_hitting_set_lareg_N} at the cost of a factor $\log N$ instead of a factor $\log d$ in the size.

The result of this section is summarized in the following lemma. 

\begin{restatable}{lemma}{ThmOptimalLinearMPCHittingSet}\label{thm:optimal_linear MPC hitting set}
    There exists a deterministic MPC algorithm that, given a universe $U$, a number $d\le |U|$, and a collection of $\HS$ subsets $S_1, \dots, S_{\HS} \subseteq U$ with $|S_i|\ge d$ for each $i\in [\HS]$, computes a hitting set $D$ of size $ O\left(\tfrac{|U|\log d}{d}+\HS/d\right)$. 
    Let $L\ge d$ denote the local space of each machine, and let $M:=L+|U|+ \sum_{i=1}^\HS |S_i|$. Then the algorithm takes $O(\log \log d \cdot \log_L \log N \cdot \log_L N+\log_L M)$ rounds and $O(M)$ total space. 
\end{restatable}

The algorithm is analogous to the algorithm in the proof of Theorem~\ref{thm:linear MPC hitting set}.
Instead of derandomizing random sampling with $p=1/d^{0.99}$, we derandomize sampling with probability $p=2\cdot \log d/d$.
Furthermore, we use \Cref{lem:hash_functions_for_optimal} instead of $O(1)$-wise independence, which gives a family of hash functions that \emph{approximates} random sampling. This has been done before by, e.g., \cite[Appendix B]{bezdrighinEG+2022}.

\paragraph{Reducing randomness via pseudorandom generators.}
We consider the random process where each element~$u\in U$ adds itself to a hitting set $D$ with probability $p=2\log d/d$.
We show how to obtain this random process with a family~$\H$ of hash functions $h: U\to [\lceil1/p\rceil]$ from \Cref{lem:hash_functions_for_optimal} with $\eps=d^{-c}$ for some sufficiently large constant $c$. 
Analogous to \Cref{thm:linear MPC hitting set}, each hash function $h\in\H$ corresponds to two element sets 
\begin{align*}
    D_h &:= \{u\in U\mid h(u)=1\}, \text{ and}\\
     D'_h &:= \{ S_i[1] : S_i\cap D_h=\emptyset, i\in [\HS]\},
\end{align*}
where $S_i[1]$ denotes the first (lexicographically ordered) element of $S_i$.

\begin{claim}\label{claim:opt_hitting_set}
     For $h\sim \H$, the expected size of $D_h$ and $D'_h$ are $O(|U|\log d/d)$ and $O(\HS/d)$, respectively.
\end{claim}
\begin{proof}
    In the following, we bound the probability of $u\in D_h$ for an element $u\in U$ and the probability that $S_i\cap D'_h=\emptyset$ for an index $i\in [\HS]$ and thus $S_i[1]\in D'_h$.

    The hash functions in $\H$ map each element to integers in $[\lceil1/p\rceil]$ according to \Cref{lem:hash_functions_for_optimal}.
    Furthermore, an element $u\in U$ is in $D_h$ if and only if $h(u)=1$.
    Therefore, for an element $u\in U$, the probability of $u\in D_h$ can be bound by \Cref{lem:hash_functions_for_optimal} as follows:
    \begin{align*}
        \P_{h\sim \H}[u\in D_h] &= \P_{h\sim \H}[h(u)\in\{1\} \textnormal{ while } h(v)\in[\lceil1/p\rceil]\textnormal{ for all } v\in U\setminus\{u\}]\nonumber \\
        &\le \P_{g\sim \mathcal U}[g(u)\in\{1\} \textnormal{ while } g(v)\in[\lceil1/p\rceil]\textnormal{ for all } v\in U\setminus\{u\}] + \eps \nonumber \\
        &= \P_{g\sim \mathcal U}[g(u)=1]
        \le p+\eps=O\left(\frac{\log d}{d}\right),\label{eq:optimal_linear_prob_D_h}
    \end{align*}
    where $\mathcal U$ denotes all mapping of $U\to [\lceil1/p\rceil]$.
    So $\E_{h\sim \H}[|D_h|]\le |U|(p+\eps)= O(\tfrac{|U|\log d}{d})$.
    
    In the following, we bound the probability that $i\in [N]$ contributed to $D'_h$, i.e., that $S_i\cap D'_h=\emptyset$.
    First of all, we consider all mappings $\mathcal U$ and bound the probability $S_i\cap S'_g=\emptyset$ for $g\sim \mathcal U$.
    That means every element $u\in S_i$ satisfies $g(u)\neq 1$.
    Therefore, the probability is bounded as follows:
    \begin{align*}
       \P_{g\sim \mathcal U}[ S_i\cap S'_g=\emptyset]&=
        \P_{g\sim \mathcal U}[g(u)\in [2,\lceil 1/p\rceil]\textnormal{ for } u\in S_i \textnormal{ while } g(u')\in  [\lceil 1/p\rceil] \textnormal{ for } u'\in U\setminus S_i] \\
        &=  \P_{g\sim \mathcal U}[g(u)\in [2,\lceil 1/p\rceil]\textnormal{ for } u\in S_i] \\
        &\le \left(1-\tfrac{1}{\lceil 1/p\rceil}\right)^{|S_i|}\leq (1-p/2)^d.
    \end{align*}
    \Cref{lem:hash_functions_for_optimal} gives that two probabilities $\P_{h\sim\H}[ S_i\cap D'_h=\emptyset]$ and $\P_{g\sim \mathcal U}[ S_i\cap S'_g=\emptyset]$ are at most $\eps=d^{-c}$ apart. Which means that with high probability, using the hash functions suffices.
    \begin{align*}
        \P_{h\sim \H}[ S_i\cap S'_g=\emptyset] \le (1-p/2)^{d} +\eps \le (1-\tfrac{\log d}{d})^{d}+d^{-c}\le \tfrac{1}{d}+d^{-c}.
    \end{align*}
    So the expected size of $D'_h$ is $\E_{h\sim\H}[|D'_h|]=O(\HS/d)$.
    
    In conclusion, we have that $\E_{h\sim\H}[|D_h\cup D'_h|]=O\left(\tfrac{|U|\log d}{d}+\HS/d\right)$.
\end{proof}

\paragraph{Derandomization via the method of conditional expectation.}

Next, we show that by \Cref{lem:MPC_cond_exp}, we can find a proper hash function deterministically. Hereto, we define the objective functions as before. 

First, consider the elements of $U$ stored on $|U|/\LS$ machines. Let $f_1, \dots f_{|U|/\LS}$ denote the number of sampled elements on each such machine. 
Hence, we can write the size of the sampled elements as $D_h = \sum_{i=1}^{|U|/\LS}f_i(h)$. 

Second, the $\HS$ sets $S_1,\dots S_\HS$ are stored on $O(Nd/\LS)$ machines $M_1, \dots, M_{O(Nd/\LS)}$ so that no two machines store elements from the same set $S_i$. On each such machine $M_i$ let $g_{i}$ denote how many sets $S_j$ stored in $M_i$ have $S_j\cap D_h=\emptyset$. Hence, we can write the size of the elements from the second phase as $D'_h = \sum_{i=1}^{O(Nd/\LS)}g_i(h)$. 

\begin{claim}\label{claim:internal_computation_copy2}
    The values $\E_{h\sim\mathcal H}[f_i(h)\mid \textnormal{the prefix of } h\in \mathcal H\textnormal{ is } \textsf{r } ]$  and  $\E_{h\sim\mathcal H}[g_i(h)\mid \textnormal{the prefix of } h\in \mathcal H\textnormal{ is } \textsf{r } ]$ can be computed in a single machine with $\LS$ local space, for any $i$ where $f_i$ or $g_i$ is defined and any $\mathsf r\in\{0,1\}^*$. 
\end{claim}
\begin{proof}
    As $f_i$ and $g_i$ are defined as in \Cref{claim:internal_computation}, the proof of \Cref{claim:internal_computation} holds again. 
\end{proof}

\paragraph{Putting everything together.}
Finally, we use the notation and results obtained above to prove \Cref{thm:optimal_linear MPC hitting set}.

\begin{proof}[Proof of \Cref{thm:optimal_linear MPC hitting set}]
     Again, we start by applying \Cref{lm:preprocessing} to reduce the size of the data we work with to $O(Nd)$ in $O(\log_L M)$ rounds. This also ensures that $U$ is stored on $O(Nd/L)$ consecutive machines and each set $S_i$ is stored within one machine (with $\lfloor L/d\rfloor $ such sets in the same machine).

     We set $s= O((\log\log |U|+\log d) \log \log d)$ and $k= |U|/\LS+O(Nd/\LS)$. Then we see that:
    \begin{enumerate}
        \item The hash functions of $\H$ are encoded by $s$-bit random seeds by \Cref{lem:hash_functions_for_optimal};
        \item $\E_{h\sim\H}[\sum_{i=1}^{|U|/L} f_i(h)+\sum_{i=1}^{O(Nd/L)} g_i(h)]=O\left(\tfrac{|U|\log d}{d}+\HS/d\right)$ by \Cref{claim:opt_hitting_set};
        \item The expectation of each objective function $f_i$ and $g_i$ can be computed in a single machine with $\LS$ local space for any prefix $\mathsf r\in\{0,1\}^*$ by \Cref{claim:internal_computation_copy2}.
    \end{enumerate}

    This means we satisfy the demands of \Cref{lem:MPC_cond_exp}, hence we can compute a hitting set $D$ of size $O\left(\tfrac{|U|\log d}{d}+\HS/d\right)$ in $O(s\log k/\log^2L)=O((\log\log |U|+\log d) \log \log d\cdot \log N/\log^2 L)= O(\log \log d \cdot \log_L \log N \cdot \log_L N)$ rounds, where we use that $|U|=O(Nd)$ and $d\le L$.
\end{proof}

\subsubsection{Hitting Sets of Size \texorpdfstring{$O\left(\tfrac{|U|\log N}{d}\right)$}{}}\label{sec:opt_lin_hitting_set_lareg_N}
In this section, we obtain a similar hitting set as in \Cref{sec:opt_lin_hitting_set}, with the main difference that we do not have the additive factor $+N/d$. This comes at the cost at a slightly higher sampling probability, that replaces $\log d$ by $\log N$. The result is as follows. 

\begin{restatable}{lemma}{ThmOptimalLinearMPCHittingSetLargeN}\label{thm:optimal_linear MPC hitting set large N}
    There exists a deterministic MPC algorithm that, given a universe $U$, a number $d\le |U|$, and a collection of $\HS$ subsets $S_1, \dots, S_{\HS} \subseteq U$ with $|S_i|\ge d$ for each $i\in [\HS]$, computes a hitting set $D$ of size $ O\left(\tfrac{|U|\log N}{d}\right)$. 
    Let $L\ge d$ denote the local space of each machine, and let $M:=L+|U|+ \sum_{i=1}^\HS |S_i|$. Then the algorithm takes $O(\log \log (d+N) \cdot (\log_L N)^2+\log_L M)$ rounds and $O(M)$ total space. 
\end{restatable}

The algorithm is analogous to the algorithm in the proof of \Cref{thm:optimal_linear MPC hitting set}.
Instead of derandomizing random sampling with $p=2\cdot \log d/d$, we derandomize sampling with probability $p=2\cdot \log N/d$.

\paragraph{Reducing randomness via pseudorandom generators.}
Now, we consider the random process where each element~$u\in U$ adds itself to a hitting set $D$ with probability $p=2\log N/d$.
We show how to obtain this random process with a family~$\H$ of hash functions $h: U\to [\lceil1/p\rceil]$ from \Cref{lem:hash_functions_for_optimal} with $\eps=(d+N)^{-c}$ for some sufficiently large constant $c$. 
Analogous to \Cref{thm:optimal_linear MPC hitting set}, each hash function $h\in\H$ corresponds to two element sets $D_h:=\{u\in U\mid h(u)=1\}$ and $D'_h := \{ S_i[1] : S_i\cap D_h=\emptyset, i\in [\HS]\}$, where $S_i[1]$ denotes the first (lexicographically ordered) element of $S_i$. However, our goal is now to make $D'_h$ much smaller. 

\begin{claim}\label{claim:opt_hitting_set_large_N}
     For $h\sim \H$, the expected size of $D_h$ and $D'_h$ are $O(|U|\log N/d)$ and $O(1)$, respectively.
\end{claim}
\begin{proof}
    In the following, we bound the probability of $u\in D_h$ for an element $u\in U$ and the probability that $S_i\cap D'_h=\emptyset$ for an index $i\in [\HS]$ and thus $S_i[1]\in D'_h$.

    The hash functions in $\H$ map each element to integers in $[\lceil1/p\rceil]$ according to \Cref{lem:hash_functions_for_optimal}.
    Furthermore, an element $u\in U$ is in $D_h$ if and only if $h(u)=1$.
    Therefore, for an element $u\in U$, the probability of $u\in D_h$ can be bound by \Cref{lem:hash_functions_for_optimal} as follows:
    \begin{align*}
        \P_{h\sim \H}[u\in D_h] &= \P_{h\sim \H}[h(u)\in\{1\} \textnormal{ while } h(v)\in[\lceil1/p\rceil]\textnormal{ for all } v\in U\setminus\{u\}]\nonumber \\
        &\le \P_{g\sim \mathcal U}[g(u)\in\{1\} \textnormal{ while } g(v)\in[\lceil1/p\rceil]\textnormal{ for all } v\in U\setminus\{u\}] + \eps \nonumber \\
        &= \P_{g\sim \mathcal U}[g(u)=1]
        \le p+\eps=O\left(\frac{\log N}{d}\right),
    \end{align*}
    where $\mathcal U$ denotes all mapping of $U\to [\lceil1/p\rceil]$.
    So $\E_{h\sim \H}[|D_h|]\le |U|(p+\eps)= O(\tfrac{|U|\log N}{d})$.
    
    In the following, we bound the probability that $i\in [N]$ contributed to $D'_h$, i.e., that $S_i\cap D'_h=\emptyset$.
    First of all, we consider all mappings $\mathcal U$ and bound the probability $S_i\cap S'_g=\emptyset$ for $g\sim \mathcal U$.
    That means every element $u\in S_i$ satisfies $g(u)\neq 1$.
    Therefore, the probability is bounded as follows:
    \begin{align*}
       \P_{g\sim \mathcal U}[ S_i\cap S'_g=\emptyset]&=
        \P_{g\sim \mathcal U}[g(u)\in [2,\lceil 1/p\rceil]\textnormal{ for } u\in S_i \textnormal{ while } g(u')\in  [\lceil 1/p\rceil] \textnormal{ for } u'\in U\setminus S_i] \\
        &=  \P_{g\sim \mathcal U}[g(u)\in [2,\lceil 1/p\rceil]\textnormal{ for } u\in S_i] \\
        &\le \left(1-\tfrac{1}{\lceil 1/p\rceil}\right)^{|S_i|}\leq (1-p/2)^d.
    \end{align*}
    \Cref{lem:hash_functions_for_optimal} gives that two probabilities $\P_{h\sim\H}[ S_i\cap D'_h=\emptyset]$ and $\P_{g\sim \mathcal U}[ S_i\cap S'_g=\emptyset]$ are at most $\eps=(d+N)^{-c}$ apart. Which means that with high probability, using the hash functions suffices.
    \begin{align*}
        \P_{h\sim \H}[ S_i\cap S'_g=\emptyset] \le (1-p/2)^{d} +\eps \le (1-\tfrac{\log N}{d})^{d}+d^{-c}\le \tfrac{1}{N}+(d+N)^{-c}.
    \end{align*}
    So the expected size of $D'_h$ is $\E_{h\sim\H}[|D'_h|]=O(\HS/\HS)=O(1)$.
\end{proof}

\paragraph{Derandomization via the method of conditional expectation.}
Next, we show that by \Cref{lem:MPC_cond_exp}, we can find a proper hash function deterministically. Hereto, we define the objective functions as before. 

First, consider the elements of $U$ stored on $|U|/\LS$ machines. Let $f_1, \dots f_{|U|/\LS}$ denote the number of sampled elements on each such machine. 
Hence, we can write the size of the sampled elements as $D_h = \sum_{i=1}^{|U|/\LS}f_i(h)$. 

Second, the $\HS$ sets $S_1,\dots S_\HS$ are stored on $O(Nd/\LS)$ machines $M_1, \dots, M_{O(Nd/\LS)}$ so that no two machines store elements from the same set $S_i$. On each such machine $M_i$ let $g_{i}$ denote how many sets $S_j$ stored in $M_i$ have $S_j\cap D_h=\emptyset$. Hence, we can write the size of the elements from the second phase as $D'_h = \sum_{i=1}^{O(Nd/\LS)}g_i(h)$. 

\begin{claim}\label{claim:internal_computation_copy}
    The values $\E_{h\sim\mathcal H}[f_i(h)\mid \textnormal{the prefix of } h\in \mathcal H\textnormal{ is } \textsf{r } ]$  and  $\E_{h\sim\mathcal H}[g_i(h)\mid \textnormal{the prefix of } h\in \mathcal H\textnormal{ is } \textsf{r } ]$ can be computed in a single machine with $\LS$ local space, for any $i$ where $f_i$ or $g_i$ is defined and any $\mathsf r\in\{0,1\}^*$. 
\end{claim}
\begin{proof}
    As $f_i$ and $g_i$ are defined as in \Cref{claim:internal_computation}, the proof of \Cref{claim:internal_computation} holds again. 
\end{proof}

\paragraph{Putting everything together.}
Finally, we use the notation and results obtained above to prove \Cref{thm:optimal_linear MPC hitting set large N}.

\begin{proof}[Proof of \Cref{thm:optimal_linear MPC hitting set large N}]
    Again, we start by applying \Cref{lm:preprocessing} to reduce the size of the data we work with to $O(Nd)$ in $O(\log_L M)$ rounds. This also ensures that $U$ is stored on $O(Nd/L)$ consecutive machines and each set $S_i$ is stored within one machine (with $\lfloor L/d\rfloor $ such sets in the same machine). 
     
     We set $s= O((\log\log |U|+\log (d+N)) \log \log (d+N))$ and $k= |U|/\LS+O(Nd/\LS)$. Then we see that:
    \begin{enumerate}
        \item The hash functions of $\H$ are encoded by $s$-bit random seeds by \Cref{lem:hash_functions_for_optimal};
        \item $\E_{h\sim\H}[\sum_{i=1}^{|U|/L} f_i(h)+\sum_{i=1}^{O(Nd/L)} g_i(h)]=O\left(\tfrac{|U|\log N}{d}\right)$ by \Cref{claim:opt_hitting_set_large_N};
        \item The expectation of each objective function $f_i$ and $g_i$ can be computed in a single machine with $\LS$ local space for any prefix $\mathsf r\in\{0,1\}^*$ by \Cref{claim:internal_computation_copy}.
    \end{enumerate}

    This means we satisfy the demands of \Cref{lem:MPC_cond_exp}, hence we can compute a hitting set $D$ of size $O\left(\tfrac{|U|\log N}{d}\right)$ in $O(s\log k/\log^2L)=O((\log\log |U|+\log (d+N)) \log \log (d+N)\cdot \log N/\log^2 L)= O(\log \log (d+N) \cdot \log_L  N \cdot \log_L N)$ rounds, where we use that $|U|=O(Nd)$ and $d\le L$.
\end{proof}


\newpage
\subsection{Sublinear MPC}\label{sec:sublinear_HS}
In this section, we modify the algorithms of \Cref{sec:linear_HS} for the sublinear MPC model with $O(d^{\delta})$ local space for some constant $\delta \in (0,1)$.
Here, we must address several technical subtleties, since we cannot fit a set $S_i$ fits in a single machine: $d^\delta < d$ for $\delta <1$. Note that this also means, that (after sorting) a machine never contains elements from multiple sets $S_i$. 

We first design a sparsification tool in \Cref{sec:sublinear_HS_sparsification}, that reduces the universe to a smaller set $U'$ that 1) still contains a hitting set 2) $S_i \cap U'$ fits on a single machine. Then, we apply this to get hitting sets of size $O(|U|\log N/d^{0.99})$ and $O(|U|\log N/d)$ in \Cref{sec:non-optimal hitting set sublinear MPC,sec:optimal hitting set sublinear MPC} respectfully. 

\tijni{next sentence is new}
Throughout, we assume that $\log_d N \le d^{\delta/4}$, or equivalently $N< d^{d^{\delta/4}}$. When $N=\poly d$, this is automatically satisfied.

\subsubsection{Sparsification Tool}\label{sec:sublinear_HS_sparsification}
The goal of this subsection is to give a routine that does some initial, crude sampling, such that the remaining sets fit in in one machine. 

\begin{lemma}\label{lem:sparsified_instance_HS}
     Let $\delta >0$ be a constant.
         There exists a deterministic MPC algorithm that, given a universe $U$, a number $d\le |U|$ and collection of $\HS<  d^{d^{\delta/4}}$\tijn{changed bound on $N$.} subsets $S_1, \dots, S_{\HS} \subseteq U$ with $|S_i|\ge d$ for each $i\in [\HS]$, computes a sparsified instance $U'\subseteq U$ such that $|U'|=O(|U|/d^{1-\delta})$,  $|S_i \cap U'|=O(d^{\delta})$ and $|S_i\cap U'| \ge d^{\delta}$.
         Let $M:=|U|+ \sum_{i=1}^\HS |S_i|$.
         The algorithm takes $O(\log_d|U|\log_d^2\HS+\log_d M)$ rounds with $O(d^\delta)$ local space and $O(M)$ total space. 
\end{lemma}

The idea is to sample with probability $d^\delta/d=d^{\delta-1}$ to obtain $U'$. In expectation, this immediately gives the bounds as stated. However, we cannot derandomize such an aggressive sampling step, as the objective function cannot be evaluated within one machine. 
Hence we design a recursive sampling scheme, using $\ell$ iterations, for some constant $\ell$ to be determined. The goal is that within each iteration, we can construct objective functions that cannot be evaluated locally. 
We sample with a larger probability $p$, such that $p^\ell= d^{\delta-1}$. If we take $p=1/d^{\delta/2}$, then we need $\ell \cdot \delta/2 = 1-\delta$, or equivalently $\ell=2/\delta -2$. For simplicity, we assume that $\ell=2/\delta -2$ is a positive integer. If not, we set $\ell=\max\{1,\lceil 2/\delta -2\rceil\}$ and set $p=d^{(\delta-1)/\ell}$, which is almost $1/d^{\delta/2}$ and only affects the constants in the proof below. \kyungjin{Rev A: Major 2 comment wanted to specify the $\ell=\max \{1,...\}$}

Now we compute $\ell$ subsets  $U=U^{(0)}\supseteq U^{(1)}\supseteq \ldots\supset U^{(\ell)}=U'$ recursively by sampling each element in $U^{(j)}$ for $U^{(j+1)}$ with probability $p=2/d^{\delta/2}$. Note that this is a factor $2$ bigger than stated previously, which is needed to reach the $d^\delta$ lower bound, \emph{without} additional constants. 
Then, during the recursive iterations, we aim to obtain sets $U^{(j)}\subseteq U$ for $i\in[1,\ell]$ such that
\begin{enumerate}
    \item $|U^{(j)}|\leq {|U|\cdot p^ j}$,
    \item $|S_i\cap U^{(j)}|\leq d\left(\frac{3p}{2}\right)^j$ for each $i\in [\HS]$.\label{enum:2}
    \item $|S_i\cap U^{(j)}|\geq d\left(\frac{p}{2}\right)^j$ for each $i\in [\HS]$.\label{enum:3}
\end{enumerate}
We note here that the upper bound of \ref{enum:2} becomes in the end $d\left(\frac{3p}{2}\right)^j= 3^\ell \cdot d^\delta$. Since $\ell$ is a constant, this becomes $O(d^\delta)$. The lower bound of \ref{enum:3} becomes $d\left(\frac{p}{2}\right)^\ell=d^{\delta}.$

In the following, we describe how to derandomize the sampling $U^{(j+1)}$ from $U^{(j)}$, satisfying these three properties.
We assume that $U^{(j)}$ is given.
Moreover, we suppose that for an index $i\in [\HS]$, the elements of the set $S_i\cap U^{(j)}$ are stored in consecutive machines, with $d^{\delta}$ elements per machine, except at most one since we can sort in $O(1/\delta)$ rounds with $O(d^{\delta})$ local space by \Cref{lem:sorting_prefixsum_MPC}.

\paragraph{Reducing randomness via limited dependent hash functions.}
We let $\mathcal H$ be a family of $(\frac{12}{\delta}\log_d \HS)$-wise independent hash functions $h: U^{(j)}\to [\lceil 1/p\rceil]$, where $p=\tfrac{200}{\delta}/d^{\delta/2}$.
Each hash function $h$ corresponds to a set of elements $D_h\coloneq \{u\in U^{(j)}\mid h(u)=1\}$.
By the following claim, we can bound the expected size of $D_h$.
For a machine $M$, we let $V(M)$ be the set of elements stored in $M$.

\begin{claim}\label{claim:size of S1}
    Let $N< d^{d^{\delta/4}}$.
    Let $\mathcal H$ be a family of $(\frac{24}{\delta}\log_d \HS)$-wise independent hash functions $h: U^{(j)}\to [\lceil 1/p\rceil]$, where $p=\tfrac{200}{\delta}/d^{\delta/2}$.
    For $h\sim \H$, the expected size of $D_h$ is $|U^{(j)}|\cdot p$.
    Moreover, for each machine $M$ storing $d^\delta$ vertices of some set $S_i$, we have that
    $\P_{h\sim\H}[\left| |V(M)\cap D_h|-\mu\right| \ge \mu/2]\le O\left(\frac 1 {d\HS^{2}}\right)$, for $\mu:=|V(M)|\cdot p\ge\tfrac{200}{\delta}\cdot d^{\delta/2}$.
\end{claim}
\begin{proof}
    It is clear that the expected size of $D_h$ is at most $|U^{(j)}|\cdot p$ by  linearity of expectation.
    In the following, we prove the latter claim.

    For a machine $M$ storing at least $d^{\delta}$ vertices of some set $S_i$, we let $\mu:= |V(M)|\cdot p$ denote the expected sampled size. We bound the probability that $|V(M)\cap D_h|$ deviates by more than $\mu/2$.
    Note that $|V(M)\cap D_h|$ is the same as the sum of the indicator random variable $X_u$ for the event $u\in D_h$ for $u\in V(M)$.
    Here, we let $X_u$ be $(\frac{24}{\delta}\log_d \HS )$-wise independent random variables, that is $X_u=1$ if $u\in D_h$, otherwise $X_u=0$. Next, we want to use \Cref{lem:tail_bound_independent}, which requires that $\mu \ge \frac{24}{\delta}\log_d N$. We see that this is the case by filling in $\mu\ge \tfrac{200}{\delta} d^{\delta/2} \ge   \frac{24}{\delta}\log_d N$.
    So \Cref{lem:tail_bound_independent} implies the following bound:
    \begin{align*}
        & \P_{h\sim\H}[\left| |V(M)\cap D_h|-\mu\right| \ge \mu/2] 
        = \P_{h\sim\H}\left[ \left|\sum_{u\in V(M)} X_u-\mu\right|\geq \mu/2 \right] \\
        \leq &8 \left(\frac{2\cdot \frac{24}{\delta}\log_d \HS}{\tfrac{1}{2^2}\cdot \mu}\right)^{\frac{12}{\delta}\log_d \HS} 
        \le 8 \left(\frac{\log_d \HS}{d^{\delta/2}}\right)^{\frac{12}{\delta}\log_d \HS}\\
        &\le 8 \left(\frac{1}{d^{\delta/4}}\right)^{\frac{12}{\delta}\log_d \HS}
        \leq O\left( \frac 1 {d\HS^{2}}\right).\qedhere
    \end{align*}
\end{proof}

\paragraph{Derandomization via the method of conditional expectation.}
Our goal is to compute a hash function $h^*$ in $\H$ which corresponds to $U^{(j+1)}\coloneq D_{h^*}$ so that 
$|D_{h^*}|$ is at most $O(|U^{(j)}|\cdot p)$ and $|V(M)\cap D_{h^*}|\geq |V(M)|\cdot p$ for every machine $M$ if $|V(M)|\geq d^\alpha$. To obtain this, we use the method of conditional expectation, \Cref{lem:MPC_cond_exp}. Hereto, we define the corresponding objective functions.  

First, consider the elements of $U^{(j)}$ stored on $|U^{(j)}|/\LS$ machines. Let $f_1, \dots f_{|U^{(j)}|/\LS}$ denote the number of sampled elements on each such machine. 
Hence, we can write the size of the sampled elements as $D_h = \sum_{i=1}^{|U^{(j)}|/\LS}f_i(h)$.

Next, we need to ensure that we do not under- or over-sample on any machine. Here the objective functions are going to be different than before. 
We consider $O(N\cdot d^{1-\delta})$ machines $M_1,\ldots, M_{O(N\cdot d^{1-\delta}})$ storing the sets $S_1\cap U^{(j)},\ldots S_\HS\cap U^{(j)}$.
On each such machine $M_i$, let 
\begin{align*}
    g_i(h) := \begin{cases}
        |U| & \text{if } \left| |V(M)\cap D_h|-|V(M)|\cdot p\right| \ge \tfrac{1}{2}|V(M)|\cdot p,\\
        0 & \text{otherwise.}
    \end{cases}
\end{align*}
This means that whenever a machine does not stay close to the objective value, we artificially blow up the final size. If we find a hash function where the sum objective functions is $<|U|$, we know that none of the machines violate $||V(M)\cap D_h|-|V(M)|\cdot p\| \ge \tfrac{1}{2}|V(M)|\cdot p$. Note that $g_i(h)$ can be evaluated within a single machine, and hence its expectation over all $h\sim \H$ given some prefix can be computed internally.

\begin{claim}\label{claim:internal_computation_copy3}
    The values $\E_{h\sim\mathcal H}[f_i(h)\mid \textnormal{the prefix of } h\in \mathcal H\textnormal{ is } \textsf{r } ]$  and  $\E_{h\sim\mathcal H}[g_i(h)\mid \textnormal{the prefix of } h\in \mathcal H\textnormal{ is } \textsf{r } ]$ can be computed in a single machine with $\LS$ local space, for any $i$ where $f_i$ or $g_i$ is defined and any $\mathsf r\in\{0,1\}^*$. 
\end{claim}
\begin{proof}
    As $f_i$ is defined as in \Cref{claim:internal_computation}, the proof of \Cref{claim:internal_computation} holds again. 

    The function $g_i(h)$ depends on the size of the sampled elements within one machine. Hence, for each $h\in \H$, we can compute $g_i(h)$ internally. The machine sums the values $g_i(h)$ for every possible $h\in \H$ with prefix $r$ and divides out by the total number of such $h$ to obtain $\E_{h\sim\mathcal H}[g_i(h)\mid \textnormal{the prefix of } h\in \mathcal H\textnormal{ is } \textsf{r } ]$. Note that since we only consider one function $h$ at the time, and we can consider $O(L)$ elements of $U$, we require only $L$ local space.
\end{proof}

\paragraph{Putting things together in the iterative procedure.}
Using the results from above, we obtain the following lemma.

\begin{claim}\label{lem:sublinear_firstphase}
    For $p=2/d^{\delta/2}$ for some constant $\delta\in(0,1)$, then we can deterministically compute a set $U^{(j+1)}\subseteq U^{(j)}$ in $O(\log_d |U|\log_d^2 \HS)$ rounds so that:
    \begin{itemize}
        \item $|U^{(j+1)}|\leq |U^{(j)}|\cdot p$, 
        \item $|S_i \cap U^{(j+1)}|\leq |S_i\cap U^{(j)}|\cdot \tfrac{3p}{2}$ for $i\in [\HS]$, and
        \item $|S_i \cap U^{(j+1)}|\geq |S_i\cap U^{(j)}|\cdot \tfrac{p}{2}$ for $i\in [\HS]$
    \end{itemize}
    in the sublinear MPC model with $O(d^{\delta})$ local space.
\end{claim}
\begin{proof}
     We have the following:
    \begin{enumerate}
        \item The hash functions of $\H$ are encoded by $O(\log |U| \log_d(|U|\cdot \HS))$-bit random seeds by \Cref{lem:encoding_independent};
        \item $\E_{h\sim\H}\left[\sum_{i=1}^{|U|/L} f_i(h)+\sum_{i=1}^{Nd^{1-\delta}} g_i(h)           
        \right]\leq O(|U^{(j)}|\cdot p)$ by \Cref{claim:size of S1};
        \item The expectation of each objective function $f_i$ and $g_i$ can be computed in a single machine with $\LS$ local space for any prefix $\mathsf r\in\{0,1\}^*$ by \Cref{claim:internal_computation_copy3}.
    \end{enumerate}
    
    This satisfies the demands of \Cref{lem:MPC_cond_exp}, hence we can compute a proper hash function $h^*$ such that $\sum_i f_i(h^*)$ is $O(|U^{(j)}|\cdot p)$ in $O(\log |U| \log(|U|\cdot \HS)\log (|U|+\HS)/\log^3 d)=O(\log_d|U|\log_d^2 \HS)$ rounds.
    If there is a machine $M_i$ in $M_1,\ldots, M_{N\cdot d^{1-\delta}}$ with $|V(M_i)\cap U^{(j+1)}|< |V(M_i)|\cdot p$, then the $f_i(h^*)=|U|$ which contradicts to $\sum_i f_i(h^*)$ is $O(|U^{(j)}|\cdot p)$.
    Therefore, the sampled $h^*$ corresponds to a set $U^{(j+1)}$ as required.   
\end{proof}

\paragraph{Putting everything together.}
Now we can use \Cref{lem:sublinear_firstphase} repeatedly to obtain \Cref{lem:sparsified_instance_HS}. 

\begin{proof}[Proof of \Cref{lem:sparsified_instance_HS}]
    We start by applying \Cref{lm:preprocessing} to reduce the size of the data we work with to $O(Nd)$ in $O(\log_d M)$ rounds. This also ensures that $U$ is stored on $O(Nd^{1-\delta})$ consecutive machines and each set $S_i$ is stored within consecutive $\lceil d^{1-\delta}\rceil$ machines. 
    
    By recursively applying \Cref{lem:sublinear_firstphase}, we can get a set of elements $U^{(\ell)}$ of size $O(|U|\cdot p^{\ell})$ in $O(\log_d|U|\log_d^2 \HS)$ rounds.
    Furthermore, $|S_i\cap U^{(\ell)}|=O(d^\delta)$, and $|S_i\cap U^{(\ell)}|\geq d\cdot p^{\ell}=d^{\delta}$ for all $i\in [\HS]$, so we return $U'= U^{(\ell)}$.
\end{proof}

\subsubsection{Fast Hitting Sets}\label{sec:non-optimal hitting set sublinear MPC}
Using our sparsified instance from \Cref{lem:sparsified_instance_HS}, we now obtain the following result. We note that the demand that $\HS< d^{d^{\delta/4}}$ is in particular satisfied for $N= O(\poly d)$. 

\begin{restatable}{lemma}{ThmSublinearMPCHittingSet}\label{thm:sublinear_MPC_hitting set}
        Let $\delta >0$ be a constant.
         There exists a deterministic MPC algorithm that, given a universe $U$, a number $d\le |U|$ and collection of $\HS< d^{d^{\delta/4}}$ subsets $S_1, \dots, S_{\HS} \subseteq U$ with $|S_i|\ge d$ for each $i\in [\HS]$, computes a hitting set $D$ of size $ O(|U|\log_d N/d^{0.99})$.
         Let $M:=|U|+ \sum_{i=1}^\HS |S_i|$.
         Then the algorithm takes $O(\log_d|U|\log_d^2\HS+\log_d M)$ rounds with $O(d^\delta)$ local space and $O(M)$ total space. 
\end{restatable}
\begin{proof}
    We start by applying \Cref{lm:preprocessing} to reduce the size of the data we work with to $O(Nd)$ in $O(\log_d M)$ rounds. This also ensures that $U$ is stored on $O(Nd^{1-\delta})$ consecutive machines and each set $S_i$ is stored within consecutive $\lceil d^{1-\delta}\rceil$ machines. 
    
    Next, we call \Cref{lem:sparsified_instance_HS} to obtain a sparsified instance $U'\subseteq U$ such that $|U'|=O(|U|/d^{1-\delta})$,  $|S_i \cap U'| =O(d^\delta)$ and $|S_i \cap U'| \ge  d^\delta$ for all $i\in [N]$. We call \Cref{lem:sorting_prefixsum_MPC} to store each set $S_i\cap U'$ on its own machine.

    In the following, we describe how to compute a vertex set $D\subseteq U' \subseteq U$ of size $O(|U'|\log_d N/d^\delta)=O(|U|\log_d N/d^{0.99})$ so that $S_i\cap D\neq \emptyset$. As in the other sections, we derandomize random sampling using the method of conditional expectation, \Cref{lem:MPC_cond_exp}. 
    We consider the family $\H$ of $(400\tfrac{1}{\delta}\log_d \HS)$-wise independent hash functions $h: [|U'|]\to \left[\left\lceil \tfrac{d^{0.99\delta}}{\tfrac{400}{\delta}\log_d N}\right\rceil\right]$. 
    Then we can bound the expectation of $D_h$ and $D_h'$ over $h\sim\H$ as follows.   
    
    \begin{claim}\label{claim:sublinear_subsampled}
        Let $\H$ be a family of $(\tfrac{200}{\delta}\log_d \HS)$-wise independent hash functions $h: [|U'|]\to \left[\left\lceil \tfrac{d^{0.99\delta}}{\tfrac{400}{\delta}\log_d N}\right\rceil\right]$.
        For $h\sim \H$, the expectation of $D_h$ and $D_h'$ are $O(|U'|\log_d N/d^{0.99\delta})$ and $O(1)$, respectively.
    \end{claim}
    \begin{proof}
         It is clear that the expected size of the sampled vertex set $D_h$ is at most $O(|U'|\log_d N/d^{0.99\delta})$ by the linearity of expectations.
         We let $X_u$'s be $(\tfrac{200}{\delta}\log_d \HS)$-wise independent random variables, that is $X_u=1$ if $u\in D_h$, otherwise $X_u=0$ for $u\in V$.
        Then the following probability holds for any $i\in [\HS]$ by
        \Cref{lem:tail_bound_independent}:
        \begin{align*}
            \P_{h\sim\H}[ V(M_i)\cap D_h=\emptyset]=\P_{h\sim\H}\left[ \sum_{u\in V(M_i)} X_u=0 \right]
            \leq \P_{h\sim\H}\left[ \left|\sum_{u\in V(M_i)} X_u-\mu\right|\geq \mu \right] \nonumber \\
            \leq 8 \left(\frac{2\cdot \tfrac{200}{\delta}\log_d \HS}{\mu}\right)^{\tfrac{100}{\delta}\log_d \HS} 
            \leq 8 \left(\frac{\tfrac{400}{\delta}\log_d N}{\tfrac{400}{\delta}\log_d N\cdot  d^{0.01\delta}}\right)^{\tfrac{100}{\delta}\log_d \HS} 
            \leq O\left (\frac 1 {\HS}\right ),
        \end{align*}
        where in the second to last inequality we use $\mu\coloneqq \E[\sum_{u\in V(M_i)}X_u]\geq\tfrac{400}{\delta}\log_d N\cdot  d^{0.01\delta}$.
        Therefore, $\E_{h\sim\H}[\sum_{i=1}^{\HS}f_{|U|/\LS+i}(h)]$ is at most $O(1)$.
    \end{proof}

    To apply \Cref{lem:MPC_cond_exp}, we define the objective functions as follows. 
     First, consider the elements of $U'$ stored on $|U'|/\LS$ machines. Let $f_1, \dots f_{|U'|/\LS}$ denote the number of sampled elements on each such machine. 
    Hence, we can write the size of the sampled elements as $D_h = \sum_{i=1}^{|U'|/\LS}f_i(h)$. 
    
    Second, the $\HS$ sets $S_1\cap U',\dots S_\HS\cap U'$ are stored on $N$ machines $M_1, \dots, M_{N}$. On each such machine $M_i$ let $g_{i}$ denote whether $S_i\cap U'\cap D_h=\emptyset$. Hence, we can write the size of the elements from the second phase as $D'_h = \sum_{i=1}^{N}g_i(h)$. 

    \begin{claim}\label{claim:internal_computation_copy4}
        The values $\E_{h\sim\mathcal H}[f_i(h)\mid \textnormal{the prefix of } h\in \mathcal H\textnormal{ is } \textsf{r } ]$  and  $\E_{h\sim\mathcal H}[g_i(h)\mid \textnormal{the prefix of } h\in \mathcal H\textnormal{ is } \textsf{r } ]$ can be computed in a single machine with $\LS$ local space, for any $i$ where $f_i$ or $g_i$ is defined and any $\mathsf r\in\{0,1\}^*$. 
    \end{claim}
    \begin{proof}
        The function $f_i(h)$ denotes how many elements of $U'$ on machine $M_i$ are sampled. For a fixed $h\colon U' \to [\lceil 1/p\rceil ]$, the machine can check internally which elements are sampled (encoded by $h(u)=1)$), by counting this for every possible $h\in \H$ with prefix $r$, and dividing out by the total number of such $h$, we obtain $\E_{h\sim\mathcal H}[f_i(h)\mid \textnormal{the prefix of } h\in \mathcal H\textnormal{ is } \textsf{r } ]$. Note that since we only consider one function $h$ at the time, and we can consider $O(L)$ elements of $U'$, we require only $L$ local space. 
    
       The function $g_i(h)$ denotes whether $S_i\cap U' \cap D_h=\emptyset$. Again, for fixed $h$ we can compute this. For $i \in [N]$, we see if $h(u)=1$ for at least one $u\in S_i\cap U'$. If not, then $g_i(h)=1$. As before, $M_i$ sums the values $g_i(h)$ for every possible $h\in \H$ with prefix $r$ and divides out by the total number of such $h$ to obtain $\E_{h\sim\mathcal H}[g_i(h)\mid \textnormal{the prefix of } h\in \mathcal H\textnormal{ is } \textsf{r } ]$. Note that since we only consider one function $h$ at the time, and we can consider $O(L)$ elements of $U'$, we require only $L$ local space.
    \end{proof}

     We now have the following:
        \begin{enumerate}
            \item The hash functions of $\H$ are encoded by $O(\log_d \HS (\log\HS+\log |U|))$-bit random seeds by \Cref{lem:encoding_independent};
            \item $\E_{h\sim\H}\left[ \sum_{i=1}^{|U|/L} f_i(h)+\sum_{i=1}^{N} g_i(h) \right]\leq O(|U'|\log_d N/d^{0.99\delta})$ by \Cref{claim:sublinear_subsampled};
            \item The expectation of each objective function $f_i$ and $g_i$ can be computed in a single machine with $\LS$ local space for any prefix $\mathsf r\in\{0,1\}^*$ by \Cref{claim:internal_computation_copy4}.
        \end{enumerate}
    
    Then we can compute a proper hash function $h^*$ in $\H$ that corresponds to $\sum_{i=1}^{|U|/\LS+\HS} f_i(h^*)\leq O(|U'|\log_d N/d^{0.99\delta})$ in $O(\log_d |U|\log_d^2 \HS )$ rounds by using \Cref{lem:MPC_cond_exp}.
    Note that the obtained $h^*$ gives a hitting set $D$ satisfying that $S_i\cap D\neq \emptyset$ if $|V(M_i)\cap U'|\geq d^\delta$.
    By the property of $U'$, the vertex set $D$ is a hitting set of the original instance.
    Furthermore, the size of $D$ is bounded as follows due to $|U'|\leq O(|U|/d^{1-\delta})$:
    \[
    |D|\leq O\left(\frac{|U'|\log_d N}{d^{0.99\delta}}\right)\leq O\left(\frac{|U|\log_d N/d^{1-\delta }}{d^{0.99 \delta}}\right)
    =O\left(\frac {|U|\log_d N}{d^{1-0.01\delta }}\right)=O\left(\frac {|U|\log_d N}{d^{0.99}}\right).\]
    This completes the proof of \Cref{thm:sublinear_MPC_hitting set}.
\end{proof}

\subsubsection{Hitting Sets of Size \texorpdfstring{$O\left(\tfrac{|U|\log N}{d}\right)$}{}}
\label{sec:optimal hitting set sublinear MPC}
In this section, we extend \Cref{thm:optimal_linear MPC hitting set} for sublinear MPC models. 
That is, we compute an $O(|U|\log N/d)$ sized hitting set. Essentially, we use the algorithm from the previous section, but with another hash function (\Cref{lem:hash_functions_for_optimal}).
The following lemma summarizes the result.
\begin{restatable}{lemma}{ThmOptimalSublinearMPCHittingSet}\label{thm:optimal_sublinear_MPC_hitting set}
    Let $\delta >0$ be a constant.
    There exists a deterministic MPC algorithm that, given a universe $U$, a number $d\le |U|$ and collection of $N<  d^{d^{\delta/4}}$\tijn{changed bound on $N$} subsets $S_1, \dots, S_{\HS} \subseteq U$ with $|S_i|\ge d$ for each $i\in [\HS]$, computes a hitting set $D$ of size $ O\left(\tfrac{|U|\log N}{d}\right)$.
    Let $M:=L+|U|+ \sum_{i=1}^\HS |S_i|$.
    Then the algorithm takes in $O(\log \log (d+N)\cdot \log_d \log |U|\log_d^2\HS+\log_d M)$ rounds with $O(d^\delta)$ local space and $O(M)$ total space. 
\end{restatable}
\begin{proof}
    Again, we start by applying \Cref{lm:preprocessing} to reduce the size of the data we work with to $O(Nd)$ in $O(\log_d M)$ rounds. This also ensures that $U$ is stored on $O(Nd^{1-\delta})$ consecutive machines and each set $S_i$ is stored within consecutive $\lceil d^{1-\delta}\rceil$ machines.

     We modify the algorithm of \Cref{thm:optimal_linear MPC hitting set}  by the same process as in \Cref{sec:non-optimal hitting set sublinear MPC} to return an $O(|U|\log d/d)$ sized hitting set in the sublinear MPC model.



      Again, we first call \Cref{lem:sparsified_instance_HS} to obtain $U'\subseteq U$ such that:
     \begin{itemize}
         \item $|U'|\leq O(|U|/d^{1-\delta})$, 
         \item $|S_i\cap U'|= O( d^{\delta})$ for $i\in [\HS]$, and 
         \item $|S_i\cap U'|\geq d^{\delta}$ for $i\in [\HS]$.
     \end{itemize}
     We compute a hitting set $D$ with respect to $U'$ of size $O(|U'|\log N/d^\delta)=O(|U|\log N/d)$. 
    The idea is to sample with probability $p=2\cdot \tfrac{\log N}{d^\delta}$, and derandomize this procedure. 
     
     We show how to obtain this random process with a family~$\H$ of hash functions $h: U'\to [\lceil1/p\rceil]$ from \Cref{lem:hash_functions_for_optimal} with $\eps=d^{-\delta\cdot c}$ for some sufficiently large constant $c$. \Cref{lem:hash_functions_for_optimal} gives that the random seed has $O(\log d^\delta \cdot \log \log d^\delta)$ bits. 
    Analogous to \Cref{thm:linear MPC hitting set}, each hash function $h\in\H$ corresponds to two element sets $D_h:=\{u\in U'\mid h(u)=1\}$ and $D'_h := \{ S_i[1] : S_i\cap D_h=\emptyset, i\in [\HS]\}$, where $S_i[1]$ denotes the first (lexicographically ordered) element of $S_i$.
    
    \begin{claim}\label{claim:opt_subl_phase1}
         For $h\sim \H$, the expected size of $D_h$ and $D'_h$ are $O(|U'|\log N/d^\delta)$ and $O(1)$, respectively.
    \end{claim}
     The proof of this claim is exactly the same as the proof of \Cref{claim:opt_hitting_set_large_N}, with $D_1$ and $d^\delta$ instead of $U$ and $d$. 
    
    To apply \Cref{lem:MPC_cond_exp}, we define the objective functions as follows. 
     First, consider the elements of $U'$ stored on $|U'|/\LS$ machines. Let $f_1, \dots f_{|U'|/\LS}$ denote the number of sampled elements on each such machine. 
    Hence, we can write the size of the sampled elements as $D_h = \sum_{i=1}^{|U'|/\LS}f_i(h)$. 
    
    Second, the $\HS$ sets $S_1\cap U',\dots S_\HS\cap U'$ are stored on $N$ machines $M_1, \dots, M_{N}$. On each such machine $M_i$ let $g_{i}$ denote whether $S_i\cap U'\cap D_h=\emptyset$. Hence, we can write the size of the elements from the second phase as $D'_h = \sum_{i=1}^{N}g_i(h)$. 

    \begin{claim}\label{claim:internal_computation_copy5}
        The values $\E_{h\sim\mathcal H}[f_i(h)\mid \textnormal{the prefix of } h\in \mathcal H\textnormal{ is } \textsf{r } ]$  and  $\E_{h\sim\mathcal H}[g_i(h)\mid \textnormal{the prefix of } h\in \mathcal H\textnormal{ is } \textsf{r } ]$ can be computed in a single machine with $\LS$ local space, for any $i$ where $f_i$ or $g_i$ is defined and any $\mathsf r\in\{0,1\}^*$. 
    \end{claim}
    Since $f_i$ and $g_i$ are defined the same, the proof of the claim is exactly the same as in \Cref{claim:internal_computation_copy4}
    
    Now, we have the following:
        \begin{enumerate}
            \item The hash functions of $\H$ are encoded by $s$-bit random seeds by \Cref{lem:hash_functions_for_optimal};
            \item $\E_{h\sim\H}\left[ \sum_{i=1}^{|U|/L} f_i(h)+\sum_{i=1}^{N} g_i(h) \right]\leq  O(|U'|\log N/d^\delta)$ by \Cref{claim:opt_subl_phase1};
            \item The expectation of each objective function $f_i$ and $g_i$ can be computed in a single machine with $\LS$ local space for any prefix $\mathsf r\in\{0,1\}^*$ by \Cref{claim:internal_computation_copy5}.
        \end{enumerate}
    
    Then we can compute a proper hash function $h^*$ in $\H$ that corresponds to $\sum_{i=1}^{|U|/\LS+\HS} f_i(h^*)\leq O(|U'|\log N/d^\delta)$ in $O(s\log k/\log^2 d)=O(\log \log (d+N) \log_d N)$ rounds by using \Cref{lem:MPC_cond_exp}.
    Note that the obtained $h^*$ gives a hitting set $D$ satisfying that $S_i\cap D\neq \emptyset$ if $|V(M_i)\cap U'|\geq d^\delta$.
    By the property of $U'$, the vertex set $D$ is a hitting set of $G$.

    Furthermore, the size of $D$ is bounded as follows due to $|U'|\leq O(|U|/d^{1-\delta})$:
    \[
    |D|\leq O\left(\frac{|U'|\log N}{d^\delta}\right)\leq O\left(\frac{|U|/d^{1-\delta}\log N}{d^\delta}\right)
    =O\left(\frac {|U|\log N}{d}\right).\qedhere\]
\end{proof}

\newpage

\section{Application I: Spanner Algorithms}\label{sec:applications_spanners}
In this section, we derive deterministic algorithms for \emph{spanners} of undirected graphs in the Congested Clique and the MPC model by applying the hitting set algorithms illustrated in~\Cref{sec:hitting sets}.

In \Cref{sec:unweighted_spanner_linearMPC}, we first construct an $O(k)$-spanner with $O(n^{1+1/k})$ edges for an unweighted graph with $n$ vertices in linear MPC model and the Congested Clique.

Then, in \Cref{sec:weighted_spanner_linearMPC}, we give a reduction from weighted spanners to unweighted spanners, obtaining \Cref{thm:linear_MPC_spanner_weighted}.
As described in \Cref{sec:our_results_linearMPC}, these spanners can subsequently be used to obtain $O(\log{n})$-approximate APSP results (see \Cref{sec:apsp_corollary}).


Finally, in~\Cref{sec:weighted_spanner_sublinear}, we give our results for weighted spanners in sublinear MPC.

\subsection{Unweighted Spanners in Linear MPC and Congested Clique}\label{sec:unweighted_spanner_linearMPC}

In this section, we let $G$ be an unweighted, undirected graph with $n$ vertices and maximum degree $\Delta$.
Dory, Fischer, Khoury, and Leitersdorf~\cite{podc2021spanner} gave a randomized $O(1)$-round algorithm constructing an $O(k)$-spanner of $G$ in the MPC model with {$\tilde O(n)$ local space}.
Briefly, their main approach was to build and merge several spanners hierarchically with respect to the degrees of the vertices.
They compute small sized $d_i$-dominating sets $D_i$ of $G$ with $d_i=2^{i-1}$ and $D_{i}\subseteq D_{i-1}$ for all $i\in[\lceil\log \Delta\rceil]$.
Here, a $d_i$ dominating set refers to a vertex set $D\subseteq V$ so that $(N(v)\cup\{v\})\cap D\neq \emptyset $ for every vertex $v\in V$ with $\deg(v)\geq d_i$.
Then they compute and union the spanners with respect to the \emph{clustering graphs} defined by the dominating sets $D_1,\ldots, D_{\lceil \log\Delta\rceil}$. 
Briefly, they decompose $G$ into $\lceil \log\Delta\rceil$ subgraphs $G_1,\ldots, G_{\lceil \log\Delta\rceil}$ so that: a vertex and an edge of $G$ is in at least one of them, and a vertex in $G_i$ is in $D_i$ or adjacent to a vertex in $D_i$. 
Then they define clustering graphs of vertices $D_i$'s, and a vertex of $G_i$ not in $D_i$ joins one of the adjacent clusters. 
Finally, they compute the spanners for the clustering graphs, and return their union and additional $O(n)$ edges inside the clusters as a spanner of the original graph $G$.

In this section, we derandomize their algorithm in \emph{the linear MPC model}. In particular, we improve the local space to $O(n)$ instead of $\tilde{O}(n)$.
Note that ~\cite{podc2021spanner} used random processes for two steps. First, to compute the dominating sets $D_i$. Second, to compute the spanners of the clustering graphs. We obtain a deterministic algorithm by \Cref{thm:MPC hitting set intro: linear} for the dominating set, and deriving deterministic sparsifying algorithm for the clustering, illustrated by \Cref{lem:sparsification_tool}.
Although \Cref{lem:sparsification_tool} gives a pretty larger spanner than our desired size, by applying it on the clustering graphs, we can reach the $O(n^{1+1/k})$ sized $O(k)$ spanner as \cite{podc2021spanner}.

\paragraph*{Clustering graphs.}
In the following, we construct and apply \emph{clustering graphs}.
For an unweighted graph $G$, a graph $C=(V_C,E_C)$ is a clustering graph if the following holds:
\begin{itemize}
    \item Each vertex $v$ belongs to exactly one cluster in $V_C$.
    \item Two clusters $c$ and $c'$ in $V_C$ are adjacent in the clustering graph $C$ if and only if there is an edge in $G$ between two vertices that belong to $c$ and $c'$, respectively. 
\end{itemize}

We say each cluster of $C$ has a \emph{bounded radius} at most $r$ when each cluster has a \emph{center vertex} $c\in V$ and for any other vertex $v\in V$ in the cluster, $d_G(c,v)\leq r$.
In this section, we only consider clustering graphs whose clusters have bounded radius.
For convenience, we let the vertex set $V_C$ of the clustering graph $C$ correspond to the center vertices of the clusters, where each vertex represents its respective cluster.
We suppose that each edge $\{c,c'\}$  in $C$ stores an actual edge in $G$ between two vertices that belong to $c$ and $c'$, respectively.
Furthermore, we suppose that each vertex $v$ in $G$ knows the cluster in $C$ to which $v$ belongs. 

The clustering graph compresses the original graph while preserving all the essential properties required for constructing spanners, and thus has been widely used as a powerful technical tool for building spanners~\cite{podc2021spanner,biswas2021massively}.
Especially, it is a well-known observation that any $\ell$-spanner of $C$ can be translated to an $O(\ell)$-spanner of $G$ when each cluster of $C$ has a radius one.
Note that if each cluster in $C$ has radius one, then there is a star subgraph in $G$ whose center is the center of the cluster and the leaves are the vertices belonging to the cluster.
Then when we have an $\ell$-spanner $H_C$ of $C$, we can obtain an $O(\ell)$-spanner $H$ of the original graph $G$ by replacing each vertex of $H_C$ (cluster of $C$) by the star graph.
Additionally, we can replace the edge between two clusters $c$ and $c'$ in $H_C$ by an actual edge in $G$ between two vertices belonging to the clusters of $c$ and $c'$, respectively.
Then the obtained graph is an $O(\ell)$-spanner.


\paragraph*{Sparsification of graphs.}
A deterministic algorithm to compute a spanner for the clustering graphs, was constructed by Leitersdorf~\cite{Leitersdorf22} in the Congested Clique model.
If we use it on the original input graph, the obtained spanner size is pretty larger than our desire.
However, after we compress the graph as a clustering graph, we can achieve the desired size.
We simulate this deterministic algorithm in the MPC model using the following lemma by setting $F$ as a clustering graph of $G$ with $M^{1-2/s}N^{2/s}=O(|V(G)|)$ and $M^{1/s}N^{1-1/s}=o(|V(G)|)$ for some parameter $s\geq 3$.
A detailed proof is in \Cref{sec:sparsification_lemma}.
\begin{restatable}{lemma}{LemSparsificationTool}\label{lem:sparsification_tool}
    Let $F=(V_F,E_F)$ be a graph with $N$ vertices and $M$ edges but no isolated vertex.
    For any positive $s\geq 3$, we can deterministically construct a $(2k-1)$-spanner for $F$ with {$O(M^{1/s}N^{1-1/s+1/k})$} edges in $O(1)$-rounds in the MPC model with $O(M^{1-2/s}\cdot N^{2/s})$ local space and $O(M)$ total space. 
\end{restatable}

\begin{proof}[Sketch of the proof]
    The algorithm is analogous to that of Leitersdorf~\cite{Leitersdorf22}.
    However, we generalize the prior work by introducing a parameter $s$, where the original lemma is formulated for the special case when $s=3$.
    Additionally, we give its implementation in the MPC model not only in the Congested Clique.
    Briefly, we let $p=\lceil(2M/N)^{1/s}\rceil$, and we use at most $O(p^2)$ machines with $O(M/p^2)$ local size.
    Note that this is at least linear local space, since $O(N)\leq O(M/p^2)= O(M^{1-2/s}N^{2/s})$, using that $M\geq \Omega(N)$ since there are no isolated verices.
    
    Our algorithm consists of two phases. 
    In the first phase, we distribute the given edges evenly across $O(p^2)$ machines so that: Each machine stores at most $M/p^2+2N$ edges while the induced subgraph of $F$ by the edges has at most $O(N/p)$ vertices.
    Therefore, we suppose that each of the $O(p^2)$ machines stores an induced subgraph.
    In the next phase, each of the $O(p^2)$ machines simultaneously runs the centralized algorithm~\cite{althofer1993sparse} to get a $(2k-1)$-spanner with $O((N/p)^{1+1/k})$ edges of its induced subgraph.
    Since the induced subgraphs form a decomposition of the whole graph $F$, the union of $(2k-1)$-spanners forms a $(2k-1)$-spanner of the graph $F$, of total size $O(M^{1/s}N^{1-1/s+1/k})$.
\end{proof}

In the following, we give the detailed algorithm computing a $O(k)$-spanner with $O(n^{1+1/k})$ edges by applying the dominating set algorithm of \Cref{thm:MPC hitting set intro: linear} and \Cref{lem:sparsification_tool}.

\subsubsection{Implementation in the MPC model}
The algorithm is the same as the previous one of~\cite{podc2021spanner}, except that we use the deterministic hitting set algorithm of \Cref{thm:MPC hitting set intro: linear} instead of the random sampling process.
Furthermore, we use the deterministic sparsification ~\cref{lem:sparsification_tool} instead of the randomized version.
Here, we give the implementation of the algorithm in the MPC model with $O(n)$ local space and $O(m)$ total space.
Initially, all the edge information is given arbitrarily distributed across the machines. 
For each edge $\{u,v\}$ stored in the machines, we replace it into two pairs $(u,v)$ and $(v,u)$.
By using the sorting of~\Cref{lem:sorting_prefixsum_MPC} to the $2m$ pairs corresponding to the edge set, we can gather the edges incident to a vertex into a single machine. 
Then we can compute the degree of every vertex, and thus, we partition the vertices according to their degree: $V= V_1 \cup V_2 \cup \cdots \cup V_{\lceil\log \Delta\rceil}$ such that $V_i := \{ v\in V \mid d_i \le \deg(v) < 2d_i\}$. It takes $O(1)$ rounds.
In the following, each vertex $v$ is stored as a pair $(v,x)$ with an index $x\in[\lceil \log \Delta\rceil]$ of which $v\in V_x$. Furthermore, we replace each pair $(v,u)$ corresponding to an edge $\{u,v\}$ as $(x,v,u)$ with $v\in V_x$.

\subparagraph*{Hitting sets $D_1\supseteq \ldots\supseteq D_{\lceil\log\Delta\rceil}$.}
We first show how to compute $d_i$-dominating sets $D_i'$ of $G$ with $d_i=2^{i-1}$ for all $i\in[\lceil\log \Delta\rceil]$. 
Here, $D_1',\ldots,D_{\lceil\log\Delta\rceil}'$ do not satisfy the monotonicity condition yet, however, we would construct the hitting set $D_i$'s from them.
We first aim to obtain sets $D_i'\subseteq V$ such that $(N(v)\cup\{v\})\cap D_i' \neq \emptyset$ for every vertex $v$ with $\deg(v)\ge d_i$. 
In other words, our goal is to solve the \emph{hitting set problem} defined on the sets $\{ N(v)\cup\{v\} \mid v\in V_i\}$, whose size is at least $d_i$ but at most $2d_i$, and the universe $U_i=V_i\cup N(V_i)$. Here, $N(V_i)$ is the set of vertices adjacent to $V_i$ in $G$. 
First, we consider the case that the instance is small, more precisely: $|V_i|d_i = O(n)$. In this case, the description of the problem fits in one machine. 
In this case,  we gather all the edge information in a single machine, using the sorting process~\Cref{lem:sorting_prefixsum_MPC} and compute a hitting set deterministically internally using $O(|V_i|d_i)$ total space with a greedy algorithm~\cite{LOVASZ1975383}. Note that since this machine is not full, it can still be shared with other instances. 

Next, we consider the larger case, $|V_i|d_i = \Omega(n)$. In this case, we use the algorithm of \Cref{thm:linear MPC hitting set}, \kyungjin{correct reference}
with $U=U_i$, $N=|V_i|$, $d=d_i$, and $L=|U_i|$. This outputs a hitting set $D_i'$ of size $O(|U_i|/d_i^{0.99}+|V_i|/d_i)\leq O(n/d_i^{0.99})$
\begin{itemize}
    \item in $O(1)$ rounds
    \item with $O(\sum_{v\in V_i}|N(v)|+|U_i|)=O(|V_i|d_i)$ total space and $O(|U_i|)\leq O(n)$ local space.
\end{itemize}
Note that the equality holds since the vertices in $V_i$ have a degree at most $2d_i$.

Although we need $[\lceil\log \Delta\rceil]$-many dominating sets, by \Cref{obs:parallelism}, we can construct them simultaneously in $O(1)$ rounds in the MPC model with $O(n)$ local space and $O(\sum_i |V_i|d_i)$ total space. 
Note that the total space $O(\sum_i |V_i|d_i)$ is at most $O(m)$ as following:
\begin{align*}
    2m=\sum_{v\in V} \deg(v)= \sum_{i\in[\lceil \log \Delta\rceil]}\sum_{v\in V_i}\deg(v) \geq \sum_{i\in[\lceil\log \Delta\rceil]}|V_i|d_i.
\end{align*}
We suppose that the algorithm returns all the information of $D_1',\ldots, D_{\lceil \log \Delta\rceil}'$ within a single machine.
Even if the information is distributed arbitrarily, we can redistribute it in a constant time using the sorting process~\Cref{lem:sorting_prefixsum_MPC}.
Furthermore, since each $D_j'$ has a size at most $O(n/d_j^{0.99})$, The total size $\sum_j |D_j'|=O(n)$ can be bounded by the the geometric series:
\begin{equation*}
    \sum_{j\in[\lceil \log\Delta\rceil]} |D_j'| \le \sum_{j\in[\lceil \log\Delta\rceil]} O(n/d_j^{0.99}) \le \sum_{j\in[\lceil \log\Delta\rceil]} O(n/2^{0.99(j-1)}) = O(n/2^{0.99\lceil \log\Delta\rceil})=O(n).
\end{equation*}

Now we have $d_i$-dominating set $D_i'$ for each $i$. 
We suppose that we have a tuple $((v,x),j)$ for each vertex $v\in V_x$ such that $j$ is the maximum index with $v\in D_j'$. Recall that we assumed at the start that each vertex is stored as a pair $(v,x)$ with $v\in V_x$.
If $v$ is not in any dominating set, then we set the index $j$ to zero.
For each $i\in[\lceil \log \Delta\rceil]$, we define $D_i$ as the prefix union $\bigcup_{j\geq i} D_j'$, so that $D_1,\ldots,D_{\lceil\log \Delta\rceil}$ are the dominating sets with monotonicity.
Then the pair $((v,x),j)$ immediately gives us the information so that $v\in D_i$ for every $i\leq j$.
Additionally, we replace the tuples $(x,v,u)$ corresponding to the edges as the tuple $(j,x,v,u)$ so that $j$ is the maximum index with $v\in D_j'$ and $v\in V_x$.
Note that we bound the size of $D_i$ using a geometric series:
\begin{equation*}
    |D_i| \le \sum_{j\ge i} |D_j'| \le \sum_{j\ge i} O(n/d_j^{0.99}) \le \sum_{j\ge i} O(n/2^{0.99(j-1)}) = O(n/2^{0.99i})=O(n/d_i^{0.99}).
\end{equation*}

\subparagraph*{Clustering graphs $C_1,\ldots, C_{\lceil\log \Delta\rceil}$.}
With respect to the constructed dominating set $D_i$ for $i\in[\lceil\log\Delta\rceil]$, 
we let {$G_i=(\bigcup_{j\geq i}V_j,E_i)$} be the subgraph of $G$ of which $V_i$ is the set of vertices of degree at least $d_i$ but less than  $2d_i$ and $E_i$ is the set of edges in $E$ whose both end vertices are in $\bigcup_{j\geq i}V_j$ but at least one is in $V_i$.
Then our goal is to construct the clustering graphs $C_i$ of $G_i$
of which the cluster centers are $D_i$ and the radii of all clusters are at most one.
Recall that we build $D_i=\bigcup_{j\geq i}D_j'$, where $D_j'$ is a dominating set of $V_j$, and thus, $D_i$ is a dominating set of $\bigcup_{j\geq i}V_j$.
Furthermore, each edges $\{u,v\}\in E_i$ with $v\in V_i$ are stored as two tuples $(j,i,v,u)$ so that $j$ is the maximum index with $v\in D_j'$. Therefore, we can redistribute the tuples of the edges so that the edges of $E_i$ are in consecutive machines.
{Furthermore, if $|E_i|\le O(n)$, we can gather all the information of edges in a single machine using the sorting process~\Cref{lem:sorting_prefixsum_MPC} and compute the simple clustering graph $C_i$ internally using $O(|E_i|)$ total space. Note that this machine might not be full, and the remaining space still be used for other computations simultaneously. If $|E_i|=\Omega(n)$, we use the algorithm of~\cite{podc2021spanner}, which computes such a clustering graph $C_i$ in $O(1)$ rounds with $O(|E_i|)$ total space and $O(n)$ local space. {Furthermore, we can remove the parallel edges in the clustering graph $C_i$ within constant rounds by \Cref{lem:sorting_prefixsum_MPC}. Briefly, we can gather the parallel edges within a single machine and remove all except one by sorting the edges labeled with the lexicographical ordering of the clusters at each endpoint of the edge.}}\kyungjin{Rev A-3} 
In conclusion, by~\Cref{obs:parallelism}, we can construct all simple clustering graphs $C_1,\ldots, C_{\lceil\log \Delta\rceil}$ simultaneously in $O(1)$ rounds with $O(n)$ local space and $O(m)$ total space.
In termination, we have clustering graphs $C_i$'s each $C_i$ has $O(n/d_i^{0.99})$ vertices and at most $O(|E_i|)$ edges.

Note that in the termination of the algorithm in~\cite{podc2021spanner}, a set of vertices in $G_i$ adjacent to $c$ is stored in a single machine for each vertex $c\in D_i$ of the clustering graph $C_i$.
Furthermore, the tuples $(i,(c,c'),(j,x,v,u))$ and $(i,(c',c),(k,y,v,u))$ are distributed across the machines for each $\{c,c'\}$ edge in the obtained $C_i$, where there is an edge $\{u,v\}\in E_i$ of which $u$ and $v$ are in $N(c)\cup\{c\}$ and $N(c')\cup\{c'\}$, respectively. Recall that we assumed that each edge $\{u,v\}$ is stored as two pairs $(j,x,u,v)$ and $(k,y,v,u)$ with the indices $x,y,j,k\in[\lceil \log\Delta\rceil]$ where $u\in V_x\cap D_{j'}$ and $v\in V_y\cap D_{k'}$ with $j'\leq j$ and $k'\leq k$.
The tuples can simulate the clustering graph $C_i$ of $G_i$.

\subparagraph*{Spanners of clustering graphs.} 
We have the clustering graphs $C_1,\ldots, C_{\lceil \log \Delta\rceil}$.
Our goal is to compute the $(2k-1)$-spanners $H_1,\ldots,H_{\lceil\log \Delta\rceil}$ of the clustering graphs and merge them.
If the complexity of the clustering graph $C_i$ is at most $O(n)$, that means $O(|E_i|)=O(n)$, then we can compute a $(2k-1)$-spanner $H_i$ for $C_i$ with $O(n^{1+1/k}/d_i^{0.99})$ edges in $O(1)$ rounds by collecting $E_i$ within a single machine with $O(n)$ local space and apply the centralized algorithm of~\cite{althofer1993sparse}.
Note that the central algorithm gives $O(N^{1+1/k})$ sized $(2k-1)$-spanner of a graph with $N$ vertices, and our clustering graph $C_i$ has $N=O(n/d_i^{0.99})$ vertices.

If $C_i$ has at least $\Omega(n)$ edges, we can use \Cref{lem:sparsification_tool} to compute the spanner $H_i$.
Recall that the lemma gets an input graph with $N$ non-isolated vertices and $M$ edges along with a constant $s\geq 3$, then it computes a $(2k-1)$-spanner of the input graph with $O(M^{1/s}N^{1-1/s+1/k})$ edges in the MPC model with $O(M^{1-2/s}N^{2/s}+N)$ local space and $O(M)$ total space.
Hereto, we set the parameter $s$ such that it balances the spanner size and the local space of the MPC model is small at the same time.
This occurs exactly when $s=3.98$ since $C_i$ has $O(n/d_i^{0.99})$ vertices and $O(|E_i|)=O(nd_i)$ edges.
Precisely, if the complexity of the clustering graph $C_i$ is at most $O(n)$, that means $O(|E_i|)=O(n)$, then we can compute a $(2k-1)$-spanner $H_i$ for $C_i$ with $O(n^{1+1/k}/d_i^{0.48})$ edges in $O(1)$ rounds by collecting $E_i$ within a single machine with $O(n)$ local space and apply the centralized algorithm of~\cite{althofer1993sparse}.
Otherwise, when the complexity of $C_i$ is at least $\Omega(n)$, then
\Cref{lem:sparsification_tool} returns a spanner $H_i$ for each $C_i$ with $O(n^{1+1/k}/d_i^{0.48})$ edges while satisfying two things: 
\begin{itemize}
    \item The total space is small: $O(|E_i|)$
    \item The local space is small: $O\left((nd_i)^{1-2/3.98}\cdot \frac{n^{2/3.98}}{d_i^{0.99\cdot 2/3.98}}+n\right)=O(n)$.
\end{itemize}

By \Cref{obs:parallelism}, we can compute all $H_i$'s simultaneously in the MPC model with $O(n)$ local space and $O(\sum_i|E_i|)=O(m)$ total space. 

\subparagraph*{$O(k)$-spanner of $G$.}Recall that each edge information of a clustering graph $C_i$ between $c$ and $c'$ in $D_i$ is given as a tuple $(i,(c,c'),(j,x,u,v))$, where $\{u,v\}$ is an actual edge in $E_i$ of which $u\in N(c)\cup\{c\}$ and $v\in N(c')\cup\{c'\}$.
Therefore, each edge of $H_i$'s also gives such tuples and vertex sets.
Then we can replace all the tuples of $H_i$'s as an edge $\{u,v\}$ in $E_i$.
Furthermore, for each vertex $u\in \bigcup_{x\geq i} V_x$, we add the edge $\{u,c\}\in E$ with $c\in D_i$ that is the last edge in the lexicographical ordering of the edge tuples $(k,y,c,u)$, and we call it an \emph{intra edge}.
Choosing the lexicographically last edge $\{u,c\}$ for $\overline H_i$ ensures that it is in $\overline H_j$ for every $j\leq i$.
Precisely, even if $u$ is a vertex in $\overline H_j$, there is a unique intra edge in $\overline H_j$ incident to $u$ that is $\{u,c\}$ due to the monotonicity $D_1\supseteq \ldots\supseteq D_{\lceil \log \Delta\rceil}$.
Recall that we assumed that each edge $\{u,v\}\in E$ is stored as two pairs $(j,x,u,v)$ and $(k,y,v,u)$ with the indices $x,y,j,k\in[\lceil \log\Delta\rceil]$ where $u\in V_x\cap D_{j'}$ and $v\in V_y\cap D_{k'}$ with $j'\leq j$ and $k'\leq k$.
Therefore, we can compute the last edge in the lexicographical order in $O(1)$ rounds by the sorting process~\Cref{lem:sorting_prefixsum_MPC}.

The translating process from $H_i$'s to $\overline H_i$'s increases the stretch of the spanners at most {constant times}~\cite{podc2021spanner}.\footnote{Precisely, in Section 3.1 of~\cite{podc2021spanner}, the reduction is proven from a spanner of a clustering graph to a spanner of the original graph.}
Therefore, it obtains $O(k)$-spanners $\overline H_1,\ldots,\overline H_{\lceil\log \Delta\rceil}$ of $G_1,\ldots,G_{\lceil\log \Delta\rceil}$, respectively.
Finally, we return the union of $\overline H_i$'s as a $O(k)$-spanner of $G$.

Even if the edges of $\overline H_i$'s are arbitrarily distributed across the machines, there are $O(n)$ edges incident to a vertex. Therefore, by sorting the edges, we can store all the edges incident to a single vertex in one machine in $O(\log_n m)=O(1)$ rounds~\Cref{lem:sorting_prefixsum_MPC}.
In conclusion, our algorithm is implemented in $O(1)$ rounds using $O(n)$ local space and $O(m)$ total space.

\paragraph*{Analysis.}
Note that each $E(\overline H_i)\setminus E(\overline H_{i+1})$ has at most $O(n^{1+1/k}/d_i^{0.48}+|V_i|)$ edges, and thus, we can obtain a subgraph $H$ of $G$ by union of the all edge sets of $\overline H_i$'s. 
Then the obtained $H$ has $O(n^{1+1/k})$ edges as follows:
\begin{align*}
    \sum_{i=1}^{\lceil\log \Delta\rceil} O(n^{1+1/k}/d_i^{0.48}+|V_i|)&=O(n^{1+1/k}\cdot \sum_{i=1}^{\lceil\log \Delta\rceil} \frac {1}{2^{0.48(i-1)}}+n)\\
    &\leq O(n^{1+1/k}+n)=O(n^{1+1/k}).
\end{align*}

We show that the obtained $H$ is an $O(k)$-spanner of $G$.
Precisely, for an edge $\{u,v\}\in E_i$ in $G$ but not in $H$, we show that the distance between $u$ and $v$ in $H$ is at most $O(k)$.
Here, we let $c(u)$ and $c(v)$ be the vertices in $D_i$ adjacent to $u$ and $v$, respectively.
Note that there is an edge between $c(u)$ and $c(v)$ in the clustering graph $C_i$, and thus, the distance between $c(u)$ and $c(v)$ is at most $O(k)$ in $H_i$.
Since we replace each cluster center of $H_i$ as a star graph, the distance between $u$ and $v$ in $\overline H_i$ is at most a constant times the distance between $c(u)$ and $c(v)$ in $H_i$, that is $O(k)$ as shown~\cite{podc2021spanner}. 
Therefore, in $\overline H_i$, the distance between $u$ and $v$ is at most $O(k)$, and thus, it holds in $H$ since it is a supergraph of $\overline H_i$.
In conclusion, our algorithm satisfies the following theorem.

\begin{restatable}{theorem}{ThmLinearMPCSpannerUnweighted}\label{thm:linear_MPC_spanner_unweighted}
    There exists a deterministic algorithm that, given a positive integer $k$ and an unweighted undirected graph $G$ of $n$ vertices, computes an $O(k)$-spanner with $O(n^{1+1/k})$ edges in $O(1)$ rounds in the linear MPC model with {$O(m)$ total space}.
\end{restatable}

\subsubsection{Sparsification Lemma}\label{sec:sparsification_lemma} 

\LemSparsificationTool*
\begin{proof}
Let $p=(2M/N)^{1/s}$ through this proof. Since $F$ has no isolated vertex, $p$ are at least one.
Furthermore, for any fixed constant $c$, $M/p^c$  is in $\Omega(N)$ due to $M/p^c=\Omega(M^{1-c/s}\cdot N^{c/s})\geq \Omega(N^{1-c/s}\cdot N^{c/s})=\Omega(N)$.
In this proof, we derive the MPC algorithm using $O(p^2)$ machines with $L=c(M/p^2+N)$ local space for some constant $c\geq 3$. 
Note that $N\le M$, so $N^{1-2/s}\le M^{1-2/s}$ and hence $N \le M^{1-2/s}N^{2/s}=M/p^2$, so can simplify the local space to $L= O(M/p^2)$.
Then \Cref{obs:parallelism} guarantees that it can work in any MPC model with $O(M)$ total space and $O(M/p^2)=O(M^{1-2/s}N^{2/s})$ local space, which completes the proof.

Our algorithm consists of two phases, like the algorithm of Leitersdorf~\cite{Leitersdorf22}.
In the first phase, we distribute the given edges so that they are evenly organized across $O(p^2)$ machines so that: Each machine stores at most $M/p^2+2N$ edges while the induced subgraph of $F$ by the edges has at most $O(N/p)$ vertices.
Therefore, we suppose that each of the $O(p^2)$ machines stores an induced subgraph.
In the next phase, each of the $O(p^2)$ machines simultaneously runs the centralized algorithm~\cite{althofer1993sparse} to get a $(2k-1)$-spanner with $O((N/p)^{1+1/k})$ edges of its induced subgraph.
Since the induced subgraphs form a decomposition of the whole graph $F$, the union of $(2k-1)$-spanners forms a $(2k-1)$-spanner of the graph $F$.
Furthermore, the obtained spanner has the desired size $O(M^{1/s}N^{1-1/s+1/k})$ as follows:
\begin{align*}
    O(p^2)\cdot O((N/p)^{1+1/k})&= O(p^{1-1/k}\cdot N^{1+1/k})\leq O(p\cdot N^{1+1/k})\\
    &\\
    &=O((Np^s)^{1/s}\cdot N^{1-1/s+1/k})\\
    &\leq O(M^{1/s}N^{1-1/s+1/k}).
\end{align*}
Note that the first inequality and the last equation holds by the definition of $p=(2M/N)^{1/s}$, which is at least one.
In conclusion, we can return all the edge information of the $(2k-1)$-spanner of $F$ at termination.
In the following, we implement our first phase since it is clear that the second phase can be done in one round.

\subparagraph*{Implementation of the edge distributing phase.}
Our goal is to distribute the given edges into $O(p^2)$ machines so that each machine stores at most $M/p^2+2N$ edges, while their induced subgraph has at most $O(N/p)$ vertices.
Initially, the edges $E_F$ of $F$ are arbitrarily distributed across the machines.

As a preprocessing step, we reorganize the edge information so that all the edges incident to a same vertex are stored in a single machine by the sorting algorithm of \Cref{lem:sorting_prefixsum_MPC} while all the edge information is stored in the first half consecutive machines.
For this, we assume an arbitrary ordering $\prec$ in the vertices $V_F$.
Each machine replaces the edge information $\{u,v\}\in E_F$ into two pairs $(u,v)$ and $(v,u)$.
Then we sort the $2M$ pairs in $O(\log_L M)=O(1)$ rounds along the lexicographical ordering by \Cref{lem:sorting_prefixsum_MPC} and $L\geq c(M/p^2+N)$ while $M\leq N^2$.
Note that it guarantees that the edges incident to one vertex are stored in consecutive machines.
After that, if one machine stores two pairs whose first elements are different, 
then it sends all the pairs $(u_{\min},\cdot)$ (and all the pairs $(u_{\max},\cdot)$ to its previous machine (and next machine), where $u_{\min}$ and the $u_{\max}$ are the smallest and largest first vertex of the pairs in the machine.  
After that, for each vertex $v$ of $V_F$, all the edges incident to the vertex are stored in a single machine, denoted by $\mathcal{M}_v$.
Note that it is possible that $\mathcal M_v$ and $\mathcal M_u$ are the same for $u\neq v\in V_F$. 
Therefore, we can compute and store the degrees for every vertex in one round at $\mathcal M_v$ in $O(1)$ rounds.
In the following, we describe how to decompose the edges into $O(p^2)$ machines.
Note that after the preprocessing, only the first half machines have some information but the others are empty.

We first decompose the vertex set into $O(p)$ subsets, each of which has at most $N/p$ vertices and is incident to at most $M/p+N$ edges.
Our goal is to inform to the machine $M_v$ which part in the decomposition belongs to $v$.
We claim to do this process in $O(1)$ rounds.
For this, we first fix the last empty machine $\mathcal M$ and broadcast its id to all machines $\mathcal M_v$ for $v\in V_F$ that are consecutive machines using ~\Cref{lem:sorting_prefixsum_MPC}.
Then each machine $\mathcal M_v$ sends the tuples $(v,\deg_F(v), \textsf{id}_v)$ to $\mathcal M$ for every $v\in V_F$, where $\textsf{id}_v$ is the id of the machine $\mathcal M_v$.
In the next round, $\mathcal M$ sorts the tuples with respect to $\prec$ ordering of the vertices.
Then it cuts the sorted sequence into $O(p)$ chunks $S_1,\ldots,S_{O(p)}$ so that each $S_i$ consists of at most $N/p$ vertices and the sum of $\deg_F(v)$ for all vertices $v$ in the chunk $S_i$ is at most $M/p+N$.
The existence is proven since adding one vertex into $S_i$ increases the size of $S_i$ by one while increasing $\sum_{v\in S_i}\deg_F(v)$ by at most $N$.
Furthermore, it gives a simple algorithm simulated in one round by $\mathcal M$ without any communication.
In the next round, $\mathcal M$ allocates the $O(p)$ consecutive empty machines $\mathcal N_i^1,\ldots\mathcal N_i^{O(p)}$ to each $S_i$ chunk.
Then $\mathcal M$ sends the information of two vertices $f(i)$, $t(i)$, and ids of $\mathcal M_{t(i)}$ and $\mathcal N_{i}^1$ to the machine $\mathcal M_{f(i)}$ for each $S_i$, where $f(i)$ and $t(i)$ are the first and the last, respectively, vertices in the $S_i$ chunk. 
Note that a vertex $v$ with $f(i)\preceq v\preceq t(i)$ is in $S_i$, and thus the machines $\mathcal M_v$ with $v$ in the same chunk $S_i$ are consecutive.
Therefore, each $\mathcal M_f(i)$ can broadcast the received messages to the machines $\mathcal M_v$ that belong to the chunk $S_i$ in $O(\log_L N/p)=O(1)$ rounds by \Cref{lem:sorting_prefixsum_MPC} and $L\geq c(M/p^2+N)$.
We can do the broadcasting simultaneously for every $i\in[O(p)]$ in $O(1)$ rounds by \Cref{obs:parallelism}.

In the following, we fix a chunk $S_i$, and we let $T_i$ be the set of vertices adjacent to a vertex belongs to $S_i$.
Although we do not know the exact vertices in $T_i$, we do not need the information here.
Our goal is to partition the vertex set $T_i$ into $O(p)$ subsets $T_i^1,\ldots,T_i^{O(p)}$ so that each $T_i^j$ has at most $N/p$ vertices while the number of edges between $S_i$ and $T_i^j$ is at most $M/p^2+2N$ edges.
Recall that the number of incident edges to $S_i$ is at most $M/p+N$, and thus, the process is analogous to the previous decomposition if we have the number of edges between a vertex $u$ and $S_i$ for every $u\in V$.
Furthermore, computing all the numbers can be done analogously to the preprocessing phase in $O(1)$ rounds as follows.
Recall that an edge incident to some vertex in $S_i$ is stored as a pair $(v,u)$ whose first element $v$ is a vertex of $S_i$.
Furthermore, each machine knows whether the first element $v$ of a pair $(v,u)$ in it belongs to $S_i$ or not. 
In the case that $v$ belongs to $S_i$, the machine also has the ids for the allocated $O(p)$ machines $\mathcal N_i^1,\ldots, \mathcal N_i^{O(p)}$.\footnote{Although the machine has the id of $N_i^1$ only, it can contact to any other machines in $\mathcal N_i^1,\ldots, \mathcal N_i^{O(p)}$ since they are consecutive machines.}
Therefore, analogously to the preprocessing phase, 
we can sort all of the pairs corresponding to incident edges to $S_i$ into $\mathcal N_{i}^1,\ldots,\mathcal N_{i}^{O(p)}$ along $\prec$ ordering of the second elements in $O(1)$ rounds.
Furthermore, we can also count the number of pairs $(v,u)$ whose first element $v$ is in the chunk $S_i$ for every vertex $u\in V$ simultaneously as we did in the preprocessing.
In conclusion, we can compute the decomposition of the vertex set into $O(p)$ subsets $T_i^1,\ldots,T_i^{O(p)}$ so that each $T_i^j$ has at most $N/p$ vertices while the number of edges between $S_i$ and $T_i^j$ is at most $M/p^2+2N$ edges.
Furthermore, we can inform to all $\mathcal N_i^1,\ldots, \mathcal N_i^{O(p)}$ so that 
for a pair $(v,u)$, the second element $u$ belongs to which part of $T_i^1,\ldots,T_i^{O(p)}$ as we did in the previous decomposition phase.

Although we described the implementation of decomposing the vertex set into $T_i^1,\ldots,T_i^{O(p)}$ by fixing one chunk $S_i$,
we can simultaneously do the process for every chunks $S_i$'s in a constant round by \Cref{obs:parallelism}.
Furthermore, since the machines $\mathcal N_i^1,\ldots, \mathcal N_i^{O(p)}$ are consecutive and each machine $\mathcal N_i^j$ knows for a pair $(v,u)$ in $\mathcal N_i^j$, the index $k\in [O(k)]$ with $u\in T_i^k$, the machines can send their pairs $(v,u)$ to the machine $\mathcal N_i^k$ with $v\in S_i$ and $u\in T_i^k$ in one round.
This is because the number of such pairs is at most $M/p^2+2N\leq L$.
Then at termination, each machine $\mathcal N_i^j$ has the information of all $M/p^2+2N$ edges between $S_i$ and $T_i^j$ for each $i\in [O(p)]$ and $j\in [O(p)]$.
Therefore, $\mathcal N_i^j$ has all the information of the edges and vertices of the induced subgraph $H_i^j$ of $F$ induced by the edges.
Since the vertices of $H_i^j$ is subset of $S_i\cup T_{i}^j$ of size at most $2N/p$, it has at most $2N/p$ vertices.
This completes the $O(1)$ rounds implementation in the MPC model for the edge distribution.
\end{proof}

\subsection{Weighted Spanners in Linear MPC and Congested Clique}\label{sec:weighted_spanner_linearMPC}
In this section, we extend our unweighted spanner algorithm to weighted graphs. Thus, we let $G$ be an edge weighted graph with $n$ vertices. Furthermore, we assume that the weights of edges are non-negative and the maximum weight is a polynomial in $n$.

There is a standard reduction from weighted graphs to unweighted graphs, by bucketing the edges according to weight into $O(\log n)$ graphs: $E_i := \{e\in E : 2^ i\le w(e) < 2^{i+1}\}$ for $i\in [0,O(\log n)]$. We then compute an $O(k)$-spanner $H_i$ on each of these as an unweighted graph. We observe that $H_i$ is a $2\cdot O(k)=O(k)$ spanner of the weighted graph $G_i$. Hence we obtain an $O(k)$-spanner of $G$ of size $O(n^{1+1/k}\log n)$ in constant rounds. Note that we need to run $O(\log n)$ algorithms in parallel now, leading to an $O(\log n)$ overhead in the number of machines \emph{and} the size. See also \cite{podc2021spanner}, where they detail this using their randomized algorithm.

However, there exist more involved reductions that do not incur a $\log n$-overhead. In particular, the following lemma shows that the deterministic Congested Clique reduction of Chechik and Zhang~\cite[Section 5]{chechik2022constant} also works in the MPC model. This reduction is a distributed implementation of a version in the sequential setting, see, e.g.,~\cite{alstrup2022constructing,elkin2018efficient,le2022near}.

\begin{lemma}\label{lm:weighted_spanner_reduction}
    Let $\eps \in (0,1/6]$ be a constant. Given a deterministic constant-round linear MPC algorithm $\A$ using $O(m+n)$ total space that computes an $\alpha$-spanner of size $O(n^{1+1/k})$ in an \emph{unweighted} graph, there exists a deterministic constant round linear MPC algorithm using $O\left(\frac{\log 1/\varepsilon}{\varepsilon}(m+n)\right)$ total space that computes an $(1+\eps)\alpha$-spanner of size $O\left(\tfrac{\log 1/\eps}{\eps}n^{1+1/k}\right)$ in a \emph{weighted} graph.
\end{lemma}
\begin{proof}
    We divide the edge set into $\mu_\eps := \ceil{\log_{1+\eps} 1/\eps}$ sets $E^1, \dots, E^{\mu_\eps}$. Each such set is defined as the union $E^\sigma:=\bigcup_{i\ge 1} E^\sigma_i$, where
    \begin{equation*}
        E^\sigma_i := \{e\in E : \tfrac{L_i}{1+\eps} \le w(e) < L_i\}\text{, where }L_i:= L_0/\eps^i,\ L_0:= (1+\eps)^\sigma. 
    \end{equation*}

    We show how to compute a spanner $H^\sigma$ of $G^\sigma =(V,E^\sigma)$ of size $O(n^{1+1/k})$, for each $\sigma\in [\mu_\eps]$. Then we output $H=\bigcup_\sigma H^\sigma$, which has size $O(\mu_\eps n^{1+1/k})= O\left(\tfrac{\log 1/\eps}{\eps}n^{1+1/k}\right)$. 

    So now we consider $G^\sigma$. Before we describe the algorithm, consider the following claim. 

    \begin{claim}[Lemma 5.1 from \cite{chechik2022constant}]\label{claim:stars_in_tree}
        For any tree $T = (U,F)$ that contains more than a single vertex, it can be divided into a set of at most $|U |/2$ vertex-disjoint star subgraphs, and the computation is done in linear time $O(|U|)$ under the classical centralized model.
    \end{claim}

    Compute a minimum spanning tree $T\subseteq G^\sigma$ deterministically in constant rounds~\cite[Corollary~2.4]{nowicki2021deterministic}. Note that if $G^\sigma$ is not connected, we can make it connected by adding some virtual edges of large weight. Since $T$ is a tree, it is of size $O(n)$ and we can make it known to every machine in $O(1)$ rounds as follows.
    In total, we only need $O(n)$ memory to describe the tree. In $1$ round, every machine can send the edges it knows after computing $T$ to a designated leader. This leader than distributes all tree edges over the machines. There are $m/n$ machines, each with memory $O(n)$. The leader sends each machine $n/(m/n)$ tree-edges in $1$ round. Now each machine sends its $n/(m/n)$ edges to all $m/n$ other machines in $1$ round, hence every machine receives and sends $O(n)$ edges in total -- as required.  
    

    Now each machine internally creates the following sequence of trees $T= T_0, T_1, T_2, \dots, T_{O(\log n)}$. Each tree $T_{i+1}$ is some contraction of the previous tree $T_{i}$, defined by the following inductive procedure. Consider the forest $F_i$ spanned by the edges $E(T_i) \setminus \bigcup_{j>i}E_j^\sigma$. For each non-singleton tree $T'\in F_i$, apply \Cref{claim:stars_in_tree}, which divides $T'$ into at most $|V(T')|/2$ stars. Then for each such star, we contract it into a single vertex. The tree resulting from these contractions is the next tree $T_{i+1}$. Note that this can be stored in $O(n)$ space, since each tree has size  $|V(T_{i+1})|\le |V(T_i)|/2$. So $\sum_i |V(T_i)|= O(n)$. 

    For each $i\ge 0$, let $V_i \subseteq V(T_i)$ denote all the vertices from $F_i$ which are not single vertices. Define $G_i^\sigma:= (V_i,E_i^\sigma)$, where edges contracted within the same vertex are removed, and parallel edges are removed. Now, we want to run algorithm~$\A$ on this graph to compute an unweighted $\alpha$-spanner $H_i^\sigma$ using $O(|E_i^\sigma|+|V_i|)$ total space. We first argue the total space and then detail how to formulate the input and execute $\A$. 

    Next, we show that any edge $e\in E$ appears in at most two sets $E_i^\sigma$ (ranging over $i$ \emph{and} $\sigma$). Suppose $e\in E_i^\sigma$, consider any other set $E_j^\tau$. We show that there is only one possible combination of $j$ and $\tau$, such that $e\in E_j^\tau$. To investigate this, let $j$ and $\tau$ be arbitrary, and suppose $e\in E_j^\tau$. If $i=j$, $\{E_i^\sigma\}_\sigma$ clearly forms a partition.  
    So, without loss of generality, we assume $j<i$. For the two regions to overlap, we must have that 
    \begin{align*}
        \frac{(1+\eps)^{\tau}}{\eps^j} \ge \frac{(1+\eps)^{\sigma-1}}{\eps^i}
    \end{align*}
    

    This implies that $\eps^{i-j} \ge (1+\eps)^{\sigma-\tau -1}$ and hence
    $i-j \le (\sigma-\tau -1)\log_\eps (1+\eps)$ (since $\eps<1)$.
   This implies that (note that $\log_\eps (1+\eps)$ is a negative number)
    \begin{align*}
        \sigma-\tau-1 < (i-j)/\log_\eps (1+\eps)= (i-j)\log_{1+\eps}\eps ;
    \end{align*}
    or equivalently $\tau-\sigma +1 >  (i-j)\log_{1+\eps}(1/\eps)$.
    Simultaneously, we have $\tau-\sigma +1 \le \mu_\eps +1-1 = \lceil \log_{1+\eps}(1/\eps)\rceil $. We immediately see that for $i\ge j+2$ we get a contradiction. For $i=j+1$, we need $\tau =\mu_\eps$ and $\sigma=1$. This means that only edges from $E_i^{1}$ could also appear in $E_{i-1}^{\mu_\eps}$, and all other edges appear exactly once. 

    
    Hence, we have that $\sum_i O(|E_i^\sigma|+|V_i|)=O(m+n)$.
    Therefore, we can compute all $H_i^\sigma$'s for each $i$ and $\sigma\in[\mu_\varepsilon]$ in parallel on the linear MPC model with $O(\mu_\varepsilon \cdot (m+n))$ total space. 
    We output $H^\sigma:= \bigcup_i H_i^\sigma$. And finally output $H= T \cup \bigcup_\sigma H^\sigma$. 

    Next, we elaborate on how we deploy algorithm $\A$ on $G_i^\sigma$ to compute $H_i^\sigma$. First, we can resort all edges using \Cref{lem:sorting_prefixsum_MPC} to delete parallel edges within $G_i^\sigma$. Next, we use \Cref{obs:parallelism} to run all calls of $\A$ for all $i,\sigma$ in parallel in a constant number of rounds. 


    To show that $H$ indeed has stretch $(1+\eps)\alpha$ with $O\left(\tfrac{\log 1/\eps}{\eps}n^{1+1/k}\right)$ edges, we refer to \cite[Lemma~5.4 and Corollary~5.1]{chechik2022constant}. 
\end{proof}

Combining \Cref{lm:weighted_spanner_reduction} with \Cref{thm:linear_MPC_spanner_unweighted}, we obtain the following result. 

\CorLinearMPCSpannerWeighted*

\subsection{Weighted Spanners in Sublinear MPC}\label{sec:weighted_spanner_sublinear}
In this section, we let $G$ be an edge weighted graph with $n$ vertices. 
Our goal is to derandomize the spanner construction algorithm of Biswas, Dory, Ghaffari, Mitrović, and Nazari~\cite{biswas2021massively} in the sublinear MPC model.
The following theorem summarizes this section. 
\ThmSublinearMPCSpannerSimplified*

Notably, the authors said that their randomized algorithm returns a spanner with high probability by the algorithm in the near-linear total space $\tilde O(m)$, instead of $O(m+n)$ total space
~\cite[Theorem 1.1]{biswas2021massively}.
    This is because they used expected size analysis for the obtained spanner by their randomized algorithm, and then they turn the expected size guarantee into a high probability bound using a Chernoff-bounds argument.
    To do this, they allowed $O(\log n)$ overhead on the overall memory to repeat the randomized algorithm in parallel.
    Here, we design a deterministic algorithm, and thus, we do not need the overhead for the repetition. 
    Therefore, our algorithm works with $O(m+n)$ total space.

Briefly, \cite{biswas2021massively} gave a hierarchical algorithm to compute a clustering graph and merge adjacent clusters by computing a hitting set obtained by random sampling.
Here, we adapt our deterministic hitting set algorithm instead of the sampling algorithm to derandomize their algorithm.
Precisely, we use the hitting set algorithm that gives a suboptimal size in constant rounds as illustrated by the first parts of \Cref{thm:MPC hitting set intro: linear,thm:MPC hitting set intro: sublinear}.
Here, the suboptimal size of the hitting set computed at each of the hierarchy levels does not affect either the final computed spanner size or the total round complexity asymptotically.
This is because we rebalance the depth of the hierarchical structure of the algorithm, which has a trade-off between the total round complexity and the returned spanner size.
Briefly, we increase the hierarchical levels until we reach the desired spanner size, but do not affect the total round complexity.

\paragraph*{Randomized of algorithm~\cite{biswas2021massively}.}
The randomized algorithm initially starts with the clustering graph $G^{(0)}=G$, where each cluster corresponds to a single vertex of $G$.
Furthermore, it sets $\textsf{E}=E(G)$ and $E_S=\emptyset$.
During the process, the algorithm removes some edges from $\textsf E$ while inserting some edges into $E_S$, and finally returns a spanner consisting of the edges $E_S$.
The algorithm recursively repeats $\ell$ \emph{epochs}, where each epoch also recursively performs $t$ \emph{iterations}.
We define the parameters $\ell$ and $t$ later with respect to $k$ and $\varepsilon$ to compute a spanner with stretch $O(k^{1+\varepsilon})$. 
The $i$th epoch grows the cluster by merging clusters in $G^{(i-1)}$ to obtain another clustering graph $G^{(i)}$.
Recall that each vertex, that is a cluster, of a clustering graph $G^{(i)}$ corresponds to its center vertex and an edge $\{c,c'\}$ in $G^{(i)}$ corresponds to an edge in $G$ whose two end vertices belong to the clusters of $c$ and $c'$, respectively.
During the algorithm, we suppose that the edge $\{c,c'\}$ corresponds to the minimum weight edge between two clusters of $c$ and $c'$ in the current remaining edge set $\textsf E$.
Furthermore, the weight of an edge $\{c,c'\}$ is the weight of its corresponding edge in $\textsf E$. 

Each $i$th epoch has two phases. {The first phase inserts edges to $E_S$ by $t$ iterations of Baswana-Sen~\cite{baswana2007simple} algorithm on $G^{(i-1)}$ with probability $p(i)$, that is an algorithm constructing a $(2t-1)$-spanner of $G^{(i-1)}$ by selecting $O(tn_{i-1}/p(i))$ edges, with high probability, in $t$ iterations.
The second phase contracts the clusters formed by the selected edges during the first phase, and returns the obtained contracted $G^{(i)}$.} 

The Baswana-Sen algorithm for $i$th epoch runs on $G^{(i-1)}$ and decomposes it into clusters in iterations.
It initializes the set $C_0$ of clusters as the family of the singletons $\{v\}$  for every $v\in V(G^{(i-1)})$. 
Then, at the $j$th iteration of the $i$th epoch, it obtains the clusters for $C_j$ from the input $C_{j-1}$ by merging some clusters.
In the following, we refer to the nodes on $G^{(i-1)}$ as \emph{super-nodes},  while the \emph{clusters} refer to the clusters in $C_j$'s. 

Each $j$th iteration randomly samples each cluster from $C_{j-1}$ with probability $p(i)$ to construct a set $D_j\subseteq C_{j-1}$. Then for each super-node $v$ in $V(G^{(i-1)})$, let $c^*$ be the closest cluster among $D_j$ that is adjacent to $v$ in $G^{(i-1)}$ connected by the current edges in $\textsf E$.\footnote{We define a proper $p(i)$ later, that is related to the size of the output spanner and round complexity of the algorithm.}
Then 
we add the shortest edge $\{u,u'\}\in \textsf E$ to $E_S$ that is between two vertices $u$ and $u'$ belonging to the super-node $v$ and the cluster $c^*$, respectively, while deleting all the edges between $v$ and $c^*$ from $\textsf E$.
Furthermore, for every adjacent, but not selected, clusters $c\in C_{j-1}\setminus D_j$ to the super-node $v$ in $G^{(i-1)}$ by an edge in $\textsf E$ shorter than $w(u,u')$ to $v$, we analogously insert the shortest edge in $\textsf E$ between the super-node $v$ and the cluster $c$ to $E_S$, and remove all the edges between $v$ and $c$ from $\textsf E$.
If there is no adjacent cluster in $D_j$ to $v$, then we insert the shortest edges in $\textsf E$ between the super-node $v$ and the cluster $c$ to $E_S$ for every adjacent cluster $c\in C_{j-1}$ to $v$ in $G^{(i-1)}$, and remove all the edges incident to $v$ from $\textsf E$.
The $j$th clusters of $C_j$ are formed by taking each cluster in $D_j$, and then extending it using the edges newly inserted to $E_S$ during the $j$th iteration, to absorb the super-nodes connected to $D_j$ by $E_S$.
Note that during each iteration, we do not contract the super-nodes in $G^{(i-1)}$, we just extend the clusters $C_j$'s step by step.

After $t$ iterations, the $i$th epoch terminates and returns $G^{(i)}$ obtained by contracting each cluster in $C_t$ from $G^{(i-1)}$. 
Then the super-nodes in $G^{(i)}$ correspond to the clusters in $C_t$.
If $G^{(i)}$ has parallel edges due to the contracting phase, we only save the shortest edge and remove the others from $G^{(i)}$ and $\textsf E$.

After $\ell$ epochs, in the obtained $G^{(\ell)}$, we consider all remaining edges in $\textsf E$ for the \textbf{Final phase}.
For each super-nodes $v$ and $u$ in $G^{(\ell)}$, if there is a remaining edge in $\textsf E$ between two vertices belonging to the super-nodes $v$ and $u$, respectively, then we choose the shortest edge among them and insert it into $E_S$.

The quality of the output $E_S$ is guaranteed by the following lemma.
Precisely, we can obtain \Cref{lem:analysis_spanner_sublinear_random} and \Cref{cor:parameters_sublinear_spanner} by restating the previous lemmas and theorems of~\cite[Lemma 4.12, Theorem 5.11, and Lemma 5.14]{biswas2021massively}.
{The original statements in~\cite{biswas2021massively} described the special case when $\ell=\lceil \log_{t+1} k\rceil$ and $p(i)=q(i)=n^{-\frac{(t+1)^{i-1}}{k}}$ for $i\in[\ell]$ which are satisfying all inequalities of \Cref{cor:parameters_sublinear_spanner}.
The authors showed that such $p(i)$'s and $q(i)$'s satisfy the conditions of \Cref{lem:analysis_spanner_sublinear_random} with high probability for all $i\in [\ell]$, and thus, the final returned subgraph induced by $E_S$ is a spanner with $O(t\cdot n^{1+1/k})$ edges.
}
\begin{lemma}\label{lem:analysis_spanner_sublinear_random}
    By repeating $\ell$ epochs with $t$ iterations,
    the subgraph of $G$ induced by the finally obtained edges in $E_S$ is an $O((2t+1)^{\ell})$-spanner.
    Furthermore, if there are $0<p(i),q(i)\leq 1$ for all $i\in[\ell]$ so that
    during every $j$th iteration of $i$th epoch:
    \begin{itemize}
        \item The size of $D_j$ (and $C_j$) is at most $O(|C_{j-1}|\cdot q(i))$, and
        \item For each $v\in G^{(i-1)}$ adjacent to at least $1/p(i)$ clusters of $C_{j-1}$ in $G^{(i-1)}$, $D_j\cap S_v\neq \emptyset$, 
    \end{itemize}
    where $S_v$ is the set of $1/p(i)$ closest adjacent clusters of $v$ among $C_{j-1}$, then the following holds:
\begin{itemize}
    \item The $i$th epoch returns $G^{(i)}$ with  $n_i=O\left(n\cdot \left(\prod_{x=1}^{i}(q(x))^{t}\right)\right)$ super-nodes,
    \item The $i$th epoch inserts $O\left(t\cdot n_{i-1}/p(i)\right)$ edges to $E_S$, and
    \item The \textbf{Final phase} inserts $O(n_\ell^2)=O\left(\left(n\cdot \left(\prod_{x=1}^{\ell}(q(x))^{t}\right)\right)^2\right)$ edges to $E_S$.
\end{itemize}
\end{lemma}


\begin{proof}[Sketch of the proof]
    Here, we consider the simple case with $t=1$.
    The $i$th epoch starts with $G^{(i-1)}$ with $n_{i-1}$ super-nodes and computes $D$ of size $O(n_{i-1}\cdot q(i))$.
    With respect to the clusters in $D$, 
    we insert at most $1/p(i)$ edges to $E_S$ per super-node of $G^{(i-1)}$, and we obtain $G^{(i)}$.
    Therefore, the number $n_i$ of the super-nodes in $G^{(i)}$ is at most $O(n_{i-1}\cdot q(i))$, inductively, it is at most $O(n\cdot q(i)q(i-1)\dots q(1))$.
    And the number of edges inserted by $i$th epoch is $O(n_{i-1}/p(i))$.
    After $\ell$ epochs, since $G^{(\ell)}$ has $n_\ell$ super-nodes, the \textbf{Final phase} inserts $O( n_\ell^2)$ edges. 
    
    For $t\geq 2$, each epoch $t$ times reduces the size of $D_j$'s recursively so that the $G^{(i)}$ has at most $O(n_{i-1}\cdot (q(i))^t)$ nodes by inserting $O(t\cdot n_{i-1}/p(i))$ edges to $E_S$.

    
    For the stretch, as proven in the previous paper~\cite[Lemma 5.8 and Theorem 5.11]{dory2024massively}, an edge with weight $w$ is removed during the $i$th epoch between two super-nodes $u$ and $v$ in $G^{(i-1)}$ only if there is a path between them in $G^{(i-1)}$ consisting of at most $(2t+1)$ edges of weight at most $w$.
    Analogously, for the \textbf{Final phase}, if an edge was removed, then the two end super-nodes in $G^{(\ell)}$ have a shorter edge which is contained in $E_S$.
    These claims guarantee that the recursive $\ell$ epochs give at most $O((2t+1)^\ell)$ stretch.
    \end{proof}

 {The lemma implies that when we increase the parameters $t$ and $\ell$, then the obtained spanner has a larger stretch while the number of edges inserted by \textbf{Final phase} is decreased.
 Furthermore, even if $q(i)$ is slightly increased, by rebalancing the $p(i)$'s and repeating $t$ iterations, we can still guarantee that a sufficiently small portion of clusters remain at the end of the epoch.
 Therefore, the obtained spanner is sufficiently small. 
 \Cref{cor:parameters_sublinear_spanner} formalizes it.}
 
\begin{corollary}\label{cor:parameters_sublinear_spanner}
    Let $k> 1$ be a constant. The number of edges in $E_S$ is at most $O(t\cdot n^{1+1/k})$ if:
     \begin{enumerate}
        \item $\sum_{i=1}^\ell\prod_{x=1}^{i-1}(q(x))^{t}/p(i)\leq n^{1/k}$,
        \item $\prod_{x=1}^{\ell}(q(x))^{t}\leq n^{-0.5+1/(2k)}$.
    \end{enumerate}
\end{corollary}
\begin{proof}
    By \Cref{lem:analysis_spanner_sublinear_random}, the number of edges inserted by $\ell$ epochs is bounded by:
    \begin{align*}
        \sum_{i=1}^\ell O\left(\frac {n_i\cdot t}{p(i)}\right)&\leq \sum_{i=1}^{\ell} O\left(n\cdot \left(\prod_{x=1}^{i-1}(q(x))^{t}\right)\cdot \frac{t}{p(i)}\right)\\
        &=O\left(tn \sum_{i=1}^\ell\prod_{x=1}^{i-1}\frac{(q(x))^{t}}{p(i)}\right)\leq O(t\cdot n^{1+1/k}).
    \end{align*}
    The last inequality holds by \textsf{Condition 1}.
    Furthermore, \textsf{Condition 2} implies that the \textbf{Final phase} inserts at most $O(n^{1+1/k})$ edges.
    Therefore, the total size of $E_S$ is at most $O(t\cdot n^{1+1/k})$.
\end{proof}
    

{
In the following, we set $p(i)=n^{-\frac{(0.7t+1)^{i-1}}{k}}$ and $q(i)=p(i)^{0.99}$, unlike the randomized one set $p(i)=n^{-\frac{(t+1)^{i-1}}{k}}$ and $q(i)=p(i)$,
so that we can achieve $D_j$'s satisfying the conditions in \Cref{lem:analysis_spanner_sublinear_random} by the deterministic hitting set algorithms of~\Cref{thm:MPC hitting set intro: linear,thm:MPC hitting set intro: sublinear}.
Note that the hitting set algorithm gives non-optimal sized $D_j$'s for each of the iterations unlike to the previous randomized algorithm that assumed the optimal size with $q(i)=p(i)$. 
However, those parameters still satisfy all the inequalities of \Cref{cor:parameters_sublinear_spanner} by slightly increasing $\ell:=\lceil\log_{0.5t+1} k\rceil$ compared to the original randomized algorithm that  set $\ell:=\lceil\log_{t+1} k\rceil$.
That means we increase the depth of the hierarchical structures by a constant factor.
}


\paragraph*{Derandomization.}
We derandomize the algorithm of~\cite{biswas2021massively}.
Note that during their algorithm, randomness only occurs for sampling $D_j$'s in each iteration of every epoch. We will replace this step with the hitting set algorithms of \Cref{thm:MPC hitting set intro: linear,thm:MPC hitting set intro: sublinear}.
Precisely, at the $j$th iteration (of the $i$th epoch),
we are given $C_{j-1}$ and claim to compute the hitting set $D_j$ by solving the hitting set problem defined by the $N:=|V'|$ sets $S_v:=$ the $1/p(i)$ closest clusters of $C_{j-1}$ in $G^{(i-1)}$ for each $v\in V'\subseteq  V(G^{(i-1)})$, and the universe $U:=C_{j-1}$. 
Here, $V'$ is the set of vertices of $G^{(i-1)}$ adjacent to at least $1/p(i)$ clusters of $C_{j-1}$ in $G^{(i-1)}$ and $S_v$'s are of size $d:=1/p(i)$.
\kyungjini{Then we obtain a hitting set of size $O(|C_{j-1}|\log |V'|/d^{0.99})$.}

\kyungjini{new analyses below}
When the size $1/p(i)$ of the sets $S_v$'s is at most $n^{\delta}$, then we are essentially in the `linear memory' case for the hitting set: every set fits in one machine. In this case, we use the algorithm of \Cref{thm:linear MPC hitting set large N} with $L:=n^{\delta}$ local space.\kyungjin{correct reference}
If $n^\delta<1/p(i)\leq n$, we are in the truly sublinear case. Hence, we apply the algorithm of \Cref{thm:sublinear_MPC_hitting set} with $L:=(1/p(i))^{\delta'}=n^{\delta}$ local space by setting some constant $\delta'=\delta/\log(1/p(i))$.\kyungjin{correct reference}
In both of the cases, the output hitting set $D_j$ has a size at most $O(|C_{j-1}| p(i)^{0.99}\log_{1/p(i)} |V'|)$ and satisfies $D_j\cap S_v\neq\emptyset$ for any $v\in V'$.
Furthermore, the round complexity is $O(\log_{1/p(i)}\frac{|V'|}{|C_{j-1}|})$ using the MPC model with $n^{\delta}$ local space and $O(m+n)$ total space. 
This is because $N=|V'|$ and $U=C_{j-1}$, and thus the following holds for a sufficiently large $n$ $$d^{d^{\delta'/4}}=d^{n^{\delta/4}}\geq n\geq N=|V'|.$$
Recall that $n$ and $m$ refer to the number of vertices and the edges, respectively, of the entire input graph $G$.

\paragraph*{Round complexity and the space complexities.}
In \cite{biswas2021massively}, each iteration is implemented by $O(1)$ rounds in the sublinear MPC algorithm of $O(n^{\delta})$ local space and $O(m+n)$ total space~\cite[Lemma 6.1 and Theorem 1.1]{biswas2021massively}, including the sampling process.
Notably, it also supports organizing the remaining edge set $\textsf E$ with respect to the clusters (and the super-nodes) belonging to its end vertices and the weight while the clusters (and the super-nodes) are updated.
\kyungjini{new anlaysis}
Therefore, we can formulate and deterministically solve the hitting set problem above, computing $S_v$'s and $V'$, at each iteration in $O(\log_{1/p(i)}\frac{|V'|}{|C_{j-1}|})$.

For parameters $k$ and $\varepsilon$, we can return $O(tk^{1+\varepsilon})$-spanner of $G$ deterministically by setting $\ell:= \lceil\log_{0.5t+1} k\rceil$, $t:=\lceil2\cdot (4^{1/\varepsilon}-1)\rceil$, and $p(i):=n^{-\frac{(0.7t+1)^{i-1}}{k}}$.
Here, we can derive~\Cref{thm:sublinear_MPC_spanner} from \Cref{lem:analysis_spanner_sublinear_random} and \Cref{cor:parameters_sublinear_spanner} by setting the defined $\ell$, $t$, $p(i)$, and $q(i)=p(i)^{0.99}\log_{1/p(i)}n$ for $i\in[\ell]$. 
Note that $q(i)\leq p(i)^{0.98}$ because the following holds for sufficiently large $n$ $$\log_{1/p(i)}n=1/\log_n(1/p(i))\leq k\leq n^{1/100k}\leq (1/p(1))^{0.01}\leq (1/p(i))^{0.01}.$$
Details are in the proof.

\begin{restatable}{theorem}{ThmSublinearMPCSpanner}\label{thm:sublinear_MPC_spanner}
    For a constant $\delta<1$, given a weighted graph $G$ on $n$ vertices and positive parameters $k\geq 1$ and $0<\varepsilon\leq 1$, we can deterministically construct a $O(tk^{1+\varepsilon})$-spanner of $G$ with $O(t\cdot n^{1+1/k})$ edges in $O\left(\frac{t^2\log k}{\log t}\right)$ rounds in the sublinear MPC model with $O(n^{\delta})$ local space and $O(m+n)$ total space, where $t:=\lceil2\cdot (4^{1/\varepsilon}-1)\rceil$.
\end{restatable}\kyungjin{round complexity is increased}
Note that \Cref{thm:sublinear_MPC_spannerSimplified} is a special case when $\varepsilon$ (and hence $t$) is a fixed constant.
Furthermore, we can choose the parameters $\varepsilon$ and $k$ such that we obtain a spanner:

\begin{enumerate}
    \item  with $k^{1+o(1)}$ stretch and $O(n^{1+1/k}\log k)$ size in $O(\tfrac{\log^3 k}{\log \log k})$ rounds;
    \item with $\log^{1+o(1)}n$ stretch and $O(n\log\log n)$ size in $O(\tfrac{\log^3\log n}{\log\log\log n})$ rounds.
    \kyungjini{Here, previousely, the round complexity were $\log^2 k/\log\log k$ and $\log^2\log n/\log\log\log n$. But it was increased.}
\end{enumerate}
\begin{proof}[Proof of \Cref{thm:sublinear_MPC_spanner}]
    Recall that by \Cref{thm:linear MPC hitting set large N} and \Cref{thm:sublinear_MPC_hitting set}, each $j$th iteration for $j\in[t]$ takes $O(\log_{1/p(i)}|V'|/|C_{j-1}|)$ rounds during $\ell$ epoches.
    When $1/p(i)\geq n^{\delta}$, it takes a constant round since $|V'|\leq n$.
    In the other case, we have $|V'|=|V(G^{(i-1)})|$ and $|C_{j-1}|=|V(G^{(i-1)})|(q(i))^{j-1}=|V(G^{(i-1)})|(p(i))^{0.98(j-1)}$.
    Therefore, the $j$-th iteration for $j\in[t]$ takes $O(0.98j)$.
    Totally, the round complexity is $O(t^2\ell)=O(t^2\log k/\log t)$ since we set $\ell:=\lceil \log_{0.5t+1} k\rceil$.

  Furthermore, \Cref{lem:analysis_spanner_sublinear_random} implies that the subgraph of $G$ induced by the finally obtained edges in $E_S$ is a spanner of $G$ with stretch:
    \begin{align*}
        O((2t+1)^{\lceil \log_{0.5t+1} k}\rceil)&\leq O((2t+1)k^{\log_{0.5t+1}(2t+1)})\leq O(tk^{\log_{0.5t+1}(2t+1)})\\
        &\leq O(tk^{1+\log_{0.5t+1} 4})\leq O(tk^{1+\log_{4^{1/\varepsilon}} 4})\\
        &= O(tk^{1+\varepsilon \log_{4} 4})
        =O(tk^{1+\varepsilon}).
    \end{align*}
    The second line holds due to $t:=\lceil 2(4^{1/\varepsilon}-1)\rceil$.
    
    Furthermore, during the $t$ iterations of $i$th epoch, 
    we can set $p(i)=n^{-\frac{(0.7t+1)^{i-1}}{k}}$ and $q(i)\leq p(i)^{0.98}$ as desired in \Cref{lem:analysis_spanner_sublinear_random}.
    Precisely, the $j$th iteration  of epoch $i$ computes $O(|C_{j-1}|q(i))$ sized $D_j$ who has at least one of $1/p(i)$ closest adjacent clusters of $v\in V(G^{(i-1)})$ among $C_{j-1}$ for every vertex $v$ if it is adjacent to at least $1/p(i)$ clusters of $C_{j-1}$ in $G^{(i-1)}$.
    In the following, we show that the inequalities of \Cref{cor:parameters_sublinear_spanner} are satisfied.

    \textsf{Condition 1} of \Cref{cor:parameters_sublinear_spanner} holds by the geometric series as follows:
    \begin{align*}
        \sum_{i=1}^\ell\prod_{x=1}^{i-1}(q(x))^{t}/p(i)&=\sum_{i=1}^\ell O\left(n^{\frac{(0.7 t+1)^{i-1}-\sum_{x=1}^{i-1}0.98t(0.7t+1)^{x-1}}{k}}\right)\\
        &\leq \sum_{i=1}^\ell O\left(n^{\frac{(0.7t+1)^{i-1}-1.4(0.7 t+1)^{i-1}+1.4}{k}}\right)\\
        &\leq \sum_{i=1}^\ell O\left(n^{\frac{-0.4(0.7t+1)^{i-1}+1.4}{k}}\right)\\
        &\leq O(n^{1/k}).
    \end{align*}
    Note that the second line holds due to $0.98/0.7= 1.4$.

    \textsf{Condition 2} holds due to $\ell=\lceil \log_{0.5t+1} k\rceil$ as follows:
    \begin{align*}
        \prod_{x=1}^{\ell}(q(x))^{t}\leq O\left(n^{-\frac{1.4(0.7t+1)^{\ell}-1.4}{k}}\right)&\leq O\left(n^{-\frac{1.4\cdot \frac{0.7t+1}{0.5t+1}(0.5t+1)^{\ell}-1.4}{k}}\right)\\
        &\leq O\left(n^{-\frac{1.58k}{k}+1.4/k}\right) \\
         &\leq O(n^{-0.5+0.5/k} )\\
    \end{align*}
    The first inequality holds since $\ell\geq 1$.
    Additionally, the second one holds since $t\geq 1$.
    The last inequality holds since $k$ is at least 1, and thus, $0.9/k$ is at most 0.9.
    Therefore, \Cref{cor:parameters_sublinear_spanner} implies that the size complexity of \Cref{thm:sublinear_MPC_spanner} holds.
\end{proof}

\newpage
\section{Application II: APSP Algorithms}\label{sec:applications_shortest_paths}\label{sec:apsp_corollary}
In this section, we give deterministic algorithms for the approximate shortest paths problems including weighted APSP.
In the previous section, we designed deterministic spanner algorithms. 
These spanner algorithms can be used for APSP in weighted graphs as follows. 

\paragraph*{$O(\log n)$-approximate APSP in linear MPC and Congested Clique.}
\cref{thm:linear_MPC_spanner_weighted} gives an $O(\log n)$-spanner of size $O(n)$. Since this fits in one machine, we can obtain the following APSP result in linear MPC and Congested Clique. In Congested Clique, the output is explicit: the full distance table can be computed within a node. In the linear MPC model, the output is implicit: any distance query can be asnwered without additional communication rounds.  
\begin{restatable}{corollary}{ThmAPSPlogn}\label{cor:APSPlogn}
    There exists a deterministic algorithms that, given a weighted graph $G=(V,E,w)$ with $w\colon E \to [\poly n]$ on $n$ vertices, computes $O(\log n)$-approximate APSP in $O(1)$ rounds in the linear MPC model with $O(m)$ total space or in Congested Clique.
\end{restatable}
\begin{proof}
    First, we consider the linear MPC model.  In constant rounds, we compute an $O(\log n)$-spanner $H$ of size $O(n)$ using \Cref{thm:linear_MPC_spanner_weighted}. Next, using \Cref{lem:sorting_prefixsum_MPC} we can store all $O(n)$ edges of the spanner on one machine. This machine can now answer any distance query $d_G(u,v)$ without additional communication, by outputting $d_H(u,v)$, which is an $O(\log n)$-approximation. This can de done with any linear-space SSSP algorithm, e.g., Dijkstra's algorithm, since $H$ has size $O(n)$. 

    Next, we consider the Congested Clique.     
    In constant rounds, we compute an $O(\log n)$-spanner $H$ of size $O(n)$ using \Cref{thm:linear_MPC_spanner_weighted}. In $O(1)$ rounds, we can make $H$ known to every machine using Lenzen's routing~\cite{Lenzen13}, since each node only sends and receives $O(n)$ messages. 
    Then, each node can internally compute APSP on $H$, which gives an $O(\log n)$-approximation of APSP on $G$.      
\end{proof}

\paragraph*{$(\log^{1+o(1)} n)$-approximate APSP using a single near-linear machine.}
\Cref{thm:sublinear_MPC_spanner} provides a spanner with $\tilde O(n)$ edges and $\log^{1+o(1)}n$ stretch in the sublinear MPC model.
Therefore, if we allow a single machine with $\tilde O(n)$ local space, then we can immediately compute a $\log^{1+o(1)} n$-approximation for the distance between any two vertices.
The proof of \Cref{cor:heterogeneous_simpleAPSP} is analogous to \Cref{cor:APSPlogn}.
Furthermore, we can simulate the algorithm of \Cref{cor:heterogeneous_simpleAPSP} in the near-linear MPC model with $\tilde O(n)$ local space and $O(n+m)$ total space by~\Cref{obs:parallelism}.

\begin{corollary}\label{cor:heterogeneous_simpleAPSP}
    For a constant $\delta<1$,
    there exists a deterministic algorithm that, given a weighted undirected graph $G=(V,E,w)$ with $n$ vertices and $m$ edges, computes $\log^{1+o(1)} n$-approximate APSP in $O(\log^2\log n)$ rounds in the MPC model with a single machine with $\tilde O(n)$ local space and additional machines with $O(n^\delta)$ local space while the total space is $O(n+m)$.
\end{corollary}

In the following, we give \emph{constant}-approximate distance algorithms.
In \Cref{sec:applications_constant_APSP_CC}, we give an approximation algorithm for weighted APSP in Congested Clique.
Then, in \Cref{sec:applications_heterogeneous}, we give algorithms for unweighted approximate SSSP, MSSP, and APSP in the near-linear MPC model.



\subsection{A Constant Approximate Weighted APSP in the Congested Clique}\label{sec:applications_constant_APSP_CC}
In the Congested Clique model, we can compute a $O(1)$-approximate APSP in $O(\log\log\log n)$ rounds.
\ThmAPSPinCC*
\begin{proof}
    Bui, Chandra, Chang, Dory, and Leitersdorf~\cite{BuiCCDL24} give an algorithm that computes a constant approximation of APSP in $O(\log \log \log n)$ rounds in Congested Clique. Their algorithm uses two randomized subroutines that can be used as a blackbox: 1) A hitting set; and 2) A spanner.
    For each of these, we detail what the requirements are, and how our algorithms satisfy these.
    
    \begin{enumerate}
        \item Hitting set.\\
        In \cite[Lemma 6.1]{BuiCCDL24}, we have set $\tilde N_k(u)\subseteq V$ for every $u\in V$. The sets $\tilde N_k(u)$ are known to the node $u$ and of size at least $k$. 
        We need to find a hitting set $V_S$ for these sets. We need to do this four times, using our linear MPC hitting sets (\Cref{thm:linear MPC hitting set} and \Cref{thm:optimal_linear MPC hitting set}) and \Cref{lem:MPC_to_CC}. \label{item:1}\kyungjin{correct reference here, and below}
        \begin{enumerate}[(i)]
            \item  When applying this in \cite[Lemma 2.1]{BuiCCDL24}, we set $k=n^{2a^{-1/4}}$ for some parameter $a\ge 1$, and need that the size of the hitting set is bounded by $|V_S| \le n^{1-a^{-1/2}} $. We can compute a hitting set of size $|V_S|=O(n/k^{0.99})$ using \Cref{thm:linear MPC hitting set} in constant rounds. We see that $n/k^{0.99} = n^{1-0.99\cdot 2a^{-1/4}} \le  n^{1-a^{-1/2}}$ whenever $0.99\cdot 2a^{-1/4}\ge a^{-1/2}$, or equivalently $a\ge (0.99\cdot 2)^{-4}$. The latter is satisfied since $a\ge 1$ \label{item:1i}
            \item  When applying this in \cite[Lemma 7.2]{BuiCCDL24}, we set $k=\sqrt n$ and need $|V_S|^{3/2}= O(n)$. Again, using \Cref{thm:linear MPC hitting set} in constant rounds, we obtain a hitting set of size $O(n/k^{0.99})$. Filling in that $k=\sqrt n$ gives $|V_S| = O(n^{0.505})$, and hence $|V_S|^{3/2}=O(n^{0.7575})$, which is less than $O(n)$ as required. \label{item:1ii}
            \item When applying this in \cite[Lemma 7.3]{BuiCCDL24}, we set $k=\sqrt n$ and need $|V_S|^{3/2} = O(n)$. We obtain this by the same settings as above. Note that in the original lemma, they get that $|V_S|^2$ small enough that they can collect it in one node and compute APSP exactly. Our hitting set is slightly bigger ($O(n^{0.7575})$ instead of $\tilde O(\sqrt n)$), but small enough that we can compute a constant-approximate spanner of size $O(n)$ with \Cref{thm:linear_MPC_spanner_weighted} and then compute APSP on the spanner. This extra constant affects the final constant, but has no impact on the running time. \label{item:1iii}
            \item When applying this in \cite[Theorem 7.1]{BuiCCDL24}, we set $k=\log^4 n$ and need $|V_S| = O(n\log k/k)$. We obtain this size with \Cref{thm:optimal_linear MPC hitting set} in $O(\log \log k)=O(\log \log \log n)$ rounds. We note that this only needs to be done once, so leads to an additive $O(\log \log \log n)$ in the running time. \label{item:1iv}
        \end{enumerate}    
        
        \item Constant-approximation spanner.\label{item:2} \\ 
        In \cite[Lemma 7.1]{BuiCCDL24}, they cite the randomized spanner from \cite[Theorem 1.2]{chechik2022constant}. They need this spanner to have constant approximation $O(k)$ and size $O(n^{1+1/k})$. Our deterministic \Cref{thm:linear_MPC_spanner_weighted} fulfills these criteria and takes constant rounds. 
    \end{enumerate}
    Since subroutines \ref{item:1} \cref{item:1i,item:1ii,item:1iii} and \ref{item:2}   take constant rounds and are called at most $O(\log\log \log n)$ times and \ref{item:1} \cref{item:1iv} takes $O(\log\log \log n)$ rounds and is called once, we obtain total round complexity $O(\log \log \log n)$.
\end{proof}

We note that current randomized spanners that take $O(1)$ rounds give a better constant approximation than our deterministic constant round spanner~\Cref{thm:linear_MPC_spanner_weighted}. This difference will result in a different constant in the approximation factor of the final result, but has no impact otherwise.

\subsection{A Constant Approximate Unweighted Shortest Paths in Near-Linear MPC}\label{sec:applications_heterogeneous}

In this section, we consider an unweighted graph $G$, and we give a deterministic algorithm for a constant approximation for shortest paths in the near-linear MPC model by derandomizing the algorithms of Dory and Matar~\cite{dory2024massively}.
In fact, the algorithms require few machines with near-linear local space while the additional machines can have sublinear local space.
This model is also called the \emph{heterogeneous MPC model}, introduced by~\cite{FischerHO25}. 
Although the space complexities of our algorithms are analyzed by assuming the heterogeneous MPC model, the algorithms can be simulated in the near-linear MPC model, of which every machine has $\tilde O(n)$ local space, with the asymptotically same total space by \Cref{obs:parallelism}.

Dory and Matar~\cite{dory2024massively} presented a $\poly(\log\log n)$-round randomized algorithm for the Single-Source Shortest Paths (SSSP) problem for unweighted and undirected graphs.
Precisely, they built a data structure in the sublinear MPC model in $\textsf{poly}(\log\log n)$ rounds with high probability, then they store the data structure in a single near-linear spaced machine so that it can answer any query of SSSP problem in $O(1)$ rounds. 
They generalized the SSSP data structure for the APSP problem by additional modification.
In the other direction, they also generalized the SSSP algorithm for the Multi-Source Shortest Paths (MSSP) problem by allowing multiple machines with near-linear local space.
In this section, we derandomize their MSSP and APSP algorithms, and ~\Cref{thm:MSSP_sublinearMPC,thm:APSP_sublinearMPC} summarize the performances.
Note that when we set $S$ as a singleton $\{s\}$ in~\Cref{thm:MSSP_sublinearMPC}, we can get the algorithm for the SSSP problem using a single near-linear spaced machine.
\begin{restatable}{theorem}{MSSPinSublinearMPC}\label{thm:MSSP_sublinearMPC}
Let $G=(V,E)$ be an unweighted graph on $n$ vertices and $m$ edges, let $\delta>0$ be a constant, and let $\varepsilon<1/2$, $\rho\in[1/\log\log n,1/2]$ be parameters.
For a fixed source set $S\subseteq V$ of size $O(n^{\rho})$,
there is a deterministic algorithm that computes all $(1+\varepsilon)$ approximate shortest paths for every pairs $(s,v)\in S\times V$
in $ O\left(\left(\frac{\log\log n(\log\log\log n-\log \varepsilon)}{\varepsilon\rho}\right)^{2+1/\rho}\right)$ rounds
in the MPC model using $|S|$ machines with $\tilde O(n)$ space and additional machines with
$O(n^{\delta})$ local space and $\tilde O((m+n^{1+\rho})n^{\rho})$ total space.
\end{restatable}

\begin{restatable}{theorem}{APSPinSublinearMPC}\label{thm:APSP_sublinearMPC}
Let $G=(V,E)$ be an unweighted graph on $n$ vertices and $m$ edges, let $\delta>0$ be a constant, and let $\varepsilon<1/2$, $\rho\in[1/\log\log n,1/2]$ and $2\leq k\leq 1/\rho$ be parameters.
There is a deterministic $ O\left(\left(\frac{\log\log n(\log\log\log n-\log \varepsilon)}{\varepsilon\rho}\right)^{2+1/\rho}\right)$ rounds algorithm that computes a distance oracle of size $\tilde O(kn^{1+1/k})$ 
in the MPC model using a single machine with $\tilde O(n)$ space and additional machines with
$O(n^{\delta})$ local space and $\tilde O((m+n^{1+\rho})n^{1/k})$ total space.
Our oracle supports $O(1)$ round algorithm for $(1+\varepsilon)(2k-1)$-approximation distance query between any two vertices in $G$. 
\end{restatable}
\textbf{Optimizing the total space.}
{
Recall that the algorithms of \Cref{thm:MSSP_sublinearMPC} (and \Cref{thm:APSP_sublinearMPC}) give near-linear MPC algorithms with local space $\tilde O(n)$ and total space $\tilde O((m+n^{1+\rho})n^{\rho})$ (and $\tilde O((m+n^{1+\rho})n^{1/k})$) with asymptotically same round complexity by~\Cref{obs:parallelism}.
As shown in \cite{dory2024massively}, it is possible to improve the total space complexity at the cost of increasing the approximation factor by a constant factor by running the algorithms on a spanner. Specifically, they adapt the constant spanner (randomized) algorithm as a preprocessing step which reduces the number of edges, and then run the algorithms on the spanner.
Recall that we have deterministic algorithms to compute sparse spanners (\Cref{sec:applications_spanners}), and thus, we can also reduce the total space requirements to $O(m+n^{1+c})$ for some constant $c\leq 1$ to obtain an $O(1/c)$-approximate MSSP (or an $O(1/c^2)$-approximate APSP algorithm) in the near-linear MPC model within $\textsf{poly}(\log\log n)$ rounds.
This space complexity is $O(m)$ if $m=\Omega(n^{1+c})$.

First, we construct an $O(1/c)$-spanner $G'$ of size $O(n^{1+{c/3}})$ for the input graph $G$ in $O(1)$ rounds by~\Cref{thm:linear_MPC_spanner_weighted}, that is an $O(1)$ rounds deterministic algorithm to compute such a spanner in the linear MPC model.
Then we run ~\Cref{thm:MSSP_sublinearMPC_simplified} 
on $G'$ choosing $\rho=c/3$ and a fixed constant $\varepsilon$.
Then we get an $O(1)$-approximate MSSP algorithm on $G'$, that is $O(1/c)$-approximate on $G$, using $O(m+n)+\tilde O((n^{1+c/3}+n^{1+c/3})n^{c/3})=O(m+n^{1+c})$ total space in $\textsf{poly}(\log\log n)$ rounds in the near-linear MPC model.
Analogously, by running ~\Cref{thm:APSP_sublinearMPCSimplified} on $G'$ with $\rho=c/3, k=2/c$, and a fixed constant $\varepsilon$, we can get $O(1/c^2)$-approximate APSP data structure on $G$ within $\textsf{poly}(\log\log n)$ near-linear MPC rounds.

Note that even in the heterogeneous MPC model with the same number of near-linear local space machines in~\Cref{thm:MSSP_sublinearMPC} and~\Cref{thm:APSP_sublinearMPC}, we can reduce the total space by using the spanner algorithm of~\Cref{thm:sublinear_MPC_spanner}, which computes an $O(1/c^2)$-spanner $G''$ of size $O(n^{1+c/3})$ within $O(\log (1/c))$ sublinear MPC rounds.
Since $c$ is a constant, this process takes constant rounds.
Then by running the algorithms of~\Cref{thm:MSSP_sublinearMPC} (and~\Cref{thm:APSP_sublinearMPC}) on $G''$ with $\rho=c/3$ (and $k=2/c$ and a fixed constant $\varepsilon$), we can get $O(1/c^2)$-approximate MSSP algorithm (and $O(1/c^3)$-approximate APSP data structure) on $G$ within $\textsf{poly}(\log\log n)$ heterogeneous MPC rounds using $O(m+n^{1+c})$ total space.
}


\medskip

The algorithms of~\cite{dory2024massively} are based on a randomized framework inspired by Thorup-Zwick approach. They demonstrate that the framework can be used to effectively generate the data structures for the approximate distance problems.
The randomness occurs for the framework only.
Therefore, in this section, we focus on following their framework deterministically, and we omit the details on how to apply it.


\paragraph*{Framework of~\cite{dory2024massively}.}
For an unweighted graph $G=(V,E)$, we first set \emph{good} probabilities $\langle p_1,\ldots,p_{\bar \ell}\rangle$ for {${\bar \ell}=O(\log\log n)$
}.\footnote{Setting \emph{good} parameters ensures that the output gives a useful data structure for distance approximation. Computation of the parameters is described in~\cite{dory2024massively}, and we can use the same parameters.}
The framework begins by hierarchically sampling the vertex sets $V=A_0\supseteq A_1\supseteq A_{\bar \ell}$ with respect to the probabilities $\langle p_1,\ldots, p_{\bar \ell}\rangle$, where the vertices $A_i$ are sampled from $A_{i-1}$ with probability $p_i$.
After the hierarchical sampling of $A_0=V\supseteq A_1\supseteq\ldots\supseteq A_{\bar \ell}$, their algorithm deterministically runs ${\bar \ell}$ iterations, of which the $i$th iteration takes $O(h_i)$ rounds.
The $i$th iteration is given a weighted edge set $E_i$, a hop parameter $h_i$, and a distance threshold $x_i$ as input.
The goal of the iteration is to compute a non-negative weighted edge set $Q_i$ as follows.
In the following, we use $\textsf{dist}_{G\cup E_i}^{(h)}(u,v)$ to refer to the length of the shortest path between $u$ and $v$ in $G\cup E_i$ using at most $h$ edges.
For each vertex $v\in A_{i-1}$ (or $v\in V(G)$),\footnote{{Depending on the stage and the application in which this framework is invoked, we need to consider both cases $v\in A_{i-1}$ and $v\in V(G)$.}}
if $v$ has a vertex $s\in A_i$ with $\textsf{dist}_{G\cup E_i}^{(h_i)}(s,v)\leq x_i$, then 
it inserts an edge $\{s,v\}$ to $Q_i$ weighted by $\textsf{dist}_{G\cup E_i}^{(h_i)}(s,v)$.
Otherwise, it computes all vertices $u$ of $A_{i-1}$ so that $\textsf{dist}_{G\cup E_i}^{(h_i)}(v,u)\leq x_i/2$. Then it inserts edges $\{v,u\}$ to $Q_i$ weighted by $\textsf{dist}_{G\cup E_i}^{(h_i)}(v,u)$.
After ${\bar \ell}$ iterations, the algorithm returns the union of edge sets $\bigcup_i Q_i$.

During the algorithm, the following conditions hold for each $i\in [{\bar \ell}]$ with high probability~\cite[Lemmas 3.5, 3.10, and 3.11]{dory2024massively}:
\begin{enumerate}
    \item $|A_i|\leq O(|A_{i-1}|p_i)$ and
    \item For each $v\in V(G)$, there is either a sampled vertex $u\in A_i$ with $\textsf{dist}_{G\cup E_i}^{(h_i)}(v,u)\leq x_i$, or at most $O(\log n/p_i)$ non-sampled vertices $u'\in A_{i-1}\setminus A_i$ such that $\textsf{dist}_{G\cup E_i}^{(h_i)}(v,u')\leq x_i/2$.
\end{enumerate}
The performance of their algorithms is based on these properties.

In the following, we define and use the \emph{hop dominating set} as follows to connect our deterministic hitting set algorithm and the random framework. 
Note that when $A_i$ is a \emph{$h_i$-hop $\lfloor \log n/p_i\rfloor$-dominating set} of $A_{i-1}$ in $G\cup E_i$, then Condition (2) holds:  for a vertex $v\in V(G)$, if there are at least $\lfloor \log n/p_i\rfloor$ vertices of $A_{i-1}$ at distance at most $x_i/2$ from $v$ with respect to $\textsf{dist}_{G\cup E_i}^{(h_i)}(v,\cdot )$,  then at least one of them is selected to $A_i$ by the definition.


\begin{definition}

    We refer to a vertex set $A'\subseteq A$ as a \textbf{$h$-hop $d$-dominating set} of $A\subseteq V_H$ in an edge-weighted graph $H=(V_H,E_H)$ if either $|S_v|<d$ or $S_v\cap A'\neq \emptyset$ for every $v\in V_H$, where $S_v$ is the $d$ closest vertices in $A$ from $v$ with respect to the distance $\textsf{dist}_{H}^{(h)}(v,\cdot)$ (or all such vertices in $A$ if fewer than $d$ exist).

\end{definition}
In the following, we compute the hop dominating set by the hitting set problem to derandomize the framework.


\paragraph*{Overview of the derandomized algorithms.}
Instead of the \emph{hierarchical implicit sampling process}, 
our deterministic algorithm starts each $i$th iteration for $i\in [{\bar \ell}]$ by explicitly defining the hitting set problem for a $h_i$-hop $\lfloor \log n/p_i\rfloor$-dominating set $A_i$ of $A_{i-1}$ in $G\cup E_i$ of size $O(|A_{i-1}|p_i)$. Here, we assume that $A_{i-1}$ is given by the previous $(i-1)$th iteration.
Then we solve the hitting set problem ~\Cref{thm:optimal_linear MPC hitting set,thm:optimal_sublinear_MPC_hitting set}, \kyungjin{correct reference here and below}
and we run the remaining process for the $i$th iteration of \cite{dory2024massively} which is deterministic.
Our deterministic algorithm runs in $O(h_i+\log\log n)$ rounds per $i$th.

\medskip

Dory and Matar described how to set the good probabilities $p_i$'s used in the sampling process and hop parameters $h_i$'s, distance thresholds $x_i$'s, and the input edges in $E_i$'s used in each iteration so that the framework gives subroutines for SSSP, MSSP, and APSP problems~\cite[Theorem 4.6 and Theorem 5.2]{dory2024massively}.
Furthermore, for the applications, the hop parameters $h_i$'s  are set at least $\Omega(\log^{1+\rho}\log n)$ for some $\rho>0$. Therefore, the round complexity $O(h_i+\log\log n)=O(h_i)$ of our deterministic iteration does not affect the total round complexity.

\paragraph*{Formulating the hitting set problem for $A_i$.}
Here, we deterministically compute a $h_i$-hop $\lfloor \log n/p_i\rfloor$-dominating set $A_i$ of $A_{i-1}$ in $G\cup E_i$ by solving the hitting set problem on the universe $U:=A_{i-1}$ and the sets $S_v$'s of size $d:=\lfloor \log n/p_i\rfloor$ defined as follows. 
To define the sets, we let $V_i$ be the set of vertices $v\in V(G)$ that is reachable from at least $\lfloor \log n/p_i\rfloor$ vertices of $A_{i-1}$ using at most $h_i$ edges.
Then we define $N:=|V_i|$ sets $S_v$ for $v\in V_i$:
\begin{equation*}
    S_v:=\{\text{the } \lfloor \log n/p_i\rfloor \text{ closest vertices of } A_{i-1} \text{ from }v\text{ with respect to }\textsf{dist}_{G\cup E_i}^{(h_i)}(v,\cdot)\}.
\end{equation*}



Our goal is to compute the set $V_i\subseteq V(G)$ and $S_v$'s.
The computation is implemented by $O(h_i)$ rounds in the sublinear MPC model with $O(n^{\delta})$ local space and $\tilde O((|E|+|E_i|+n)/p_i)$ total space by  \Cref{lem:restricted_BF}.
Precisely, we recursively apply \Cref{lem:restricted_BF} $O(h_i)$ times by setting $H:=G\cup E_i$, the source set $S:=A_{i-1}$, and $d:=\lfloor \log n/p_i\rfloor$. 
It is analogous to the multi-sources restricted Bellman-Ford algorithm of~\cite[Lemma 3.7]{dory2024massively}. Detailed implementation is in~\Cref{ap:multisourceBF}.

\begin{restatable}{lemma}{LemMultisourceBF}\label{lem:restricted_BF}
    Let $H=(V_H,E_H)$ be an edge-weighted graph, let $S\subseteq V_H$, and let $h,d\ge 1$. 
    Suppose that for every vertex $v\in V_H$, we are given the $d$ closest sources in $S$ from $v$ with respect to the distance $\textsf{dist}_{H}^{(h)}(v,\cdot)$ (or all such sources if fewer than $d$ exist).

    Then we can compute, for every vertex $v$, the $d$ closest sources in $S$ with respect to the distance $\textsf{dist}_H^{(h+1)}(v,\cdot)$ (again truncated at $d$) in $O(1/\delta)$ rounds in the sublinear MPC model using $O(|V_H|^{\delta})$ local memory per machine and $\tilde O(|E_H|\,d)$ total memory.
\end{restatable}

\paragraph*{Computing $A_i$ of size $O(p_i|A_{i-1}|)$.}
Our goal is to compute a hitting set $A_{i}$ of the universe set $U:=A_{i-1}$ over the $N:=|V_i|$ sets $S_v$'s of size at least $d:=\lfloor \log n/p_i\rfloor$ defined by $v\in V_i$.
We solve the hitting set problem in $O(\log\log n)$ rounds by~\Cref{thm:optimal_linear MPC hitting set,thm:optimal_sublinear_MPC_hitting set} within the sublinear MPC model with $L=O(n^{\delta})$ local space.
Precisely, if $d\leq L:=O(n^\delta)$ (or $d>L$), then \Cref{thm:optimal_linear MPC hitting set} (or \Cref{,thm:optimal_sublinear_MPC_hitting set}) deterministically computes a set $A_i$ so that:
\begin{itemize}
    \item  $|A_i|\leq O(\frac{p_i}{\log n}|A_{i-1}| \log n)=O(p_i|A_{i-1}|)$, and
    \item $A_i\cap S_v\neq \emptyset$ for every $v\in V_{i-1}$.
\end{itemize}
Note that the algorithm holds when $N=\textsf{poly}(d, L)$ (and $N<d^{L^{1/4}}=d^{n^{\delta/4}}$ when  $d>L$), that are hold for sufficiently large $n$ as follows $|V_i|=\textsf{poly}(\log n, n^{\delta})$ and $|V_i|<n^{\delta n^{\delta/4}}$ in this problem. 
Furthermore, the round complexity is at most $O(\log\log (d+N))=O(\log\log n)$.
In other words, we can deterministically compute $A_i\subseteq A_{i-1}$ satisfying the two invariants in $O(\log\log n)$ rounds.

\paragraph*{Computing $Q_i$ edge set with respect to $A_i\subseteq A_{i-1}$.}
During the $i$th iteration, we are given the vertex set $A_{i-1}$ and the weighted edge set $E_i$ along with the parameters $h_i$ and $x_i$.
We have computed $A_i$ with respect to the $A_{i-1}$.
Recall that the randomized algorithm computes the edges as follows for each $v\in A_{i-1}$ (or $v\in V(G)$):
\begin{enumerate}
    \item If a vertex $u\in A_i$ exists with $\textsf{dist}_{G\cup E_i}^{(h_i)}(v,u)\leq x_i$, then add an edge $\{u,v\}$ weighted by $\textsf{dist}_{G\cup E_i}^{(h_i)}(v,u)$ to $Q_i$.
    \item Otherwise, for every vertex $u\in A_{i-1}\setminus \{v\}$ with $\textsf{dist}_{G\cup E_i}^{(h_i)}(v,u)\leq x_i/2$, add an edge $\{u,v\}$ weighted by $\textsf{dist}_{G\cup E_i}^{(h_i)}(v,u)$ to $Q_i$.
\end{enumerate}
{The previous randomized algorithm also gives the deterministic algorithm to compute the edges, we also give the details for completeness.}

The first case edges can be computed in $O(h_i)$ rounds in the sublinear MPC model by using the single-source Bellman-Ford algorithm~\cite{BellareR94} by adding a dummy source $s$ connected to each vertex in $A_{i}$ by zero-weighted edges.
Then, for the second case edges, we use the computed sets $S_v$'s of the $\lfloor \log n/p_i\rfloor$ closest vertices of $A_{i-1}$ in $G\cup E_i$ from $v$.
Recall that in the previous phase, we computed the distance $\textsf{dist}_{G\cup E_i}^{(h_i)}(v,u)$ for pairs $v\in V$ and $u\in S_v$, and we selected the vertices $A_i$ so that $S_v\cap A_i\neq \emptyset$ for any $v\in V$.
Therefore, for the second case, it suffices to add the edges $\{u,v\}$ weighted by $\textsf{dist}_{G\cup E_i}^{(h_i)}(v,u)$ to $Q_i$ for each vertex $u\in S_v$ if $\textsf{dist}_{G\cup E_i}^{(h_i)}(v,u)\leq x_i/2$ and $\textsf{dist}_{G\cup E_i}^{(h_i)}(v,u')> x_i$ for every $u'\in S_v\cap A_{i}$.
In conclusion, we can compute the weighted edge set $Q_i$ in $O(h_i)$ rounds.

\subsubsection{Multi-source restricted Bellman-Ford: Proof of \texorpdfstring{\Cref{lem:restricted_BF}}{Lemma 6.6}}
\label{ap:multisourceBF}
In this section, we prove \Cref{lem:restricted_BF}, restated here for convenience. 

\LemMultisourceBF*
\begin{proof}
Recall that we need the subroutine for two purposes: 
\textsf{(1)} to explicitly formulate the set systems $S_v$'s of the hitting set problem for $A_i$ with respect to $A_{i-1}$, and 
\textsf{(2)} to compute the weighted edge set $Q_i$ with respect to $A_i$ and the sets $S_v$'s.
Although the previous randomized algorithm of~\cite{dory2024massively} does not explicitly formulate the set system, it also requires a similar subroutine to deterministically compute the weighted edges $Q_i$. 
Their subroutine~\cite[Lemmas 3.6-3.7]{dory2024massively} is designed to compute $\textsf{dist}_H^{(h+1)}(v,s)$ for every $v \in V_H$ and $s \in S$ if $\textsf{dist}_H^{(h+1)}(v,s) \leq x$ for a given distance threshold $x$, under the assumption that the number of such sources is at most $d$ for each fixed vertex $v \in V_H$.
By slightly modifying their deterministic algorithm, we obtain the subroutine of ~\Cref{lem:restricted_BF}. 
Briefly, both our algorithm and the previous one start by updating the distance estimates
$\textsf{dist}_H^{(h)}(u,s) + w(u,v)$ for each edge $\{u,v\} \in E_H$ and given $\textsf{dist}_H^{(h)}(u,s)$'s for $s\in S$ and $u\in V_H$. 
After this step, the previous algorithm first discards the distance estimates exceeding the threshold $x$, and then rearranges the remaining ones to maintain the desired values $\textsf{dist}_H^{(h+1)}(v,s)$'s. 
In contrast, our algorithm first rearranges all updated distance estimates to maintain the values $\textsf{dist}_H^{(h+1)}(v,s)$, and later prunes (at most) $d$ closest sources for each vertex $v \in V$ with respect to the updated distances.
In the following, we give the implementation in the sublinear MPC model with $\tilde O(|E_H|d)$ total space for completeness. 

    For $v\in V_H$, we let $d_v$ (and $d_v'$) be the number of reachable sources in $S$ from $v$ using at most $h$ (and $(h+1)$) edges in $H$ if the number of such sources is fewer than $d$, otherwise, $d_v=d$ (and $d_v'=d$). 
    Let $S_v$ be the set of $d_v$ closest sources in $S$ from $v$ in $H$ using at most $h$ edges.
    We claim to compute the $d_v'$ closest sources to $v$ among $S_v\cup (\bigcup_{\{u,v\}\in E_H}S_u)$ using at most $h+1$ edges for all vertices $v\in V_H$ in $O(1/\delta)$ rounds.
    We prove that it is sufficient to consider the sources $S_v\cup (\bigcup_{\{u,v\}\in E_H} S_u)$ instead of the whole source set $S$ in the analysis phase. 
    
    The algorithm has three phases: Redistributing the input distances $\textsf{dist}_H^{(h)}(v,s)$ between $s\in S_v$ and $v$, computing the distance $\textsf{dist}_H^{(h+1)}(v,s)$ for $s\in S_v\cup (\bigcup_{\{v,u\}\in E_H} S_u)$, and selecting $d_v'$ closest sources in $S_v\cup (\bigcup_{\{v,u\}\in E_H} S_u)$ with respect to the computed $\textsf{dist}_H^{(h+1)}$.

    \textbf{Redistributing the input distances.}
    In $O(1/\delta)$ rounds, we can redistribute the initially given $O(\sum_{v\in V_H}d_v)=O(|V_H|d)$ distances $\textsf{dist}_H^{(h)}(v,s)$'s with $s\in S_v$ by replacing the distance information as a tuple $(v, \textsf{dist}_H^{(h)}(v,s), s)$ and sorting them along the lexicographical ordering so that:
    The tuples are stored on consecutive machines, and especially for each vertex $v$, its $|S_v|$ tuples are stored on consecutive machines by applying~\Cref{lm:preprocessing}.
    Then, in constant rounds, we can replace each tuple $(v, \textsf{dist}_H^{(h)}(v,s_j), s_j)$ as $(v,j, \textsf{dist}_H^{(h)}(v,s_j), s_j)$, where $s_j$ is the $j$th closest source in $S_v$ to $v$ in $H$.
    {To achieve this, we apply the broadcasting algorithm from \Cref{lem:sorting_prefixsum_MPC}, which ensures that, for each fixed vertex $v$, the ID of the first machine containing tuples whose first element is $v$ becomes known to all relevant machines.
    Moreover, by \Cref{obs:parallelism}, this broadcasting can be performed simultaneously for all vertices $v\in V_H$. Consequently, each machine can locally determine the correct order of its tuples.
    }

    We start by duplicating each edge $\{v,u\}\in E_H$ weighted by $w(u,v)$ into $d$ tuples $(u,j,w(u,v),v)$ for $j\in [d]$ and $d$ tuples $(v,j', w(v,u),u)$ for $j'\in[d]$. 
    We can do the duplicating process within $O(1/\delta)$ rounds \cite[Lemma 3.6]{dory2024massively}.
    Then resort $O(|V_H|d)$ number of distance tuples $(v,j, \textsf{dist}_H^{(h)}(v,s_j), s_j)$'s and $O(|E_H|d)$ number of duplicated edge tuples $(v,j',w(v,u),u)$'s along the lexicographical ordering by \Cref{lem:sorting_prefixsum_MPC} so that two tuples $(v,j, w(u,v),u)$ and $(v, \textsf{dist}_H^{(h)}(v,s_j), s_j)$ are stored in the single machine, denoted by $\mathcal M_{(u,v,j)}$.
    Note that in~\cite[Section 3.2]{dory2024massively}, there is a detailed implementation so that $\mathcal M_{(u,v,j)}$ can detect the id of $\mathcal M_{(v,u,j)}$ by its local computation, but we omit this part here.

    \textbf{Computing the distance estimates using $h+1$ edges.}
    We let $\mathcal M(v)$ be the consecutive machine storing the tuples whose first element is the vertex $v$.
    For each $j\in [d]$ and $\{u,v\}\in E_H$, the machine $\mathcal M_{(u,v,j)}$ computes the distance $\textsf{dist}_H^{(h)}(v,s_j)+w(u,v)$ that is the distance estimate from $s_j$ to $u$ through $v$, by using at most $h+1$ edges.
    Then it sends the distance to $\textsf{dist}_H^{(h)}(v,s_j)+w(u,v)$ to the machine $\mathcal M_{(v,u,j)}\in \mathcal M(u)$ which is storing the $j$th tuple $(u,j,w(v,u),v)$.
    
    The machines in $\mathcal M(u)$ have $d_u$ distances $\textsf{dist}_H^{(h)}(u,s_i)$'s for $s_i\in S_u$ received as inputs and at most $O(d\cdot \deg_H(u))$ distance estimates $(\textsf{dist}_H^{(h)}(v,s_j)+w(u,v))$'s for $s_j\in \bigcup_{\{u,v\}\in E_H} S_v$ received by the machines $\mathcal M_{(u,v,j)}$'s.
    By~\Cref{lem:sorting_prefixsum_MPC}, 
    we sort the distances and distance estimates within the machines $\mathcal M(u)$ in $O(1/\delta)$ rounds with respect to the starting source $s$. 
    Then in $O(1/\delta)$ rounds, we can compute the minimum value over all distances $\textsf{dist}_H^{(h)}(u,s)$ and distance estimates $\textsf{dist}_H^{(h)}(v,s)+w(u,v)$ within the machines $\mathcal M(u)$ for each $s\in S_u\cup (\bigcup_{\{u,v\}\in E_H} S_v)$, and we define $\textsf{dist}_H^{(h+1)}(u,s)$ as the minimum value.
    
    \textbf{Selecting at most $d$ closest sources with respect to $\textsf{dist}_H^{(h+1)}$.}
    For each $u\in V_H$, we can prune the $d_u'$ closest sources $s$ among $S_u\cup (\bigcup_{\{u,v\}\in E_H} S_v)$ and the distances $\textsf{dist}_H^{(h+1)}(u,s)$ as we desired by applying~\Cref{lm:preprocessing}.   
    Recall that $d_u'$ is the number of reachable sources in $S$ from $u$ using at most $(h+1)$ edges in $H$ if the number of such sources is fewer than $d$, otherwise, $d_u'=d$. 

    \textbf{Analysis.}
    Note that the three phases used the constant number of calls for the subroutines of \Cref{lem:sorting_prefixsum_MPC}, \Cref{lm:preprocessing}, and the duplicating process~\cite[Lemma 3.6]{dory2024massively} that are implemented by $O(1/\delta)$ rounds in the sublinear MPC model with at most $\tilde O(|E_H|d)$ total space.
    Totally, our algorithm achieves the desired space complexities and the round complexity.
    In the following, we prove the correctness of the algorithm.
    Note that the correctness of the algorithm is guaranteed by the claim: 
    For every vertex $u\in V_H$, $S_u'$ is subset of $S_u\cup (\bigcup_{\{u,v\}\in E_H} S_v)$.
    Recall that $S_u'$ (and $S_v$'s) is the $d_u'$ closest sources in $S$ to $u$ (and $v$) with respect to $\textsf{dist}_{H}^{(h+1)}(u,\cdot)$ (and $\textsf{dist}_{H}^{(h)}(v,\cdot)$).
    
    Note that if $d_u'<d$, then there are at most $d_u'$ sources of $S$ reachable to $u$ using at most $h+1$ edges. By the definition, these sources are reachable to $u$ or its neighbor using at most $h$ edges, which implies that $S_u'\subseteq S_u\cup(\bigcup_{\{u,v\}\in E_H}S_v)$. Therefore, the claim holds.
    In the following, we consider the case that $d_u'=d$.
    For the contradiction, we assume that there is a source $s$ not in $S_u\cup(\bigcup_{\{u,v\}\in E_H}S_v)$ but it is in $S_u'$.
    We consider a shortest path $\pi$ from $s$ to $u$ using at most $h+1$ edges, we let $\{u,v\}$ be the last edge of $\pi$.
    Then the prefix subpath $\pi'$ of $\pi$ obtained by removing $\{u,v\}$ is a shortest path between $s$ and $v$ using at most $h$ edges.
    Since $s$ is not in $S_v$ while it is a source in $S$ reachable to $v$ using at most $h$ edges, there are $\pi_1',\ldots, \pi_{d}'$ paths from $d$ sources $s_1,\ldots, s_{d}$ in $S\setminus\{s\}$ to $v$ that use at most $h$ edges and shorter than the path $\pi'$.
    Then we can obtain the path $\pi_1,\ldots, \pi_{d}$ paths from $s_1,\ldots, s_{d}$ to $u$ by concatenating the paths $\pi_1,\ldots, \pi_{d}$ and $\{u,v\}$, respectively.
    The paths are shorter than $\pi$ and use at most $h+1$ edges in $H$.
    This implies that $s_1,\ldots, s_{d}$ are closer to $u$ than $s$ when we allow at most $h+1$ edges, which contradicts the assumption that $s$ is one of the $d$ closest sources from $u$ using at most $h+1$ edges. 
    This completes the proof of this lemma.
\end{proof}



\newpage
  \printbibliography[heading=bibintoc]

\newpage
\appendix
\section{Proof of the Conditional Expectation Lemma in MPC}
\LemCondExp*

\begin{proof}
Let $f(h)=\sum_{i=1}^{\FC} f_i(h)$.
We claim to compute a hash function $h^*$ in $\H$ with $f(h^*)\leq x$. The first condition,  $\E_{h\sim \mathcal H}[f(h)]\leq x$, guarantees that such a hash function exists, and the hash functions in $\mathcal H $ are encoded by $s$-bit random seeds.
The conditional expectation method determines the $s$-bit random seed of $h^*$.
This algorithm decomposes the random seeds into $t=O(s/\log \LS)$ chunks $R_1,\ldots, R_t$ of $\lfloor \log \LS\rfloor$ bits random seeds, and iteratively determines each chunk.
More precisely, at the $j$-th iteration, we assume that the chunks $R_1=\textsf{r}_1,\ldots, R_{j-1}=\textsf{r}_{j-1}$ are known and determine $\textsf{r}_j$ for the chunk $R_j$
so that $\E[f(h) \mid \textnormal{the prefix of } h \textnormal{ is }\textsf{r}_1\ldots\textsf{r}_{j}]$ is at most $\E[f(h) \mid \textnormal{the prefix of } h \textnormal{ is }\textsf{r}_1\ldots\textsf{r}_{j-1}]$.
After $t=O(s/\log \LS)$ iterations, we determine all chunks of the random seed of $h^*$ so that:
\[
f(h^*)=\E[f(h) \mid  \textnormal{the prefix of } h \textnormal{ is }\textsf{r}_1\ldots\textsf{r}_{t}]\leq \E[f(h)]\leq x.
\]

In the following, we describe how to implement each iteration in $O(\log \FC/\log\LS)$ rounds in the MPC model, which concludes that our algorithm runs in $O(s\log\FC/\log^2 \LS)$ rounds, and thus, the lemma holds.
We assume that the $j-1$ chunks $R_1,\ldots,R_{j}$ have been fixed, and we describe how to compute a proper $R_{j+1}$.
We consider all $O(\LS)$ possible assignments $\textsf r$ of $R_{j+1}$.
Individually, we assign a machine $M^i$ for each objective function $f_i$ with $i\in [\FC]$.
The algorithm has three steps.
Each step takes a constant round with $\LS$ local space and $O(\LS\FC)$ total space. 

First, for each index $i\in [\FC]$, the machine $M^i$ locally computes and stores all expectations $$\E[f_i(h) \mid  \textnormal{the prefix of } h \textnormal{ is }\textsf{r}_1\ldots\textsf{r}_{j}\textsf{r}] \textnormal{ for every possible assignments }\textsf r.$$
Hereto, we do the process for every $i\in [\FC]$ in the machines $M^1,\ldots, M^{\FC}$ simultaneously.

Next, we let $\E_\textsf r^i\coloneqq \E[f_i(h) \mid  \textnormal{the prefix of } h \textnormal{ is }\textsf{r}_1\ldots\textsf{r}_{j}\textsf{r}]$, and we compute the sum $\sum_{i=1}^{\FC} \E_\textsf r^i$ for each  assignment $\textsf r$. Note that by the linearity of the expectations, it matches the desired expectation value of $f$ as follows:
\begin{align*}
    \E\left[f(h) \mid  \textnormal{the prefix of } h \textnormal{ is }\textsf{r}_1\ldots\textsf{r}_{j}\textsf{r}\right] 
    &=\E\left[\sum_{i=1}^{\FC}f_i(h) \mid  \textnormal{the prefix of } h \textnormal{ is }\textsf{r}_1\ldots\textsf{r}_{j}\textsf{r}\right]\\
    &=\sum_{i=1}^{\FC}\E[f_i(h) \mid  \textnormal{the prefix of } h \textnormal{ is }\textsf{r}_1\ldots\textsf{r}_{j}\textsf{r}]=\sum_{i=1}^{\FC} \E_\textsf r^i.
\end{align*}
For each computed values $\E_\textsf r^i$ that is stored in $M^1,\ldots, M^{\FC}$, we replace it as a tuple $(i,\textsf r, \E_\textsf r^i)$ of the index $i\in [\FC]$, the assignment $\textsf r$, and the expectation value. 
Then we sort the tuples along the lexicographical ordering.
Then we compute $\sum_{i=1}^{\FC} E_\textsf{r}^i$ for all $\LS$ assignments $\textsf{r}$ simultaneously.
These sorting and summation take $O(\log \FC/\log\LS)$ rounds by~\Cref{lem:sorting_prefixsum_MPC}. 
Finally, we compute $R_{j+1}=\textsf r^*$ which minimizes the expectation value $\sum_{i=1}^{\FC} E_{\textsf{r}^*}^i$ over $\LS$ assignments. It takes a constant time by \Cref{lem:sorting_prefixsum_MPC}.

In conclusion, each iteration takes $O(\log \FC/\log\LS)$ rounds, and thus, we can compute the desired hash function $h^*$ in $O(s\log\FC/\log^2 L)$ rounds.
\end{proof}


\end{document}